\documentclass[12pt]{article}

\usepackage{amsmath}
\usepackage{amssymb}
\usepackage{amsthm}
\usepackage{array}
\usepackage[footnotesize,bf]{caption}
\usepackage{color}
\usepackage{enumitem}
\usepackage{eurosym}
\usepackage[top=2.4cm,left=2.7cm,right=2.7cm,bottom=3cm]{geometry}
\usepackage{graphicx}
\usepackage{lineno}
\usepackage[sc]{mathpazo}
\usepackage{mathtools}
\usepackage{multirow}
\usepackage[sort&compress,comma]{natbib}
\usepackage{setspace}
\usepackage{times}
\usepackage{titlesec}
\usepackage[normalem]{ulem}
\usepackage{verbatim}
\usepackage{xr}

\smallskip

\definecolor{mygreen}{rgb}{0.0, 0.6, 0.0}

\titleformat{\section}{\sffamily \fontsize{13}{16}\bfseries}{\thesection}{1em}{}
\titleformat{\subsection}{\sffamily \fontsize{11.5}{11.5}\bfseries}{\thesubsection}{1em}{}

\newcommand*\patchAmsMathEnvironmentForLineno[1]{%
	\expandafter\let\csname old#1\expandafter\endcsname\csname #1\endcsname
	\expandafter\let\csname oldend#1\expandafter\endcsname\csname end#1\endcsname
	\renewenvironment{#1}%
	{\linenomath\csname old#1\endcsname}%
	{\csname oldend#1\endcsname\endlinenomath}}%
\newcommand*\patchBothAmsMathEnvironmentsForLineno[1]{%
	\patchAmsMathEnvironmentForLineno{#1}%
	\patchAmsMathEnvironmentForLineno{#1*}}%
\AtBeginDocument{%
	\patchBothAmsMathEnvironmentsForLineno{equation}%
	\patchBothAmsMathEnvironmentsForLineno{align}%
	\patchBothAmsMathEnvironmentsForLineno{flalign}%
	\patchBothAmsMathEnvironmentsForLineno{alignat}%
	\patchBothAmsMathEnvironmentsForLineno{gather}%
	\patchBothAmsMathEnvironmentsForLineno{multline}%
}

\newcommand{\delhat}{\widehat{\Delta}}
\newcommand{\delhatsel}{\widehat{\Delta}_{\mathrm{sel}}}

\newcommand{\vx}{\mathbf{x}}
\newcommand{\vy}{\mathbf{y}}

\newcommand{\vA}{\mathbf{A}}
\newcommand{\vB}{\mathbf{B}}
\newcommand{\MSS}{\mathrm{MSS}}
\newcommand{\RMC}{\mathrm{RMC}}

\theoremstyle{definition}
\newtheorem{corollary}{Corollary}

\newtheorem{example}{Example}
\newtheorem{lemma}{Lemma}

\newtheorem{remark}{Remark}
\newtheorem{theorem}{Theorem}
\newtheorem*{fixation}{Fixation Axiom}

\newcommand{\eq}[1]{Equation~\ref{eq:#1}}
\newcommand{\fig}[1]{Figure~\ref{fig:#1}}
\newcommand{\sfig}[1]{Extended Data Figure~\ref{fig:#1}}
\newcommand{\sifig}[1]{Supplementary~Figure~\ref{fig:#1}}

\newcommand{\lem}[1]{Lemma~\ref{lem:#1}}

\title{\begin{center} \bfseries \singlespacing
Social goods dilemmas in heterogeneous societies
\end{center}}
\author{\parbox[c]{16cm}{\onehalfspacing \normalsize \centering ~\\[-0.4cm] Alex McAvoy$^{1,}$\footnote{Corresponding authors: Alex McAvoy (\texttt{alexmcavoy@g.harvard.edu}) and Martin A. Nowak (\texttt{martin\_nowak@harvard.edu}).} \quad Benjamin Allen$^{2}$ \quad Martin A. Nowak$^{1,3,\ast}$\\ \quad\\ \footnotesize
$^{1}$Department of Organismic and Evolutionary Biology, Harvard University, Cambridge, MA~02138 USA; \\
$^{2}$Department of Mathematics, Emmanuel College, Boston, MA~02115 USA; \\
$^{3}$Department of Mathematics, Harvard University, Cambridge, MA~02138 USA \\[0.2cm]}
\date{}
}

\begin{document}

\maketitle

\begin{abstract}
Prosocial behaviors are encountered in the donation game, the prisoner's dilemma, relaxed social dilemmas, and public goods games. Many studies assume that the population structure is homogeneous, meaning all individuals have the same number of interaction partners, or that the social good is of one particular type. Here, we explore general evolutionary dynamics for arbitrary spatial structures and social goods. We find that heterogeneous networks, wherein some individuals have many more interaction partners than others, can enhance the evolution of prosocial behaviors. However, they often accumulate most of the benefits in the hands of a few highly-connected individuals, while many others receive low or negative payoff. Surprisingly, selection can favor producers of social goods even if the total costs exceed the total benefits. In summary, heterogeneous structures have the ability to strongly promote the emergence of prosocial behaviors, but they also create the possibility of generating large inequality.
\end{abstract}

\section{Introduction}
Prosocial behaviors are often studied using two-player or many-player games. In the first case, we encounter the donation game \citep{sigmund:PUP:2010,radzvilavicius:eLife:2019}, prisoner’s dilemma \citep{axelrod:BB:1984,szabo:PRE:1998,abramson:PRE:2001,broom:TF:2013}, or relaxed social dilemmas \citep{maynardsmith:CUP:1982,hauert:Nature:2004,doebeli:S:2004}. In the second case, we are typically in the world of public goods games \citep{lloyd:OUP:1833,hardin:Science:1968,szabo:PRL:2002,pinheiro:PLOSCB:2014,pena:PLOSCB:2016,zhong:CSF:2017}. In both kinds of games, it is usually assumed that players are in identical positions and affect all others equally. This homogeneity can be a consequence of the spatial structure of the population; for example, all individuals might have the same number of neighbors. However, even within spatially-heterogeneous populations, it is often assumed that every group (or pair) plays the same game. In this study, we consider ``social goods dilemmas'' in which individuals may pay a cost to produce a good that benefits their neighbors. In social goods dilemmas, the distribution of benefits and costs can depend on the population structure as well as on the nature of the good itself. If some individuals are central and well-connected within a group, while others are peripheral, social goods dilemmas lead to heterogeneous game structures with surprising evolutionary dynamics.

For social goods produced within an interaction structure, two questions become immediately apparent: \emph{(i)} is the benefit of receiving the social good from a specific donor independent of the number of recipients (non-rival), or does one neighbor's access to the good decrease that of another (rival)? and \emph{(ii)} is the cost of producing a social good a function of the number of recipients, or is it fixed? In the traditional setting of homogeneous population structures \citep{nowak:Nature:1992,nakamaru:JTB:1997,lieberman:Nature:2005,ohtsuki:Nature:2006,taylor:Nature:2007,chen:AAP:2013,debarre:NC:2014}, there is no reason to consider these cases separately since the differences between them amount to a simple rescaling of the benefits and/or costs. However, important distinctions among those social goods arise in heterogeneous societies \citep{santos:PRL:2005,antal:PRL:2006,gomez:PRL:2007,sood:PRE:2008,cao:PA:2010,maciejewski:PLoSCB:2014,fan:PA:2017,allen:Nature:2017}. In fact, distinguishing among various kinds of social goods is a common practice in economics \cite{goldfarb:JEL:2019}, one that has not fully permeated evolutionary game theory \citep{santos:Nature:2008,li:PLOSONE:2013}. The simplest dichotomy is between benefits that are proportional (``p'') to the number of recipients and those that are fixed (``f''). The same two options for the cost of a good gives four types of social goods, representing the combinations of benefits and costs: pp (proportional benefits, proportional costs), ff (fixed benefits, fixed costs), pf (proportional benefits, fixed costs), and fp (fixed benefits, proportional costs). Our primary focus here is on pp-, ff-, and pf-goods, which are summarized in \fig{ppff}.

\begin{figure}
	\centering
	\includegraphics[width=0.9\textwidth]{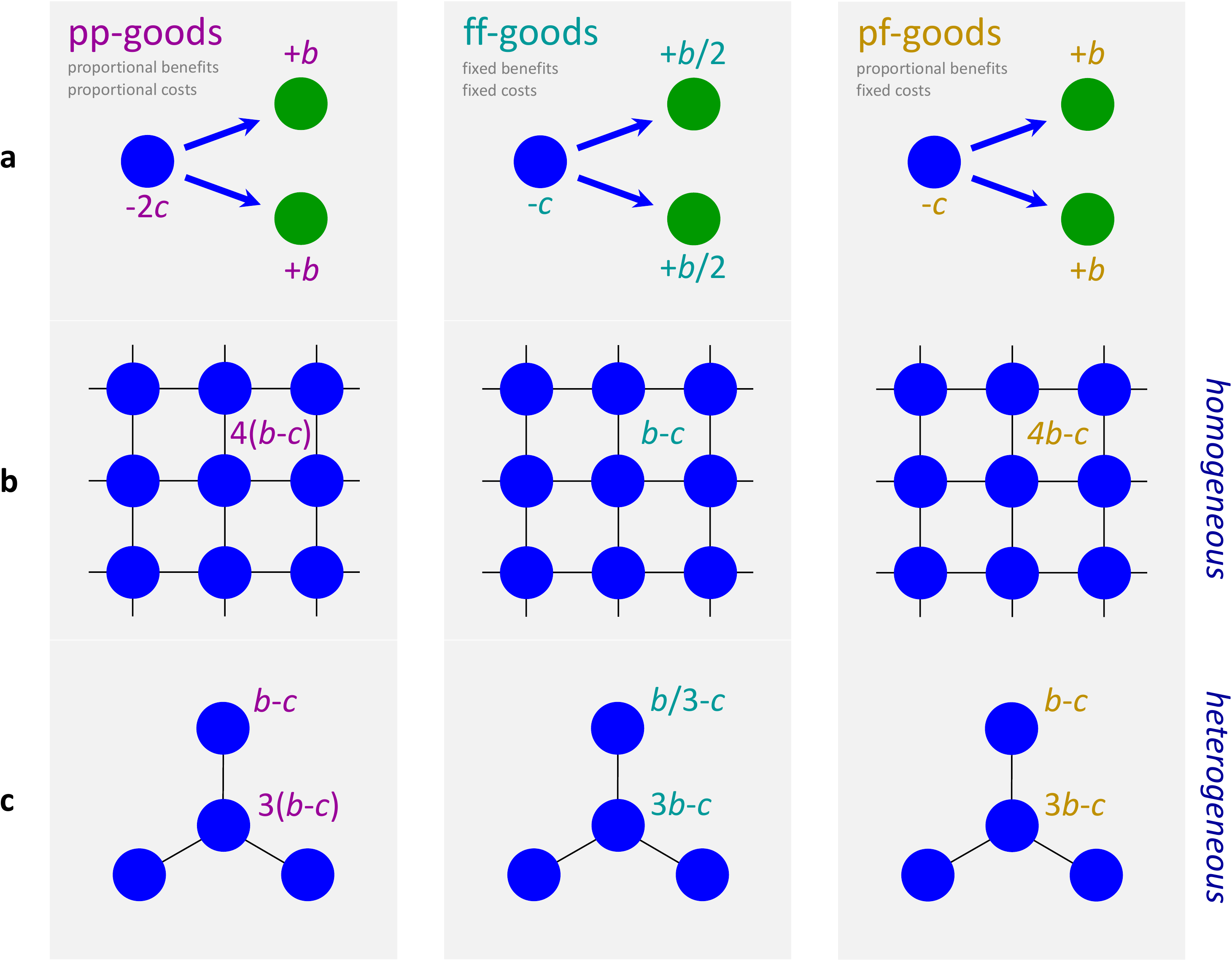}
	\caption{\textbf{Social goods and prosocial behaviors.} \textbf{a}, For pp-goods, a producer pays cost $c$ for each neighbor to receive benefit $b$. For ff-goods, a producer pays a fixed cost $c$ irrespective of the number of neighbors, $k$; each neighbor receives benefit $b/k$. For pf-goods, the cost, $c$, is again independent of the number of neighbors, but now each neighbor gets to enjoy the benefit, $b$, in its entirety. \textbf{b}, On a regular graph, such as a two-dimensional grid, all individuals have the same number of neighbors; here $k=4$. For pp-goods, each individual receives payoff $4(b-c)$. For ff-goods, each individual receives payoff $b-c$. Therefore, on regular graphs the payoffs arising for pp-goods and ff-goods are equivalent up to rescaling both $b$ and $c$. Scaling only $b$ (resp. $c$) gives an equivalence between pf-goods and ff-goods (resp. pf-goods and pp-goods). \textbf{c}, On heterogeneous population structures, such as the star, the three kinds of social goods lead to distinct payoff distributions, and one cannot be obtained from another by rescaling $b$ and/or $c$. Therefore, heterogeneous graphs highlight important differences among social goods.\label{fig:ppff}}
\end{figure}

As an example, consider the prosocial act of donating blood \citep{stutzer:TEJ:2011}. One recipient's use of blood decreases that of another, so blood is a rival good. Blood is also divisible, and a fixed volume of it can be distributed among several individuals in need. However, the nature of this donation as a social good depends not only on the good itself (blood) but also on how the behavior is expressed within the population. A donor might attempt to give each individual in need as much blood as possible, potentially incurring a huge cost for doing so. But they might also decide on a more modest, fixed donation, to be divided evenly among those in need (and possibly supplemented by donations from others). The former case is modeled better as a pp-good, while the latter could be viewed as an ff-good. Similar arguments can be made for other kinds of social behaviors in human communities, such as helping out coworkers, volunteering at a charity, and donating money. In nonhuman societies, relevant examples include social grooming among primates \citep{dunbar:FP:1991}, food delivery among magpies \citep{horn:BL:2016}, and blood donation (as food) among vampire bats \citep{wilkinson:Nature:1984}.

All neighbors might also be able to benefit from a good in its entirety, even if production of the good entails a fixed cost. Volunteering to maintain a public space, such as a park, is one such example. The cost for doing so can be quantified in terms of time, effort, or money (e.g. purchasing supplies or hiring a groundskeeper). Barring extenuating circumstances, the benefit of having a clean park is not necessarily reduced by another person's use of the space; this can therefore be seen as a pf-good. Such is also the case for information transmission within a social network, whose initial acquisition could entail a cost but whose full benefit can be enjoyed by more than just a single individual. Entertainment, such as podcasts, radio programs, and video streaming, can also be non-rival (and, in fact, pf-) goods. Publication of a novel scientific finding or method and development of open-source software ordinarily represent pf-goods as well.

The remaining class, fp-goods, is somewhat less natural than the other three because the per-capita benefit decreases with the number of recipients, while the overall cost grows. One way in which such a cost structure might arise is via the production of a divisible, rival good that involves a cost associated to its transmission to a recipient. That being said, we focus our examples primarily on pp-, ff-, and pf-goods since all of the interesting behavior we observe can be illustrated using these kinds of social goods. Our theoretical results cover a much broader class of social goods, however, and we discuss how they can be used to understand the effects of general functional dependencies, asymmetric games, and stochastic payoffs on selection.

\section{Results}
In our model, the interaction structure of the game is given by a graph (or social network) of size $N$, in which individuals occupy nodes. The adjacency matrix of this graph satisfies $w_{ij}=1$ if $i$ and $j$ are neighbors and $w_{ij}=0$ if $i$ and $j$ are not neighbors. The links specify interactions between individuals. Each individual can choose between two strategies. An individual with strategy $A$ (a ``producer'') generates goods to distribute among neighbors and pays costs for doing so. An individual with strategy $B$ (a ``non-producer'') provides no benefits to others and incurs no costs.

Let $w_{i}=\sum_{j=1}^{N}w_{ij}$ denote the number of neighbors of $i$ (the ``degree'' of $i$). For a pp-good, a producer at location $i$ pays total cost, $cw_{i}$, and the total benefit $bw_{i}$ is split among the $w_{i}$ neighbors; thus, each neighbor receives $b$. For an ff-good, the producer pays cost, $c$, and the total benefit $b$ is split among the $w_{i}$ neighbors; thus, each neighbor receives $b/w_{i}$. A pf-good is a hybrid of these two goods; the total cost is $c$ and each neighbor gets $b$. The behavior is prosocial if both $b$ and $c$ are positive, which we assume throughout this study. However, we make no assumptions regarding the ranking of $b$ and $c$; we allow $b>c$, $b=c$, and $b<c$.

The first question that needs to be explored is: when is a prosocial behavior wealth producing? A natural measure for total wealth is simply the sum over all benefits minus all costs, assuming everyone is a producer. Using this approach, we find that the answer for both pp- and ff-goods is immediate: on any graph, the prosocial good is wealth producing if and only if $b>c$. In contrast, pf-goods can be wealth producing even when $b<c$ since the total benefit, $b\sum_{i,j=1}^{N}w_{ij}$, is based on the number of edges and the total cost, $cN$, is based on the number of nodes. More specifically, a pf-good is wealth producing on a graph if and only if $b/c>N/\sum_{i,j=1}^{N}w_{ij}$.

The second question concerns inequality and possible social harm. On a heterogeneous graph, it is clear that even if everyone produces the social good, highly-connected individuals can accumulate a much higher payoff than others. Depending on the graph structure, a small number of individuals could hold the large majority of the wealth that is being produced. The poorest individuals can also end up with negative payoffs, which we call ``harmful prosociality.'' In this case, the poorest members of the population would be better off in the all-$B$ state than in the all-$A$ state. For pp- and pf-goods, harmful prosociality can arise only if $b<c$, but for ff-goods it can arise even if $b>c$.

The third question is: under which conditions do producers evolve in a structured population? Since there is neither mutation nor migration in our model, we use the notion of ``fixation probability'' to quantify the effects of selection on a population. Let $\rho_{A}$ be the probability that trait $A$, held initially by just a single individual within the population, eventually fixes and replaces the resident $B$-population. Similarly, we denote by $\rho_{B}$ the fixation probability of type $B$, defined in the same way as $\rho_{A}$ but with $A$ and $B$ swapped. We say that selection favors $A$ relative to $B$ if $\rho_{A}>\rho_{B}$. Intuitively, in a process with negligible mutation rates between the types, this condition means that the population spends more time in the all-$A$ state than in the all-$B$ state \citep{fudenberg:JET:2006}.

We first derive a general result that applies to almost any evolutionary update mechanism as long as some natural properties hold. Suppose that a producer at location $i$ donates $B_{ij}$ to $j$ at a cost of $C_{ij}$ (\fig{generalModel}). Let $\pi_{i}$ be the fixation probability of a neutral trait starting in location $i$ (also known the ``reproductive value'' \citep{fisher:OUP:1930,maciejewski:JTB:2014a} of $i$). In the SI, we define a natural distribution over the non-monomorphic states (meaning states with both types, $A$ and $B$) under neutral drift, and we let $x_{ij}$ be the probability that $i$ and $j$ have the same type in this distribution. Finally, let $m_{k}^{ij}$ be the marginal effect of $k$'s fecundity on the probability that $i$ replaces $j$. Using these quantities, which are described in detail in the SI, we show that producers ($A$) are favored over non-producers ($B$) under weak selection \citep{wild:JTB:2007,fu:PRE:2009,wu:PRE:2010,wu:PLoSCB:2013,mullon:JEB:2014} whenever
\begin{align}
\sum_{i,j,k,\ell =1}^{N} \pi_{i} m_{k}^{ji} \left( -x_{jk}C_{k\ell} + x_{j\ell}B_{\ell k} \right) > \sum_{i,j,k,\ell =1}^{N} \pi_{i} m_{k}^{ji} \left( -x_{ik}C_{k\ell} + x_{i\ell}B_{\ell k} \right) . \label{eq:generalConditionSG}
\end{align}
This condition can be evaluated by solving a linear system of $O\left(N^{2}\right)$ equations, giving an overall complexity of $O\left(N^{6}\right)$ (since solving a linear system of $n$ equations requires $O\left(n^{3}\right)$ operations). We give examples and interpretations of this condition for specific update rules in Methods.

\begin{figure}
	\centering
	\includegraphics[width=0.9\textwidth]{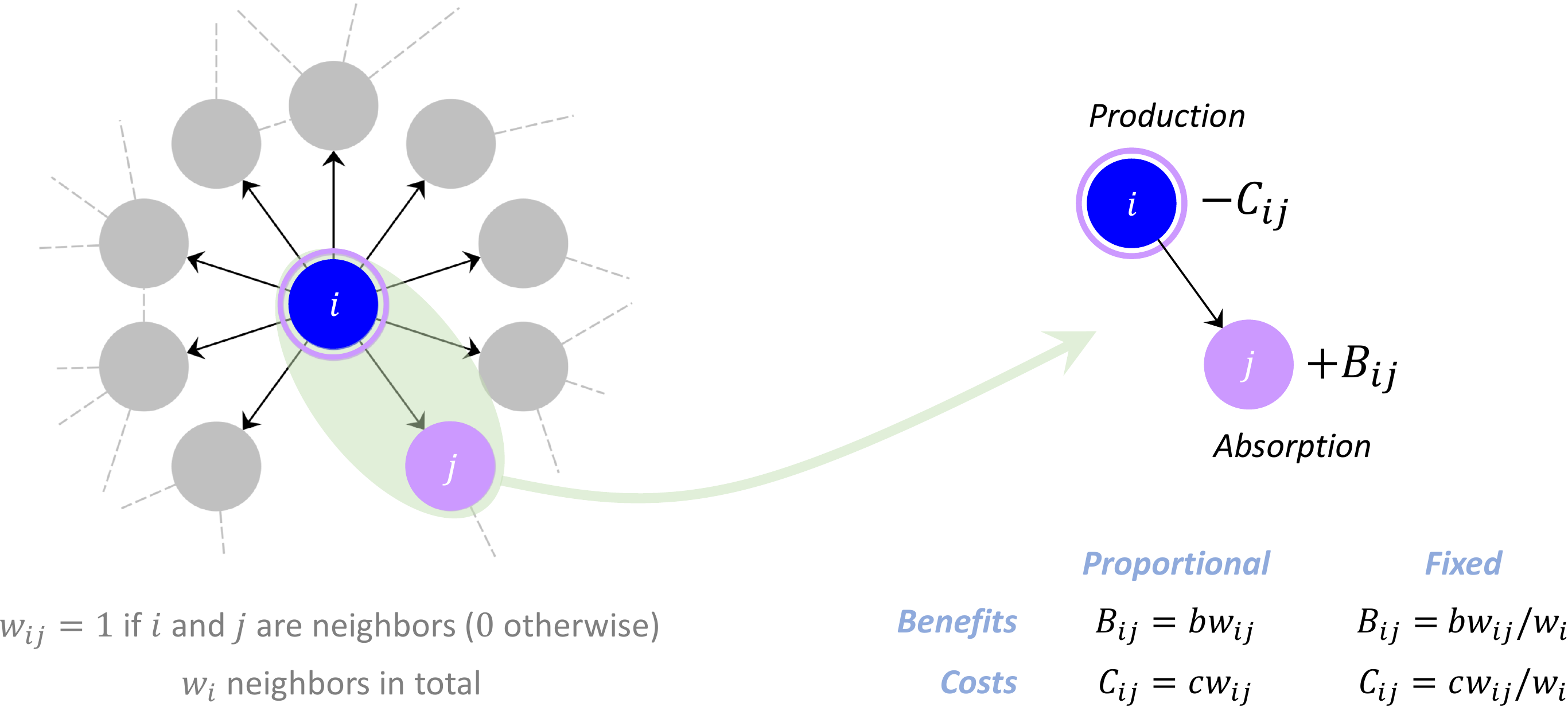}
	\caption{\textbf{Production and absorption of social goods.} A producer at location $i$ donates $B_{ij}$ to $j$, at a cost of $C_{ij}$. We give formulas for $B_{ij}$ and $C_{ij}$, for each of the social goods we analyze, in terms of the adjacency matrix of the population structure, which is defined by $w_{ij}=1$ if $i$ and $j$ share an edge and $w_{ij}=0$ otherwise. For ff-goods, with fixed benefits and fixed costs, a producer at location $i$ pays $c$ to donate a total benefit of $b$, giving $b/w_{i}$ to each of $i$'s $w_{i}$ neighbors. The effective cost attributable to a neighbor, $j$, is $C_{ij}=c/w_{i}$. Our main examples of benefits and costs all have the property that $B_{ij}=b\beta_{ij}$ and $C_{ij}=c\gamma_{ij}$ for some $\beta_{ij}$ and $\gamma_{ij}$ (which are both independent of $b$ and $c$). $b$ and $c$ quantify the ``magnitude'' of the social good, while $\beta_{ij}$ and $\gamma_{ij}$ quantify both the nature of the good and how it is distributed by a producer.\label{fig:generalModel}}
\end{figure}

The social goods we consider here have the property that $B_{ij}=b\beta_{ij}$ and $C_{ij}=c\gamma_{ij}$ for some $b,c>0$, where $\beta_{ij}$ and $\gamma_{ij}$ are independent of $b$ and $c$. As a consequence, \eq{generalConditionSG} can be written as $\gamma b>\beta c$, where $\beta$ and $\gamma$ are independent of $b$ and $c$. When $\gamma >0$, this condition implies that $\rho_{A}>\rho_{B}$ in the limit of weak selection whenever $b/c>\left(b/c\right)^{\ast}=\beta /\gamma$. $\left(b/c\right)^{\ast}$ is known as the ``critical benefit-to-cost ratio'' for producers to evolve. If $\gamma <0$, then the condition for producers to evolve is $b/c<\left(b/c\right)^{\ast}=\beta /\gamma$. Thus, a negative critical ratio implies that prosocial behaviors, i.e. those with $b,c>0$, cannot evolve; instead, selection can favor spiteful behaviors with $b<0$ and $c>0$ (``costly harm'') \citep{iwasa:EE:1998,forber:PRSB:2014}. This property is one nuance of critical ratios, namely that they are lower bounds on $b/c$ when $\gamma >0$ and upper bounds when $\gamma <0$.

\subsection{Evolutionary outcomes on heterogeneous structures}
We consider several natural update rules that drive evolution through imitation. Under pairwise-comparison (PC) updating \citep{szabo:PRE:1998}, a random individual is chosen to update its strategy. It compares its own payoff with that of a single, randomly chosen neighbor. If the neighbor has a higher payoff, then the focal individual adopts the neighbor's strategy. If the neighbor has a lower payoff, then the focal individual retains its current strategy. The payoff comparison is subject to noise. Death-birth (DB) and imitation (IM) updating are similar, but they differ in the number of neighbors chosen for comparison and/or whether imitating some neighbor is compulsory (see \fig{update_rules}).

\begin{figure}
	\centering
	\includegraphics[width=0.8\textwidth]{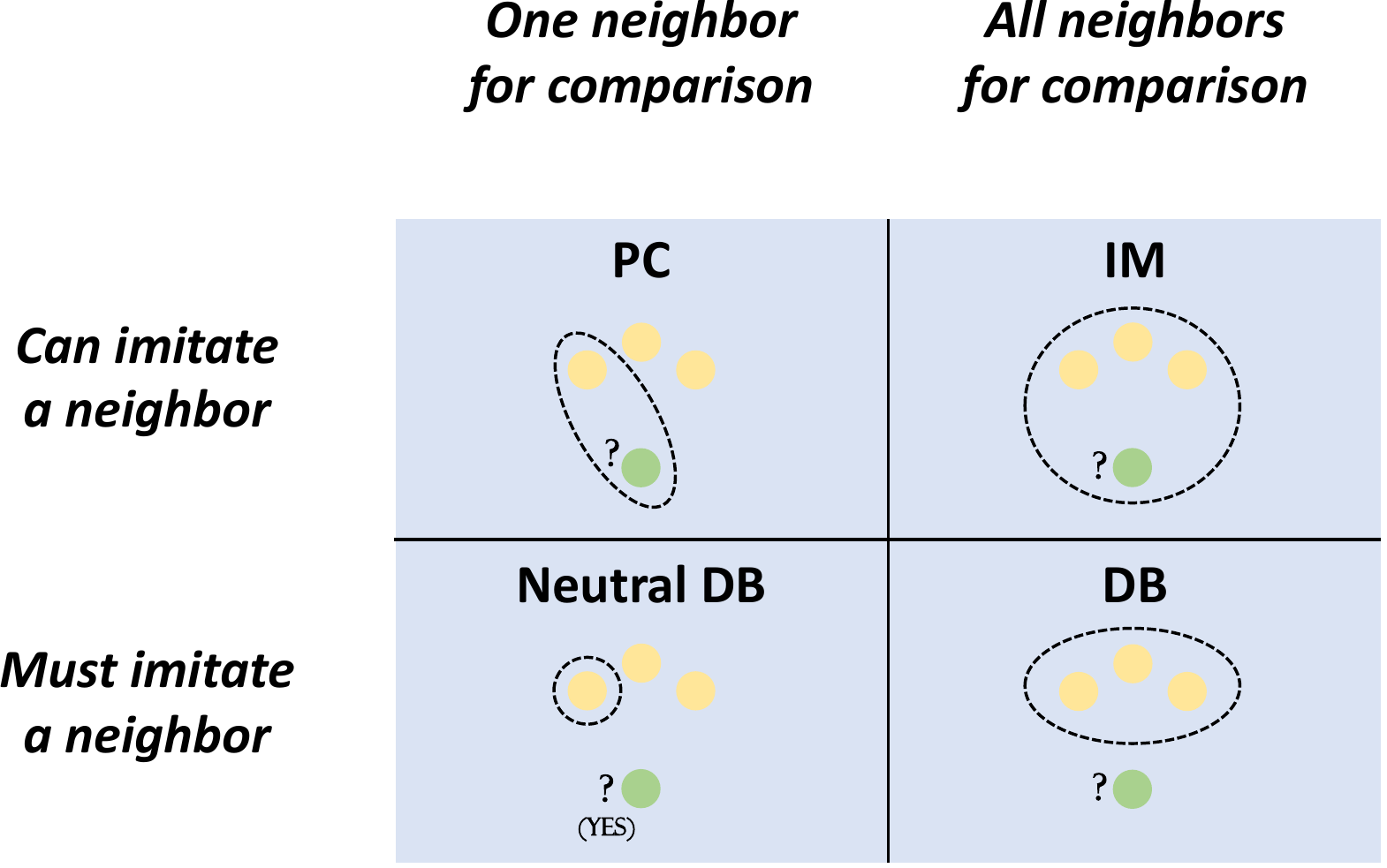}
	\caption{\textbf{Four update rules driving evolutionary dynamics through imitation.} When considering how to imitate a neighbor's action ($A$ or $B$) based on payoff, four natural update rules arise, which have each been considered extensively in the literature. Pairwise-comparison (PC) updating involves an individual choosing a random neighbor (yellow) with whom he or she compares payoffs. There is an option to imitate the neighbor, but the focal individual (green) may also choose to retain their existing action. Imitation (IM) updating is similar to this rule, except that the payoff comparison involves a focal individual and all neighbors. Again, this individual can imitate a neighbor but does not have to do so. If one insists that one of these neighbors must be imitated, then we have death-birth (DB) updating. In this case, the green individual is effectively chosen for death because retaining its current behavior is not an option. The final logical case is when, like in PC updating, only a single neighbor is chosen for comparison. This time, however, the focal individual must imitate this neighbor. This model turns out to be equivalent to DB updating when there is no game (i.e. neutral drift) and is therefore not relevant to studying the effects of selection. Our general result (\eq{generalConditionSG}) can also account for many update rules beyond these simple (but important) examples.\label{fig:update_rules}}
\end{figure}

These update rules are highly idealized, but they capture important qualitative features of behavior imitation \citep{fudenberg:JET:2008,roca:PLR:2009,traulsen:PNAS:2010}. For one thing, a learner is more likely to imitate a model individual's behavior as the model's payoff increases. If a learner cannot compare his or her payoff to all neighbors at once (for instance, if he or she encounters neighbors only occasionally), then PC updating is relevant. When information is more readily available, e.g. within scientific collaboration networks, then IM updating could serve as a better model. In both cases, a learner is not compelled to imitate a behavior. DB updating, which requires imitating some neighbor, could be interpreted in terms of personnel turnover within an organization, for example. An individual in the network might be replaced by a newcomer, who then copies a behavior of someone nearby. We use these (well-studied) update rules to illustrate interesting evolutionary dynamics of social goods, but we emphasize that \eq{generalConditionSG} can readily be applied to a wide variety of update mechanisms.

For PC updating, we prove that producers of pp-, ff-, and pf-goods are never favored on homogeneous graphs. On heterogeneous graphs, it is possible that producers evolve if the benefit-to-cost ratio exceeds a critical value, $\left(b/c\right)^{*}$. For pp-goods, we find that $\left(b/c\right)^{*}$ can never be between zero and one, which means that $b>c$ is a necessary condition for producers to evolve. Thus, for pp-goods, producers can evolve only if they improve the overall wealth of the population. Moreover, they can evolve only if they lead to a positive payoff for even the poorest individuals. In the SI, we also establish this result for pp-goods on heterogeneous graphs under DB and IM updating. In contrast to PC updating, DB and IM updating are known to support the evolution of producers of pp-goods on homogeneous graphs provided the benefit-to-cost ratio is sufficiently large \citep{ohtsuki:Nature:2006,chen:AAP:2013}.

Consider, for example, a ``rich club'' network \citep{zhou:IEEE:2004,colizza:NP:2006,mcauley:APL:2007,fotouhi:NHB:2018}, which is defined by a central clique of $m$ individuals, who are connected to each other as well as to $n$ individuals at the periphery. The peripheral individuals are connected to only those in the central clique. This network provides an abstraction of an oligarchy, which is defined by three primary components: ``the elite are tightly interconnected among themselves, forming an `inner circle'; the masses are organized through the intermediation of this inner circle; and the masses are poorly interconnected among themselves'' \citep{ansell:C:2016}. Structures like the rich club arise within corporate hierarchies \citep{dong:PLOSONE:2015}, among students in a classroom (based on academic performance) \citep{vaquero:SR:2013}, and among academic institutions (based on funding) \citep{ma:PNAS:2015,szell:PNAS:2015}. Surprisingly, $\left(b/c\right)^{*}$ can fall between zero and one for both ff- and pf-goods in such populations (see \fig{richClub} and \sfig{star_intuition}). In particular, producers of ff-goods can evolve even when the total cost of a good exceeds its total benefit.

\begin{figure}
	\centering
	\includegraphics[width=0.9\textwidth]{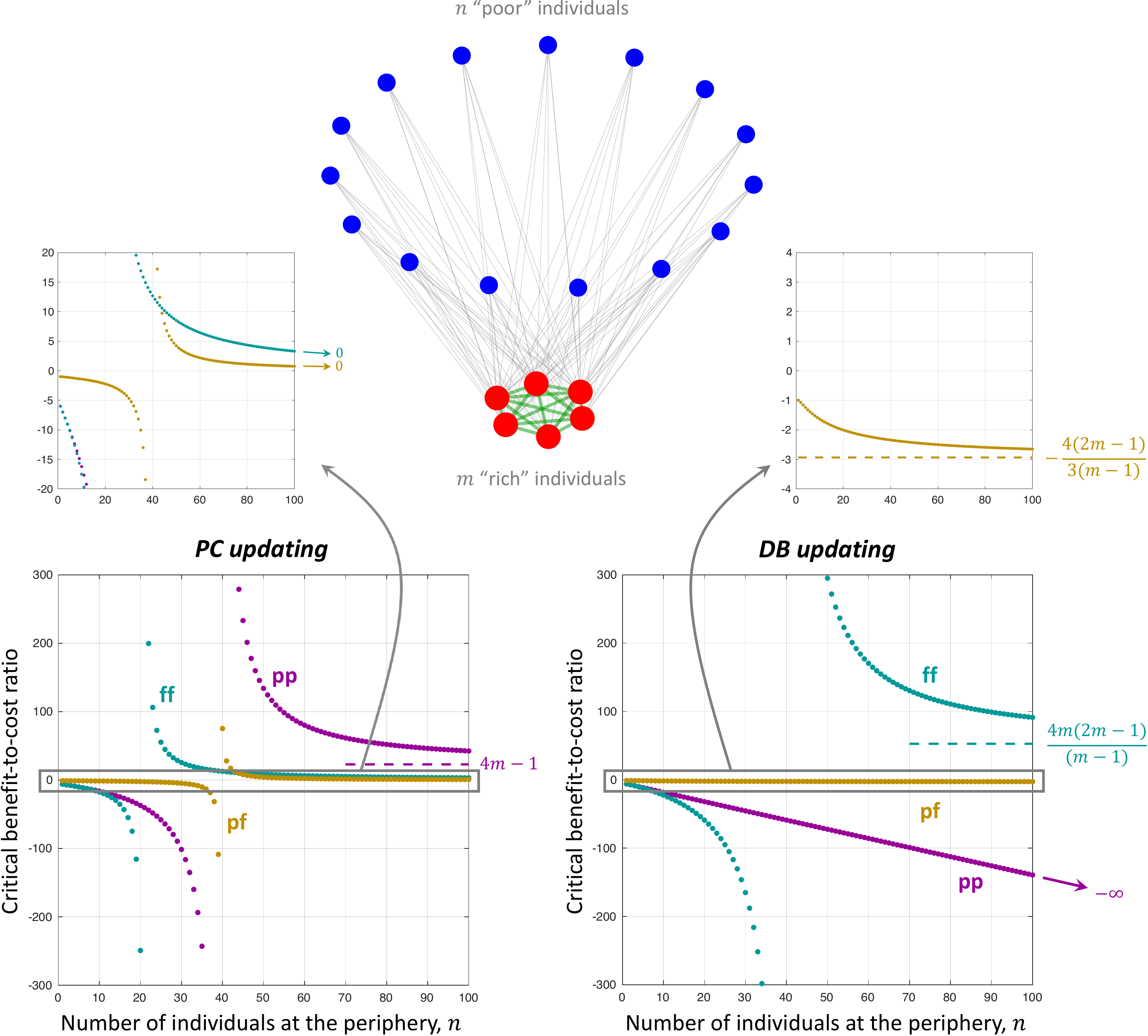}
	\caption{\textbf{Heterogeneous graphs allow for efficient evolution of prosocial behavior.} The ``rich club'' network consists of a central clique of $m$ well-connected individuals surrounded by $n$ individuals at the periphery \cite{fotouhi:NHB:2018}. Each of the $m$ ``rich'' individuals is connected to every other member of the population, but the $n$ ``poor'' individuals are connected to only the central clique. For $m=6$, we illustrate the effects of increasing $n$ under PC and DB updating. Under PC updating, the critical ratio for pp-goods approaches $4m-1$ while those of ff- and pf-goods, remarkably, approach $0$. Therefore, producers of ff-goods can evolve even if the total costs exceed the total benefits. For DB updating with $m>1$, only ff-goods have a positive critical ratio when $n$ is large. This ratio is negative for both pp- and pf-goods, which means that selection can favor ``spiteful'' pp- and pf-goods provided these goods are sufficiently harmful to others (relative to their cost). For an ff-good with benefit $b$ and cost $c$, as $n$ grows large the all-producer state results in a payoff that approaches $\infty$ to each of the $m$ individuals and a payoff that approaches $-c$ to each of the $n$ individuals. Thus, the rich club can hold more than 100\% of the total wealth. We give explicit formulas for $\left(b/c\right)^{\ast}$ for any $m$ and $n$ in the SI.\label{fig:richClub}}
\end{figure}

\begin{figure}
	\centering
	\includegraphics[width=1.0\textwidth]{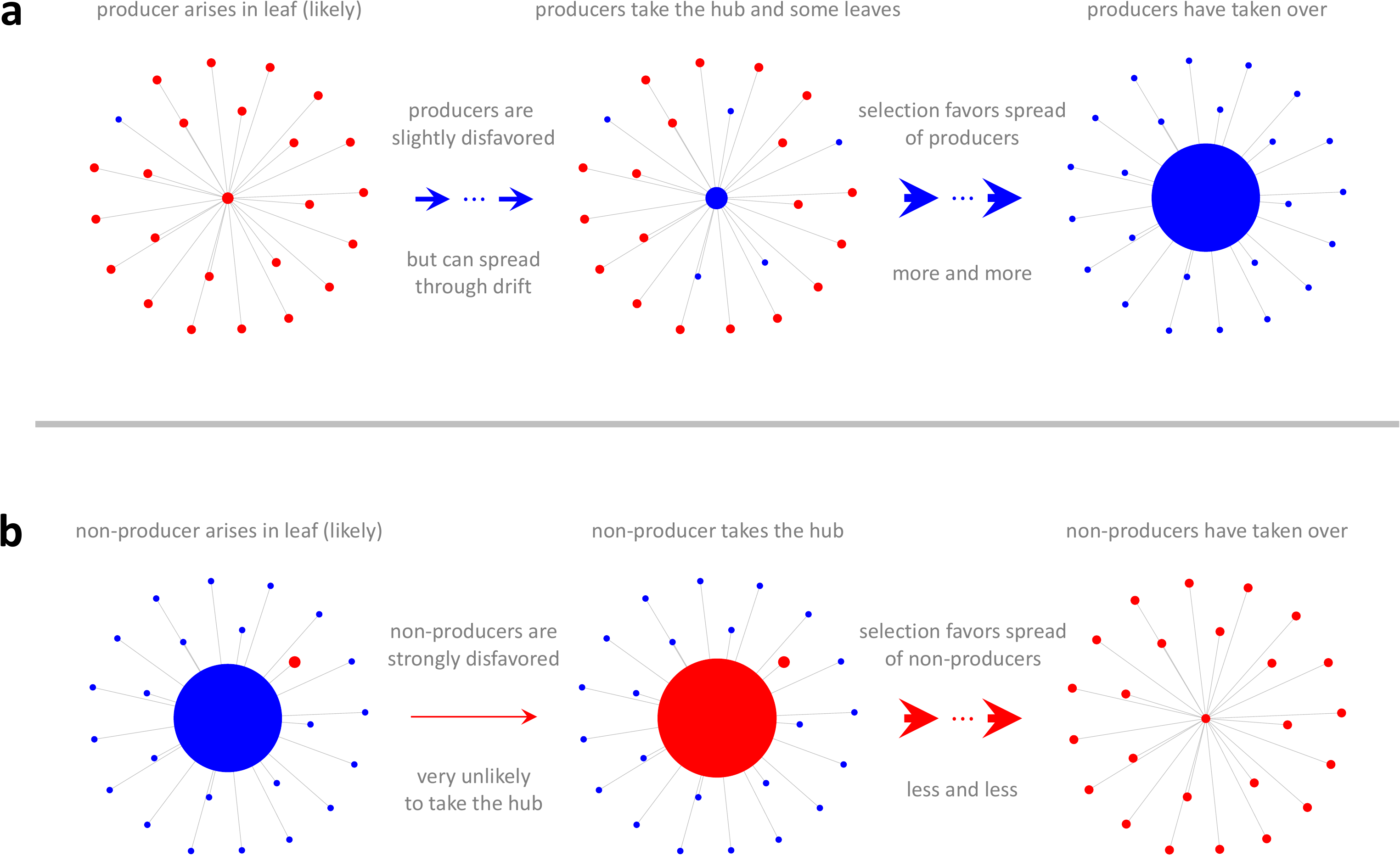}
	\caption{\textbf{Evolution of producers of ff-goods with $0<b\leqslant c$ on the star (PC updating).} The star may be viewed as a special case of the rich club, in which there is just a single ``rich'' individual ($m=1$). \textbf{a}, invasion and fixation of a mutant producer arising in a leaf under PC updating. This producer has a payoff of $-c$, and the non-producer at the hub gets $b$. Through drift, this producer can take the hub and propagate a small portion of producers to the leaves. Once there are $k>c/b+1/\left(N-1\right)$ producers at the periphery, a central producer's payoff exceeds that of everyone else in the population and selection favors the further spread of producers. \textbf{b}, invasion and fixation of a mutant non-producer arising in a leaf. As soon as a non-producer captures the hub, selection favors the proliferation of non-producers. However, when there is just a single non-producer in the population, a producer at the hub has a much greater payoff than everyone else in the population (even when $0<b\leqslant c$). Thus, relative to the initial invasion of a producer in \textbf{a}, selection acts much more strongly against the initial invasion of a non-producer in \textbf{b}. For any fixed $b,c>0$, these effects become strong enough as $N$ grows that we find $\rho_{A}>\rho_{B}$.\label{fig:star_intuition}}
\end{figure}

Note that the critical ratios for pp- and pf-goods in \fig{richClub} are either both positive or both negative. This example alludes to a more general finding, for all update rules: a graph can support producers of pp-goods for sufficiently large $b/c$ if and only if the same is true for pf-goods. A similar pairing occurs between ff- and fp-goods. In fact, whether there exists some $b,c>0$ for which selection favors producers depends on only the benefits of the social good; the costs influence the magnitude of $\left(b/c\right)^{\ast}$ but not its sign (see Methods and SI for details). Statistically speaking, we find that there are more population structures on which producers of ff-goods can evolve than there are for producers of pp- or pf-goods (see Extended~Data~Figures~\ref{fig:ER_SW}--\ref{fig:smallGraphs}). We also observe critical ratios of strictly less than one for ff- and pf-goods on other kinds of graphs, including random graphs (\sfig{preferential_attachment}) and other rich club structures \citep{jiang:NJP:2008} like dense clusters of stars (\sfig{clusterOfStars}). Further comparisons of critical ratios for different kinds of social goods are shown in \sfig{smallGraphsSmallRatios} (small graphs) and \sfig{zachary} (division of a group into factions).

\newpage

\begin{figure}
	\centering
	\includegraphics[width=1.0\textwidth]{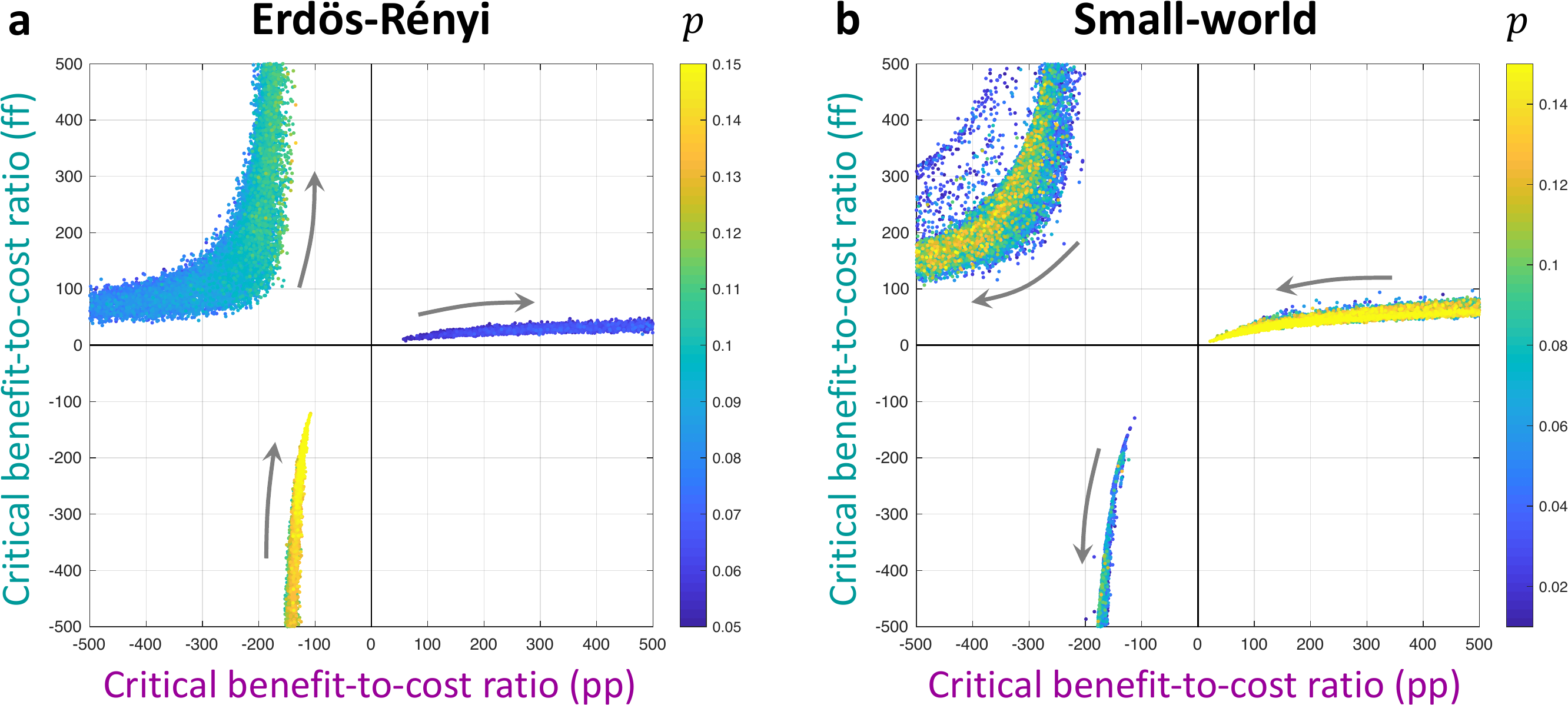}
	\caption{\textbf{Evolution of producers can be possible for ff-goods but not for pp-goods.} \textbf{a}, PC updating on Erd\"{o}s-R\'{e}nyi graphs of size $N=100$ for various edge-inclusion probabilities, $p$. If $p$ is sufficiently small, the critical benefit-to-cost ratio is positive for both pp- and ff-goods, but for slightly larger $p$ values this ratio can be positive for ff-goods and negative for pp-goods. In the latter case, producers cannot evolve under any $b/c$ ratio for pp-goods, but they can evolve for ff-goods as long as $b/c$ is sufficiently large. \textbf{b}, PC updating on small-world networks with different rewiring probabilities, $p$. Again, there are many examples for which the critical benefit-to-cost ratio is positive for ff-goods but negative for pp-goods.\label{fig:ER_SW}}
\end{figure}

\begin{figure}
	\centering
	\includegraphics[width=1.0\textwidth]{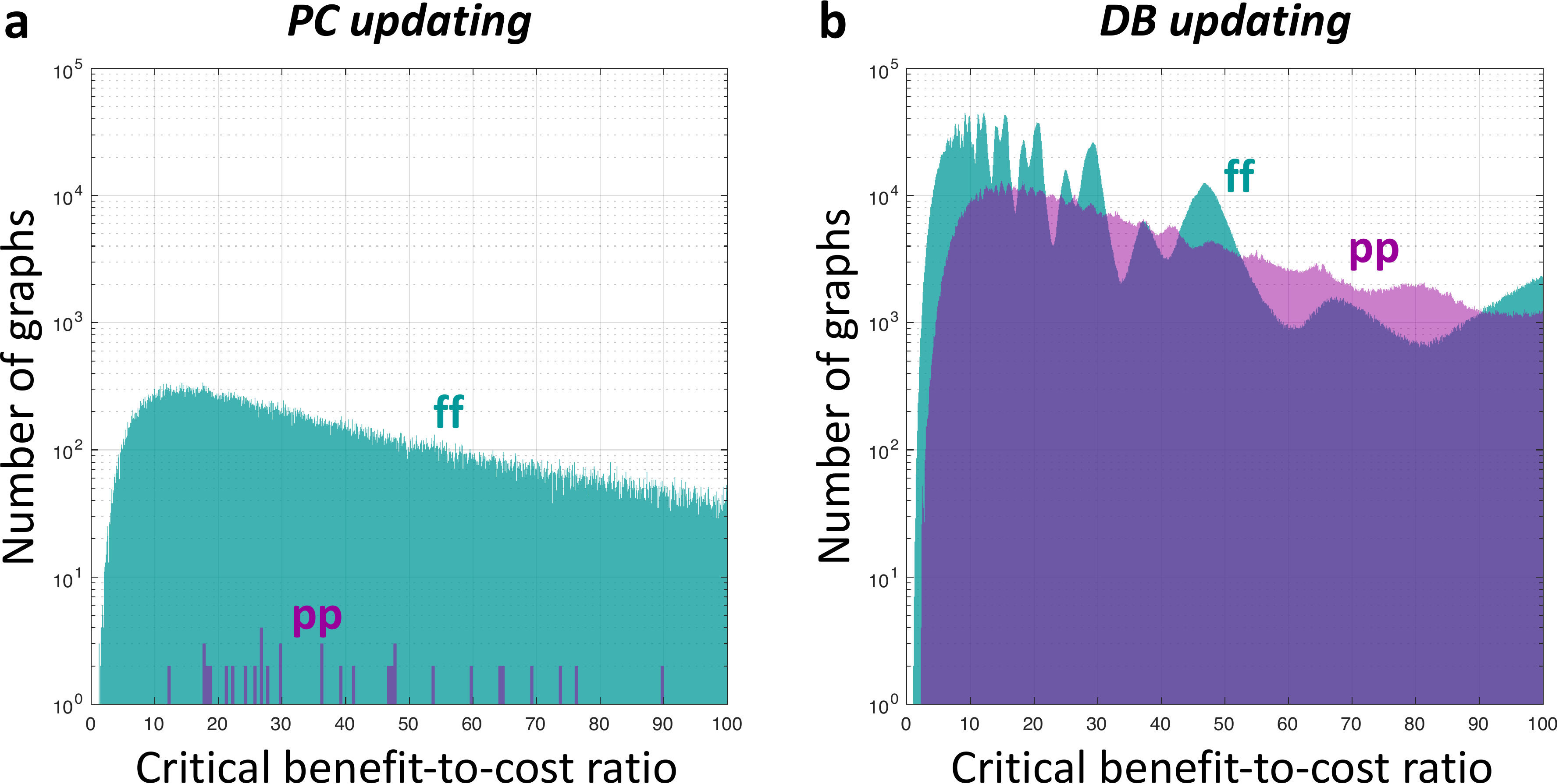}
	\caption{\textbf{Distributions of critical ratios on small graphs.} There are 11,989,763 undirected, unweighted graphs of size at most $N=10$. Of those that can support the evolution of prosocial behaviors, the critical benefit-to-cost ratios are given for PC updating, \textbf{a}, and DB updating, \textbf{b}.\label{fig:smallGraphs}}
\end{figure}

\begin{figure}
	\centering
	\includegraphics[width=1.0\textwidth]{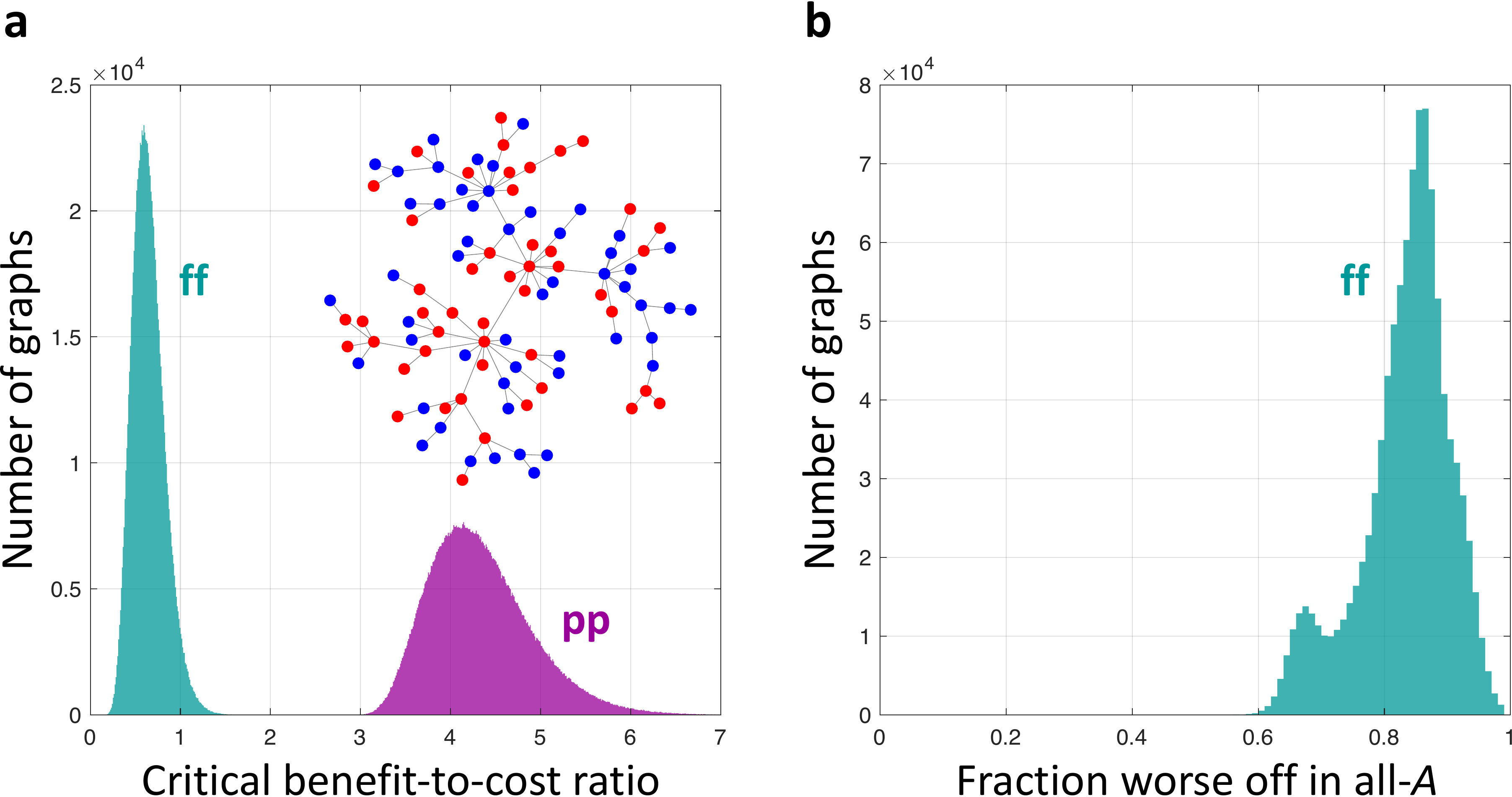}
	\caption{\textbf{Heterogeneous graphs allow efficient evolution of prosocial behavior.} \textbf{a}, The distribution of the critical benefit-to-cost ratio under PC updating for $10^{6}$ preferential-attachment graphs of size $N=100$ (see Methods). For ff-goods, these structures often have critical benefit-to-cost ratios that are less than one. However, the critical ratio for pp-goods is always greater than one. \textbf{b}, When $b/c=\left(b/c\right)^{\ast}$, these graphs result in a majority (but not all) of the population being worse-off in the all-producer state than in the all-non-producer state.\label{fig:preferential_attachment}}
\end{figure}

\begin{figure}
	\centering
	\includegraphics[width=1.0\textwidth]{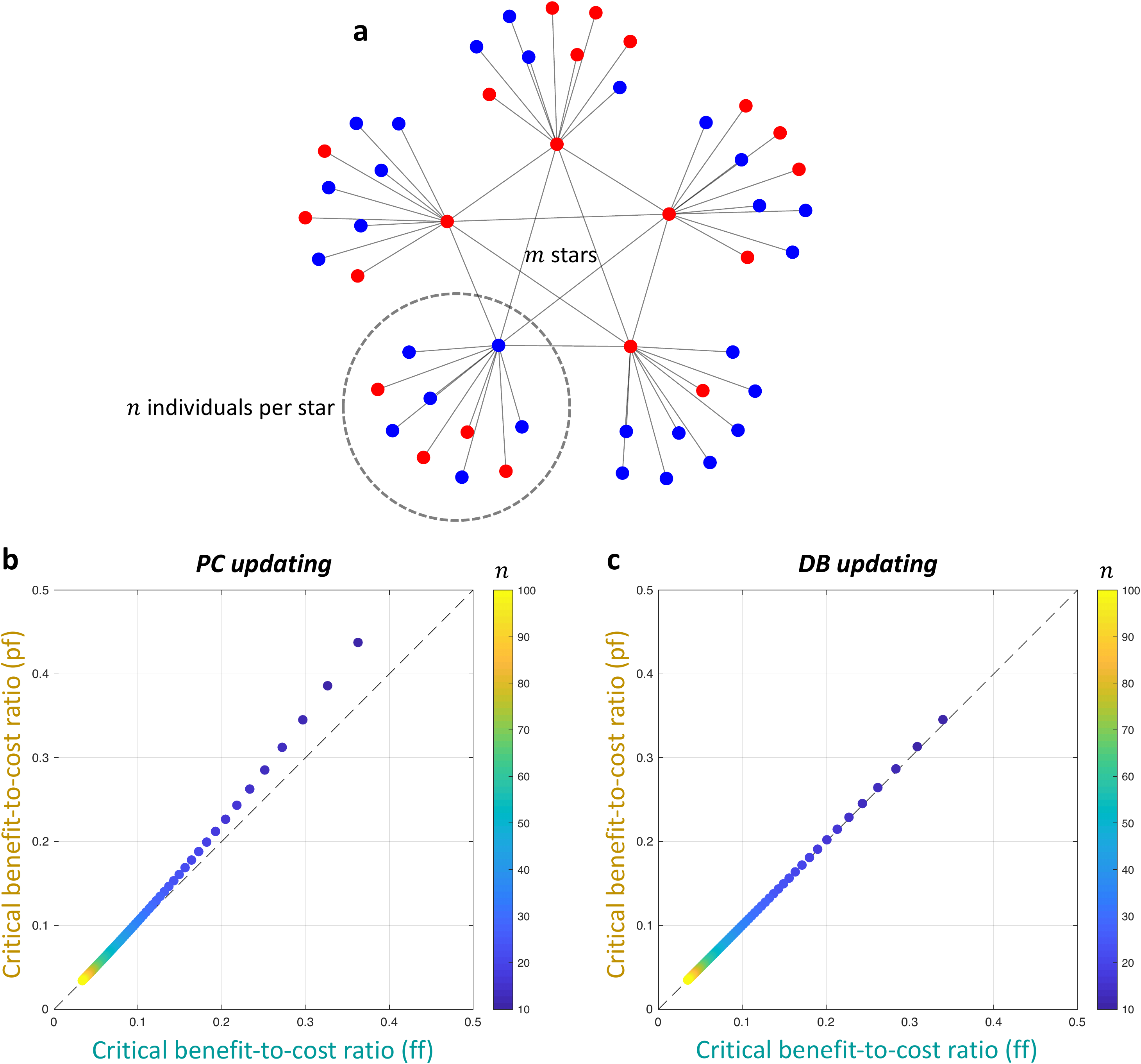}
	\caption{\textbf{Fixed costs on a dense cluster of stars.} Consider a population consisting of $m$ stars, each of size $n$, connected by a complete graph at their hubs. Provided $m>1$, this structure results in extremely low critical ratios under both PC and DB updating when $n$ is large. Illustrated here is the case in which $m=5$. This structure has the interesting property that ff-goods result in lower critical thresholds than pf-goods (both of which are lower than that of pp-goods, which is not depicted here). Qualitatively, the results are similar for both PC and DB updating, with the exception that producers can never be favored by selection on the star ($m=1$) under DB updating but can be favored under PC updating. We derive explicit formulas for $\left(b/c\right)^{\ast}$ for any $m$ and $n$ in the SI.\label{fig:clusterOfStars}}
\end{figure}

\begin{figure}
	\centering
	\includegraphics[width=1.0\textwidth]{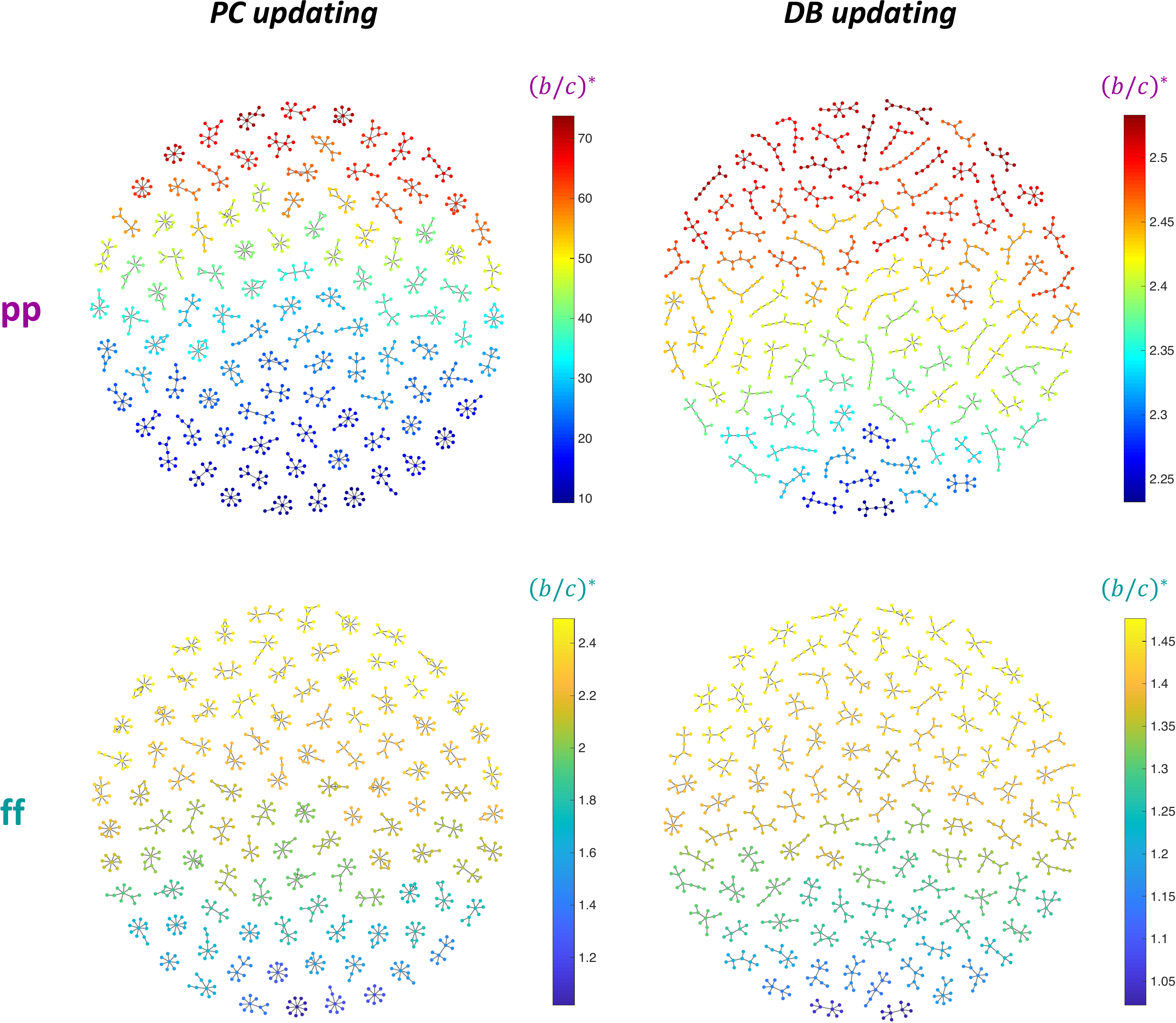}
	\caption{\textbf{Graphs of size $N\leqslant 10$ that most easily support prosocial behavior.} For PC and DB updating, we illustrate the $100$ graphs with the lowest positive critical ratios for pp- and ff-goods. In each case, the graphs are colored according to their critical ratios. In these examples, ff-goods result in lower critical ratios than pp-goods, and DB updating tends to give lower ratios than PC updating.\label{fig:smallGraphsSmallRatios}}
\end{figure}

\begin{figure}
	\centering
	\includegraphics[width=1.0\textwidth]{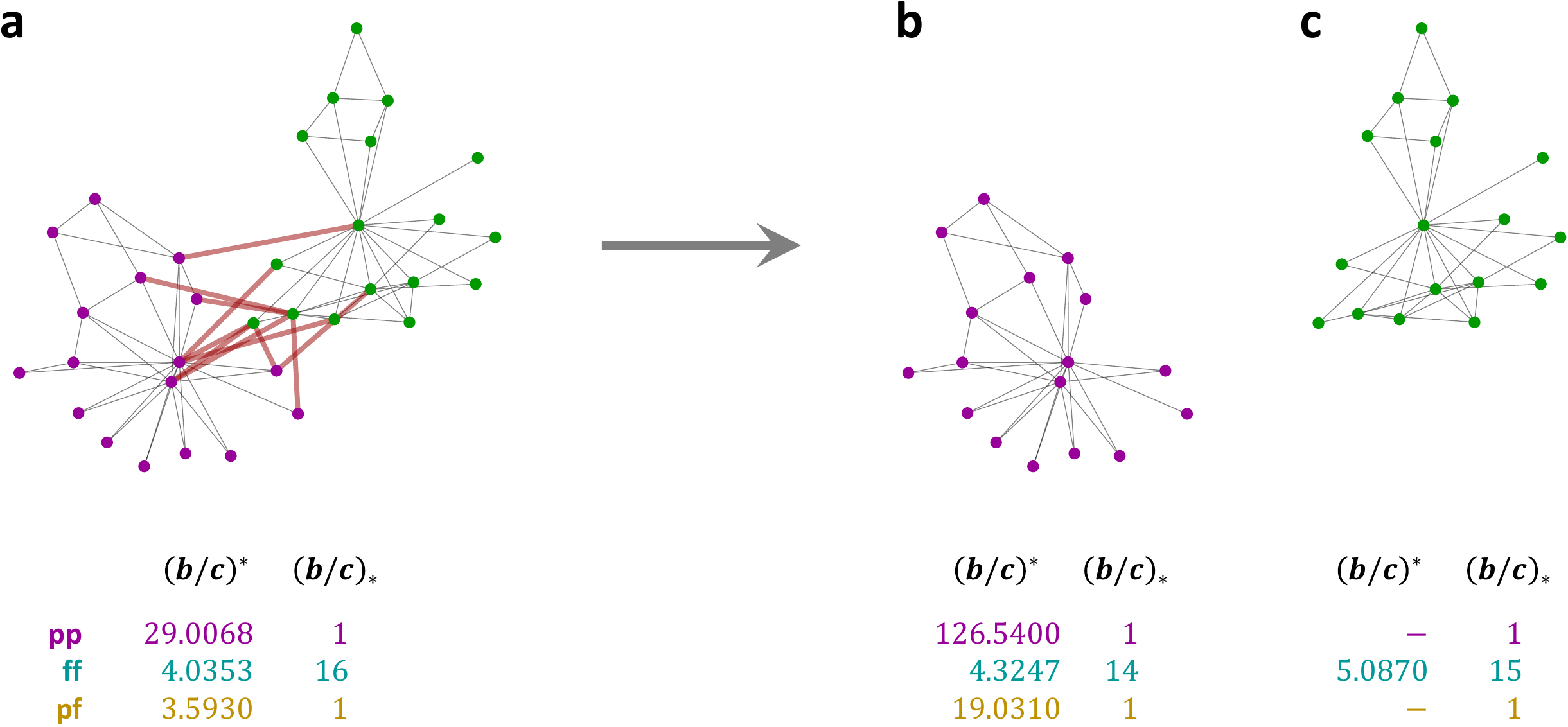}
	\caption{\textbf{Division of a society into two factions.} To illustrate the effects of the division of real-world interaction topologies on evolutionary dynamics, we consider Zachary's karate club \citep{zachary:JAR:1977}, \textbf{a}, and the subsequent split of the karate club into two disjoint groups \citep{girvan:PNAS:2002}, \textbf{b} and \textbf{c}. Under PC updating, producers can evolve on all three networks only in the case of ff-goods. Moreover, even for the two populations (\textbf{a} and \textbf{b}) in which both ff- and pf-goods can evolve, this split swaps the rankings of the two. In particular, the critical ratio for pf-goods is lower in \textbf{a} but that of ff-goods is lower in \textbf{b}. The threshold for all individuals to be better off in the all-$A$ state than in the all-$B$ state, $\left(b/c\right)_{\ast}$, is lowered by the split.\label{fig:zachary}}
\end{figure}

Even when producers improve the total wealth of a population, their evolution can leave the poorest individuals worse off. In \fig{netscience}, we illustrate this phenomenon using IM updating on an empirical coauthorship network of $379$ scientists \citep{newman:PRE:2006}. Producers of ff-goods can evolve only when $b/c\gtrapprox 2.5703$, in which case the total wealth of the population increases. But the proliferation of producers on this graph can leave nearly $20\%$ of the population with negative payoffs, meaning these individuals would prefer the all-$B$ state (non-producers) to the all-$A$ state (producers). The benefit-to-cost ratio must be at least $34$ before everyone is better off in all-$A$. Therefore, wealth-producing goods are not necessarily optimal for everyone in the population (see \sfig{summaryTable} for a summary of possible evolutionary outcomes).

\begin{figure}
	\centering
	\includegraphics[width=0.9\textwidth]{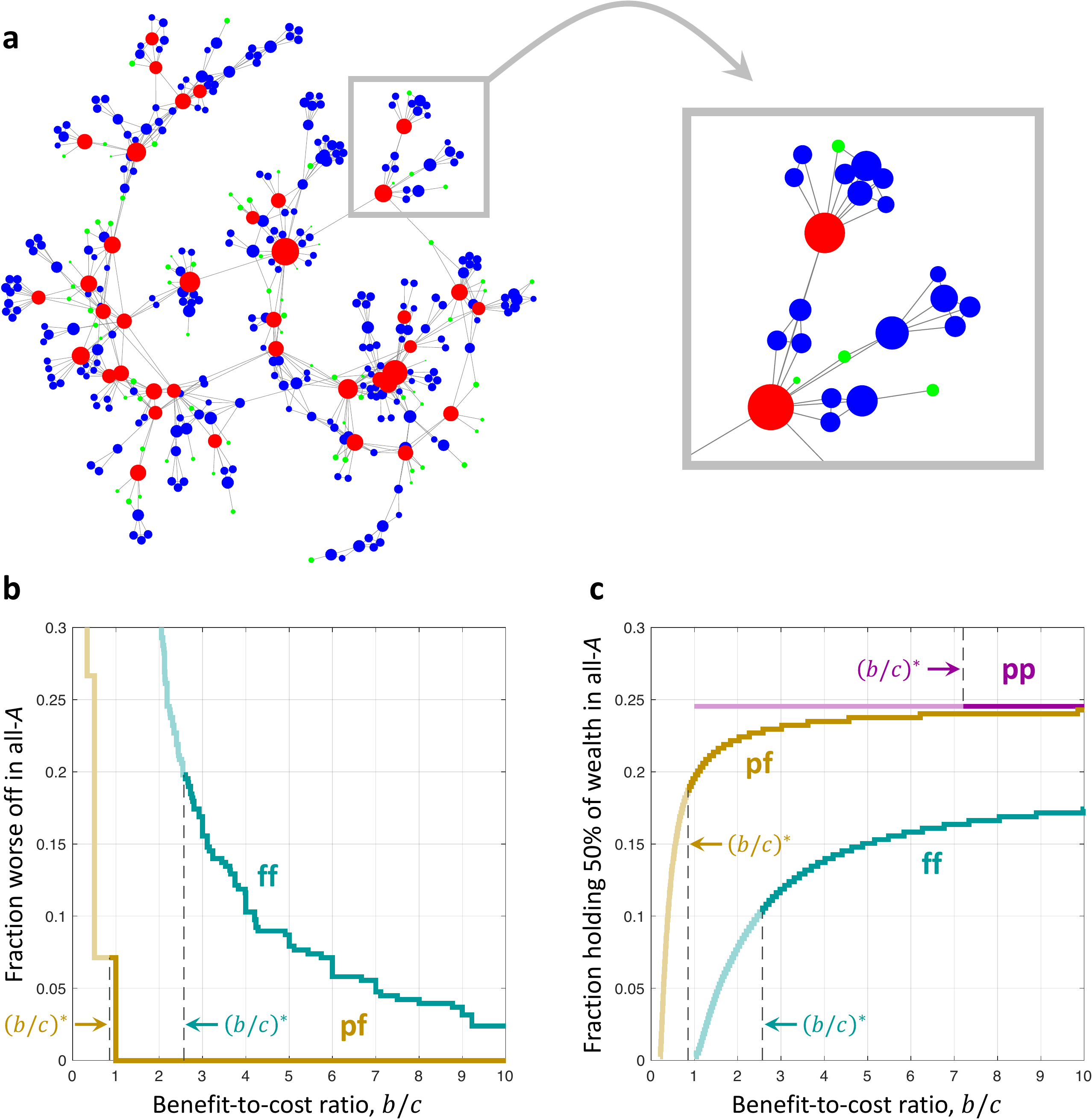}
	\caption{\textbf{Evolution of prosocial behavior can result in widespread inequality.} \textbf{a}, Prosocial inequality arising in a coauthorship network of $379$ scientists \citep{newman:PRE:2006}. For IM updating, ff- and pf-goods have lower barriers for selection to favor producers than pp-goods ($\left(b/c\right)^{\ast}\approx 7.2104$ for pp-goods, $\left(b/c\right)^{\ast}\approx 2.5703$ for ff-goods, and $\left(b/c\right)^{\ast}\approx 0.8506$ for pf-goods). In the all-$B$ (non-producer) state, all individuals have the same payoff. But for ff-goods, at the critical ratio where producers evolve, $75$ out of $379$ individuals (depicted in green on the graph) are worse off in all-$A$ than they would be in all-$B$. \textbf{b}, This fraction decreases as the benefit-to-cost ratio increases. The benefit-to-cost ratio must be at least $\left(b/c\right)_{\ast}=34$, which is much larger than $\left(b/c\right)^{\ast}=2.5703$, for all-$A$ to be better than all-$B$ for everyone. For pf-goods, there can be some harm to the poorest, but only when $b<c$. In contrast, when producers of pp-goods evolve, all individuals in the population are better off. \textbf{c}, In the all-$A$ state, a smaller portion of the population holds at least 50\% of the wealth (depicted in red on the graph) for ff- and pf-goods than for pp-goods.\label{fig:netscience}}
\end{figure}

\begin{figure}
	\centering
	\includegraphics[width=1.0\textwidth]{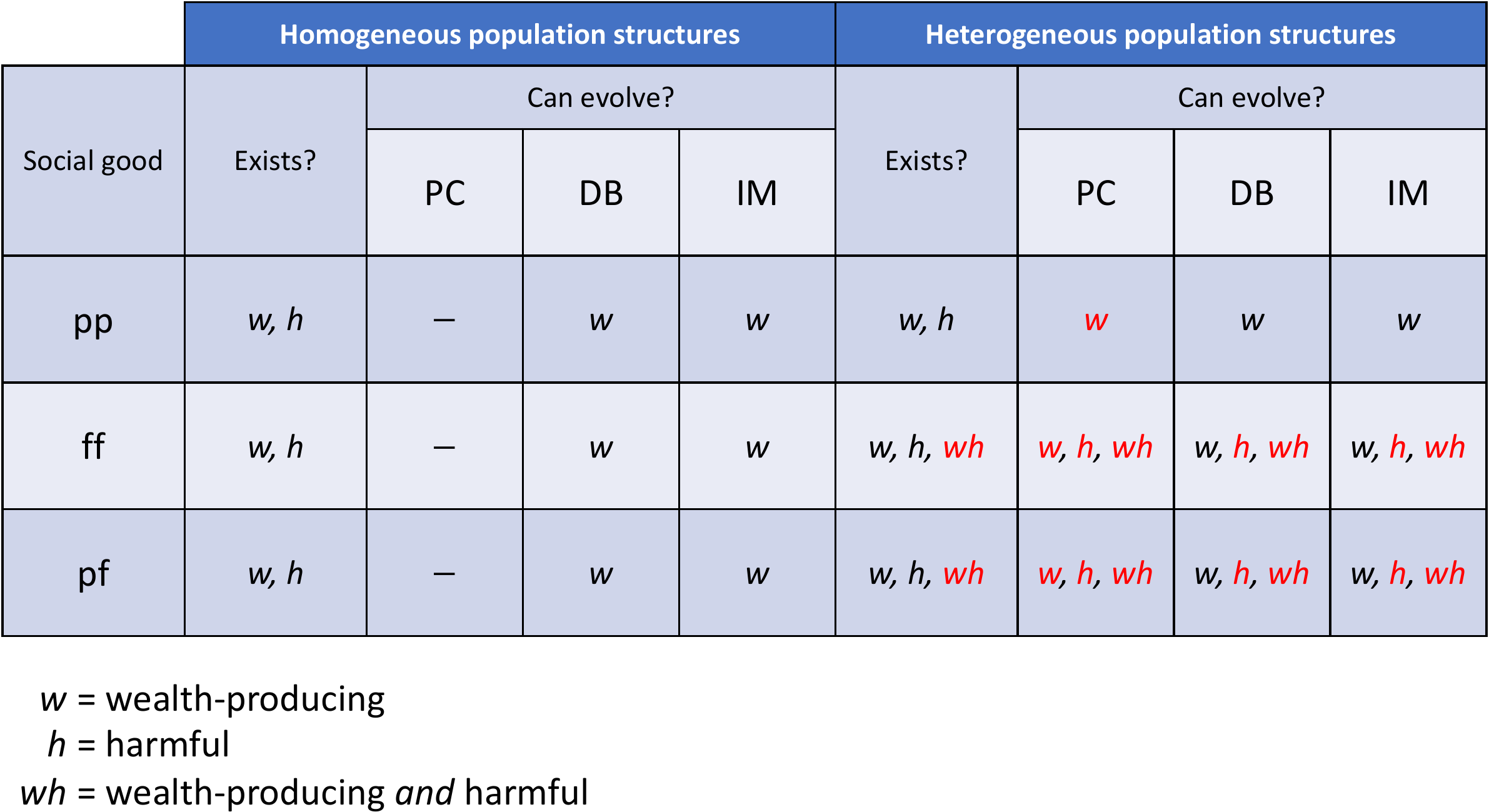}
	\caption{\textbf{Summary of main examples.} A good is wealth-producing (\emph{w}) if the total payoff (sum of all benefits minus sum of all costs) is positive when everyone in the population is a producer. It is harmful (\emph{h}) if at least one individual has a negative payoff in the all-producer state. For three kinds of social goods (pp, ff, and pf) and update rules (PC, DB, and IM), this table summarizes when a good can be wealth-producing and/or harmful, as well as when such a good can evolve. Notably, these results are not influenced much by the choice of update rule.\label{fig:summaryTable}}
\end{figure}

\subsection{Asymmetric games}
Although our focus so far has been on simple kinds of social goods, our model covers much more complicated asymmetric games. In particular, $B_{ij}$ and $C_{ij}$ can each be any number, and these values need not come from a social good with certain properties shared among all producers. For example, the cost of scientific collaboration, in terms of effort and time, can be less for those a supervisory role than it is for more junior authors. The benefits to the authors might be comparable across roles, e.g. in terms of recognition. Of course, the nature of this difference (both in terms of costs and benefits) is highly dependent on the discipline, with some disciplines being more egalitarian than others \citep{laurance:Nature:2006,venkatraman:Science:2010,bosnjak:S:2012,mcnutt:PNAS:2018}.

In addition to benefits and costs, the nature of the social goods themselves might vary from location to location. Returning to the rich club (\fig{richClub}), the central clique might represent a network of content producers (e.g. radio content), while those at the periphery are consumers. If listeners donate money, then this scenario could be reasonably modeled using pf-goods in the center and ff- or pp-goods on the periphery. When those on the periphery produce pp- or ff-goods and those in the central clique produce pf-goods, a sufficient condition for all individuals to be better off is $b>c$ (provided $n\gg m$). It is perhaps unsurprising that in the context of \fig{richClub}, relative to when everyone produces a pp-good (resp. ff-good), it is generally easier (resp. harder) for producers to evolve when those in the center produce pf-goods with the same $b$ and $c$.

For ff-goods on a rich club, another natural question is whether the individuals in the central clique can scale up their contributions in order to create better outcomes for those at the periphery. In Methods, we show that if wealthy producers scale up their contributions in a way that ensures everyone in the population benefits, it is much harder for producers to evolve at all. Such a population leads to a trade-off between the following two scenarios: \emph{(i)} All individuals in the population produce the same total benefit at the same total cost. Producers easily evolve, leaving well-connected individuals wealthy at the expense of everyone else. \emph{(ii)} Well-connected producers ensure that each neighbor gets back what they contributed. Selection now opposes the spread of producers, leaving the population more often in the asocial (non-producer) state. In the SI, we also discuss how an institution can mitigate the harmful effects of certain prosocial behaviors (like the production of ff-goods), e.g. through a ``tax'' (\sfig{publicPool}).

\begin{figure}
	\centering
	\includegraphics[width=1.0\textwidth]{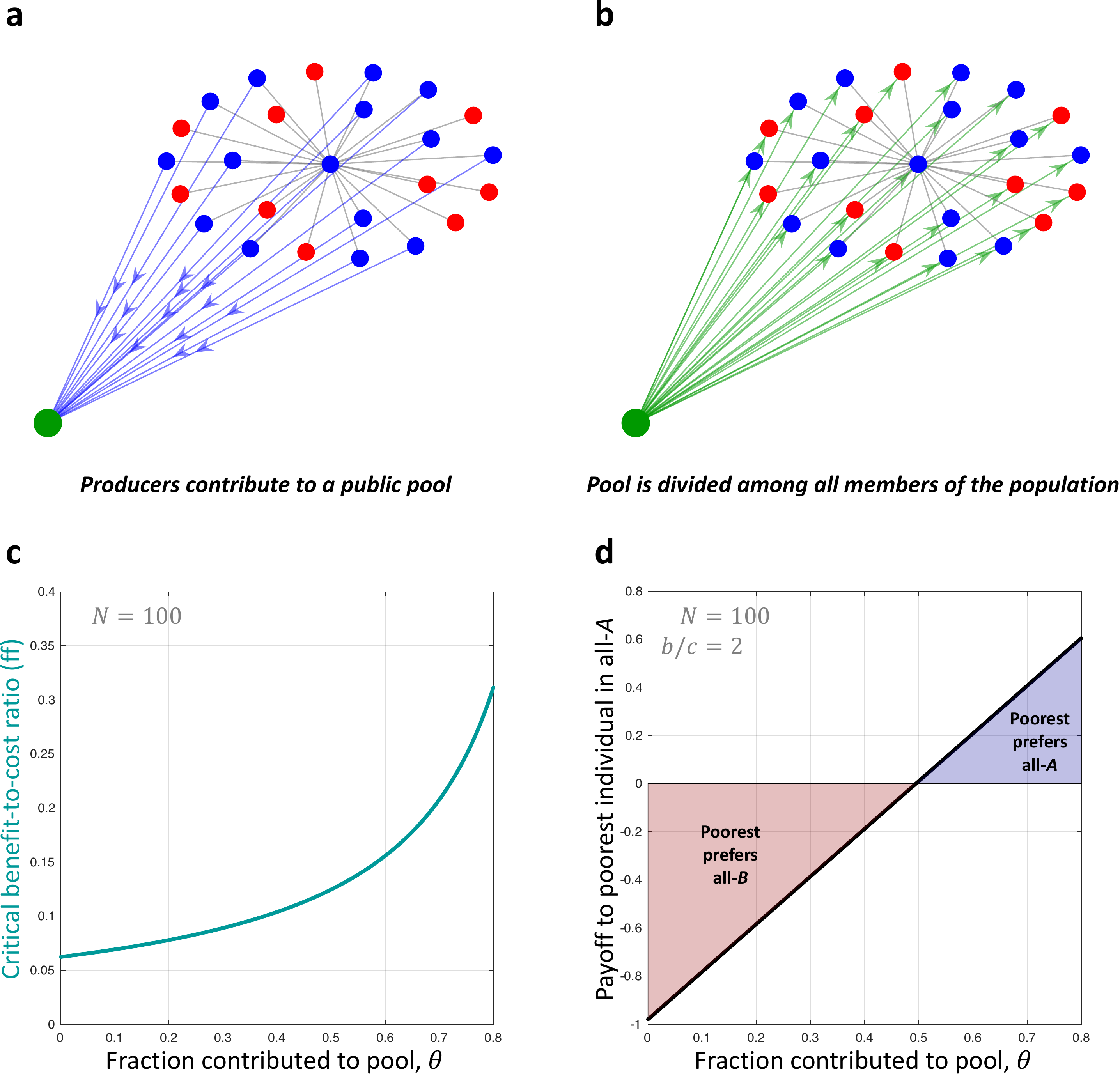}
	\caption{\textbf{Contributions by producers to a public pool can ameliorate payoff inequality.} For ff-goods, suppose that each producer (blue) donates $\theta b$ to a pool (green) and $\left(1-\theta\right) b$ to neighbors, \textbf{a}. If the total value of the public pool is divided among all members of the population (green arrows, \textbf{b}), then the situation can improve for those who are worst-off in the all-producer state. In particular, such a pool can result in a positive payoff to everyone in the population provided the contribution, quantified by $\theta$, is sufficiently large. The trade-off is that this pool also increases the critical benefit-to-cost ratio required for producers to evolve by a multiplicative factor of $1/\left(1-\theta\right)$ (see SI), illustrated in \textbf{c} on a star of size $N=100$ under PC updating. For this population structure, \textbf{d} depicts the payoff of the poorest individual (``leaf'' player, at the periphery of the star) in the prosocial (all-$A$) state as a function of the fraction contributed to the pool, $\theta$, when $b=2$ and $c=1$. This payoff is negative when $\theta\lessapprox 1/2$, which means that $99\%$ of the population is better off in the asocial (all-$B$) state. However, when $\theta\gtrapprox 1/2$, all individuals are better off when producers proliferate.\label{fig:publicPool}}
\end{figure}

Another form of asymmetry arises from stochasticity in the recipient of a donation. Instead of either producing a social good for each neighbor or dividing it up among the neighborhood, a producer might choose a single random neighbor as the recipient of the good in its entirety (\fig{concentrated}). Such a payoff scheme could be driven by indifference, meaning a donor does not care who receives the benefit, or by a mechanism external to the donor. For example, if an individual receives a request to participate in double-blind peer review, then this individual's donation, which is derived from their referee report, is conferred upon recipient(s) who are not directly chosen by the reviewer (the journal chooses). In the SI, we show that, under weak selection, randomly choosing a recipient is equivalent to dividing one's contribution among all possible recipients. This is because, when selection is weak, the conditions for success depend only on expected, rather than actual, payoffs. In particular, from the perspective of evolutionary dynamics, the stochastic donation scheme of \fig{concentrated}\textbf{b} is equivalent to that of ff-goods, \fig{concentrated}\textbf{a}, and producers can evolve even when $b<c$.

\begin{figure}
	\centering
	\includegraphics[width=0.9\textwidth]{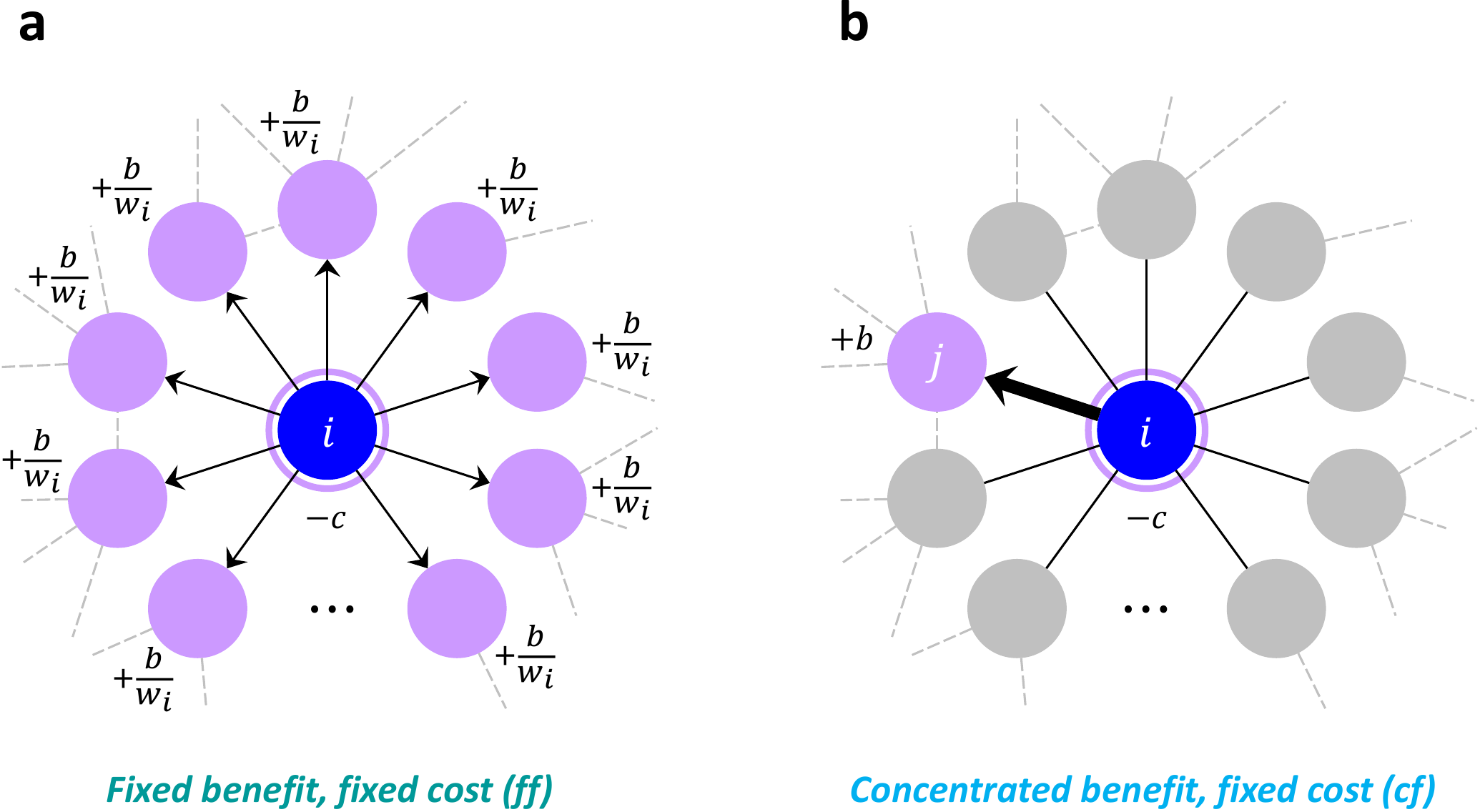}
	\caption{\textbf{Diffuse versus concentrated benefits.} When a producer, $i$ (center), has a rival good of total benefit $b$ and cost $c$ available to donate, two natural ways to distribute this good are \textbf{a} to all $w_{i}$ neighbors, divided evenly, and \textbf{b} to a single neighbor, $j$, chosen at random. The former describes the scheme of ff-goods. The latter represents a stochastic payoff scheme in which one lucky neighbor benefits from the good in its entirety. Remarkably, both methods result in identical critical benefit-to-cost ratios in the limit of weak selection. Therefore, the surprising results reported for ff-goods, such as the existence of critical ratios strictly less than one, also hold for this scheme. More generally, we show in the SI that any stochastic payoff scheme can be replaced by an ``equivalent'' deterministic scheme.\label{fig:concentrated}}
\end{figure}

\subsection{Reciprocity}
The behavioral types considered so far are quite simple: produce ($A$) or do not produce ($B$), unconditionally. When individuals have more than one chance to interact prior to an update to the population, more complex behavioral strategies can emerge. In the iterated prisoner's dilemma, an individual can punish past acts of defection and reward past acts of cooperation, and this mechanism of ``direct reciprocity'' is well-known for its ability to facilitate the emergence of cooperation \citep{trivers:TQRB:1971,nowak:Science:2006,press:PNAS:2012,stewart:PNAS:2013,stewart:SR:2016,hilbe:NHB:2018}.

There are many different ways to model reciprocity of social good production in heterogeneous populations. As a starting point, consider the strategy ``tit-for-tat'' (TFT), which cooperates (donates) in the first round and then subsequently copies what the opponent did in the previous round \citep{axelrod:BB:1984}. Let $B_{ij}$ and $C_{ij}$ be the benefit and cost of $i$ donating to $j$. In our model, an individual using TFT gives $B_{ij}$ to every $j$ (at a cost of $C_{ij}$) in the first round. In subsequent rounds, $i$ donates $B_{ij}$ to $j$ (still at cost $C_{ij}$) if and only if $j$ donated to $i$ in the previous round; otherwise, $i$ gives nothing to $j$ and pays no cost associated to $j$. Other individuals have no effect on $i$'s choice of whether to donate to $j$.

For pp-goods, this model gives the classical interpretation of TFT in the prisoner's dilemma, only now the two-player interactions are the pairwise encounters on a graph. For ff-goods, this model may be understood as follows. In the first round, each TFT player produces a good at cost $c$ and divides the benefit, $b$, among all neighbors. Subsequently, a TFT-player looks around and counts how many neighbors produced a good in the previous round. If a fraction, $x$, of one's neighbors produced a good, then in the subsequent round this TFT player produces a good of benefit $xb$ and cost $xc$ (i.e. a fraction $x$ of the original good). This benefit is divided among only those neighbors who produced a good in the last round. We avoid pf-goods here because the interpretation of ``reciprocity'' is much more nuanced for non-excludable goods. The benefits of a clean environment, for instance, normally cannot be denied to an individual.

We consider competition between the strategies ALLD and TFT, where ALLD (which stands for ``always defect'') is an unconditional non-producer. Since ALLD never produces a good, this iterated game exhibits quite straightforward behavior: In the first round, TFT is a producer and ALLD is a non-producer. In all subsequent rounds, TFT produces for, and donates to, only the other TFT players in their neighborhood. ALLD players get only the benefits they receive in the first round. In order to distinguish between present and future payoff streams, we use a discounting factor, $\lambda\in\left[0,1\right]$, which can be interpreted the probability of another encounter before the game ends. When $\lambda =0$, we recover the original model of producers versus non-producers, with no reciprocity. When $\lambda =1$, the time horizon of the game is infinite (``undiscounted''). For all values of $\lambda$ strictly between $0$ and $1$, the game is finite with $1/\left(1-\lambda\right)$ rounds, on average.

For any $\lambda\in\left[0,1\right]$, the selection condition for TFT to be favored over ALLD is
\begin{align}
\sum_{i,j,k,\ell =1}^{N} &\pi_{i} m_{k}^{ji} \left( - \left(x_{jk}+\lambda x_{j\ell}\right) C_{k\ell} + \left(x_{j\ell}+\lambda x_{jk}\right) B_{\ell k} \right) \nonumber \\
&> \sum_{i,j,k,\ell =1}^{N} \pi_{i} m_{k}^{ji} \left( - \left(x_{ik}+\lambda x_{i\ell}\right) C_{k\ell} + \left(x_{i\ell}+\lambda x_{ik}\right) B_{\ell k} \right) . \label{eq:generalConditionReciprocity}
\end{align}
When $\lambda =0$, \eq{generalConditionReciprocity} reduces to \eq{generalConditionSG}. We derive this result in the SI and provide simplified formulas for PC, DB, and IM updating in Methods. \fig{reciprocity} illustrates the effects of increasing the time horizon on the critical ratio for TFT to evolve. In each case, reciprocity lowers the threshold for the evolution of producers. We observe, however, that payoffs in the all-TFT state are the same as those in the all-producer state. Therefore, the potential for wealth-reducing and/or harmful prosociality is not eliminated by reciprocity. On the contrary, reciprocity can enable such outcomes to arise under an expanded range of conditions. Overall, reciprocity typically facilitates the evolution of prosocial behaviors, which may be either helpful or harmful to the population at large and/or to the least well-off. For example, on the coauthorship network of \fig{netscience}, which corresponds to $\lambda =0$, ff-goods must be wealth-producing for producers to evolve. But reciprocity, in the form of sufficiently large $\lambda >0$, can support the evolution of producers even when the underlying social good is wealth-decreasing ($b<c$).

\begin{figure}
	\centering
	\includegraphics[width=1.0\textwidth]{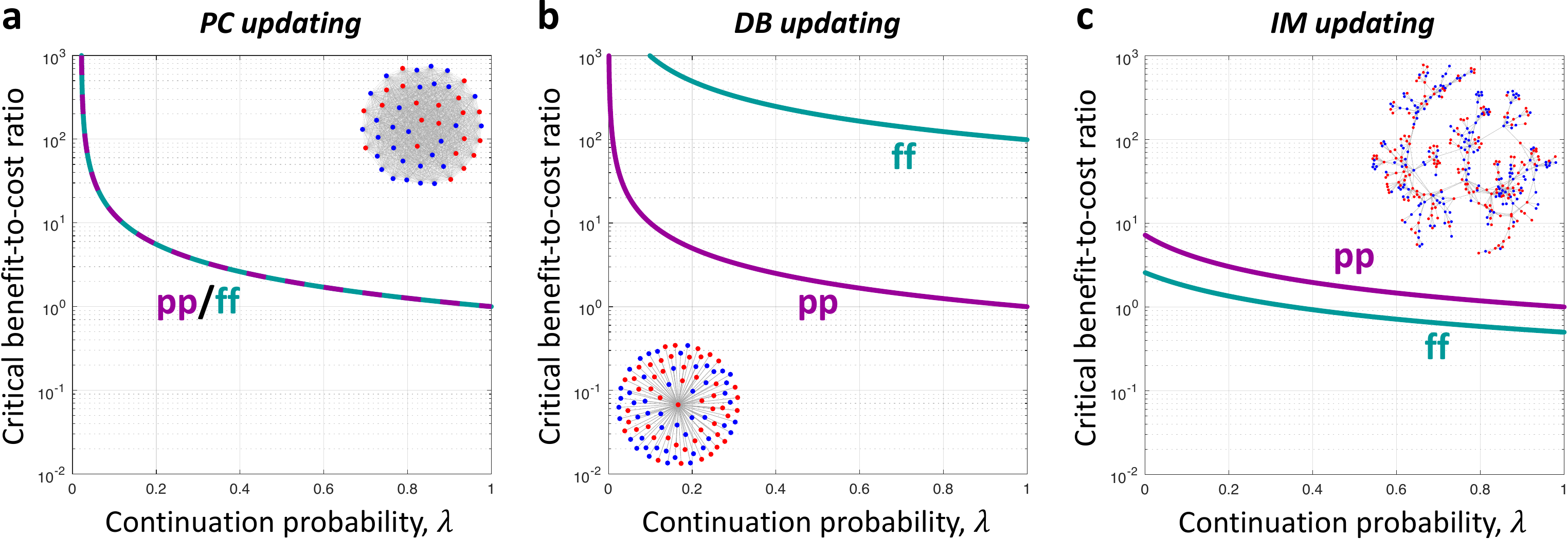}
	\caption{\textbf{Effects of reciprocity on the selection of prosocial behaviors.} Sufficiently large continuation probabilities can facilitate the evolution of TFT by reducing the critical benefit-to-cost ratio. \textbf{a}, Whereas producers can never evolve on a homogeneous graph under PC updating when there is just one interaction ($\lambda =0$), reciprocity can favor producers (TFT) even on well-mixed structures. \textbf{b}, The star never supports producers under DB updating, but for non-zero continuation probabilities, $\lambda$, TFT can be favored as long as $b/c$ is sufficiently large. \textbf{c}, On the empirical coauthorship network considered in \fig{netscience}, increasing $\lambda$ further promotes the spread of prosocial traits with low benefit-to-cost ratios under IM updating. In particular, since payoffs in the all-TFT state are identical to those of the all-producer state, reciprocity for ff-goods can cause worse outcomes for the poorest in the population.\label{fig:reciprocity}}
\end{figure}

There is one particular condition under which reciprocity cannot be harmful for a population. We say that a prosocial behavior is ``pairwise mutually beneficial'' (PMB) if $B_{ji}\geqslant C_{ij}$ for all $i$ and $j$, with strict inequality for at least one pair $\left(i,j\right)$. In other words, a behavior is PMB if and only if for each pair of individuals expressing this behavior, both partners receive at least as much as they pay (for that particular interaction). PMB behaviors cannot be wealth-decreasing, nor can they lead to negative payoffs in the all-producer state. The simplest example of a PMB behavior is the production of pp-goods when $b>c$. However, whereas the production of pp-goods when $b>c$ is PMB on any graph, in general whether a behavior is PMB depends on both the good and the population structure. For example, when $b>c$, ff-goods are always PMB on regular graphs, but not necessarily on heterogeneous networks such as the rich club. PMB behaviors are not always favored by selection, as was demonstrated in our baseline model of one-shot interactions. However, we show in the SI that, for any given PMB behavior, TFT is favored for all $\lambda$ greater than some threshold value $\lambda^{*}<1$.

\subsection{Stronger selection}
The analytical conditions presented so far (e.g. \eq{generalConditionSG}) are valid under the assumption that selection is weak. In the case of PC updating, this assumption means that if $i$ has payoff $u_{i}$ and $j$, a neighbor of $i$, has payoff $u_{j}$, and if $j$ is chosen as a model individual for comparison to $i$, then $i$ imitates the behavioral type of $j$ with probability $\left(1+e^{-\delta\left(u_{j}-u_{i}\right)}\right)^{-1}$, where $\delta$ is positive but sufficiently small. A natural question arising here is what happens when $\delta$ is not necessarily small, which is relevant when individuals are highly inclined to imitate a neighbor with a larger payoff.

Analytical calculations for stronger selection (larger $\delta$) quickly become infeasible for arbitrary heterogeneous graphs \cite{ibsenjensen:PNAS:2015}, but remain tractable for structures with a high degree of symmetry, such as the star \citep{hadjichrysanthou:DGA:2011}. In \fig{strongerSelection}, we consider ff-goods on a star graph of size $N=500$. We find that stronger selection can promote the evolution of producers significantly, both when the total benefits exceed the total costs and vice versa. We can quantify selection for producers by the fraction of time the population spends in the all-producer state under rare mutation, which is $\rho_{A}/\left(\rho_{A}+\rho_{B}\right)$. When $b=1/10$ and $c=1$, the star spends approximately $64\%$ of the time in the all-producer state at its peak (\fig{strongerSelection}\textbf{c}). When $b=2$ and $c=1$, this number is slightly larger at approximately $66\%$ (\fig{strongerSelection}\textbf{d}). In both cases, the all-producer state is highly unequal, with a large, positive payoff for the individual at the center of the star and negative payoffs for the $N-1=499$ remaining individuals at the periphery.

\begin{figure}
	\centering
	\includegraphics[width=0.9\textwidth]{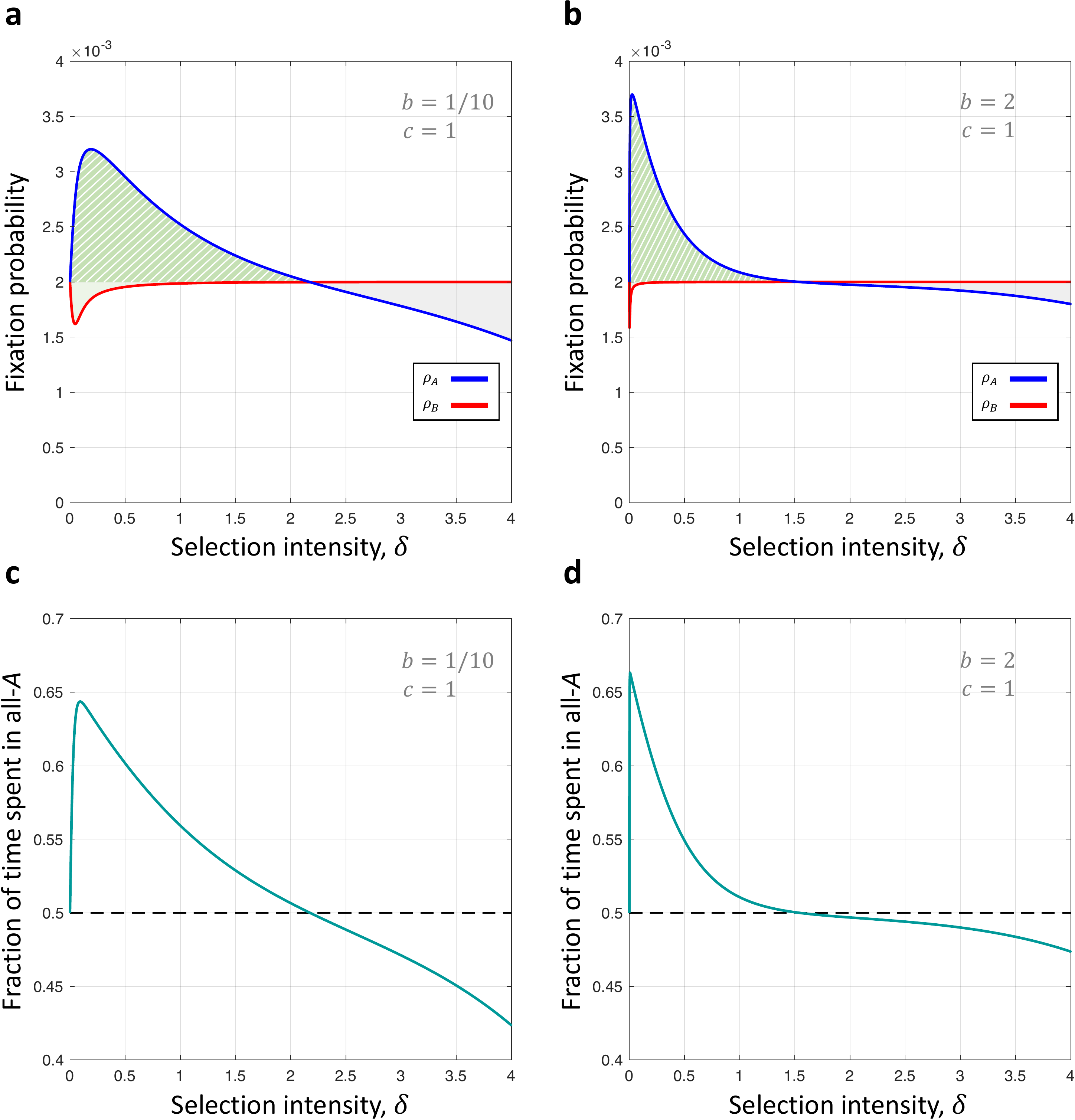}
	\caption{\textbf{Beyond weak selection on the star graph.} Increasing the selection intensity can magnify the wealth inequality due to prosocial behaviors. We illustrate this phenomenon using PC updating on a star graph of size $N=500$ (i.e. rich club with $m=1$ and $n=499$), with ff-goods. Panels \textbf{a} and \textbf{c} differ from \textbf{b} and \textbf{d} only in the benefit, $b$, which is $1/10$ in the former and $2$ in the latter. The cost is $c=1$ in all panels. For both of these cases, our theoretical results imply that increasing the selection strength from $\delta =0$ favors the evolution of producers. As is illustrated, this effect remains for a significant range of selection intensities, $\delta$. Moreover, at its peak, the amount of time the population spends in the all-producer state is approximately $64\%$ when $b=1/10$ (\textbf{c}) and approximately $66\%$ when $b=2$ (\textbf{d}). When $b=1/10$, the hub gets a payoff of $48.9$, while individuals at the periphery get a payoff of approximately $-1$. When $b=2$, the hub gets $997$ and individuals at the periphery again get roughly $-1$. Thus, natural selection of prosocial behaviors can result in extremely bad outcomes for large portions of a population.\label{fig:strongerSelection}}
\end{figure}

\section{Discussion}
A large and growing body of research \citep{nowak:Nature:1992,taylor:EE:1992,wilson:EE:1992,mitteldorf:JTB:2000,irwin:TPB:2001,santos:PRL:2005,antal:PRL:2006,ohtsuki:Nature:2006,grafen:JEB:2007,taylor:Nature:2007,gomez:PRL:2007,lion:EL:2008,sood:PRE:2008,santos:Nature:2008,tarnita:PNAS:2009,nowak:PTRSB:2009,cao:PA:2010,chen:AAP:2013,rand:PNAS:2014,pinheiro:PLOSCB:2014,debarre:NC:2014,maciejewski:PLoSCB:2014,pena:PLOSCB:2016,fan:PA:2017,debarre:JTB:2017,allen:Nature:2017,su:PLOSCB:2019,su:PRSB:2019,allen:NC:2019} has shown that spatial structure can promote the evolution of cooperative or prosocial behaviors. Our work, while affirming this principle, reveals it to be more complicated than it may seem. First, we show that the conditions for prosocial behaviors to evolve depend crucially on how the costs and benefits are distributed. This, in turn, depends on whether the goods produced are rival or non-rival, and whether the cost to produce scales with the number of recipients. To link these economic concepts with the evolutionary dynamics literature, we have introduced three natural schemes (ff, pp, and pf) for the production and sharing of goods on networks. Second, and more strikingly, the ``prosocial'' outcome may not be socially optimal. Selection can favor outcomes in which all individuals contribute to the social good, but some expend more in costs than they receive in benefits, leading to negative payoffs. Additionally, there are structures on which production of ff-goods can evolve when even they are a net detriment to the population ($b<c$).

Qualitatively, these results hold for all update rules considered here; the primary differences among these update rules are in the precise population structures for which the production of goods is favored (see Extended Data Figures and SI). For example, homogeneous graphs can support the evolution of producers under DB and IM updating but not under PC updating. Some heterogeneous structures can favor producers under one update rule but not another. These differences are expected given that different update rules describe different evolutionary dynamics. Our general result (\eq{generalConditionSG}) is not restricted to any particular update rule or sharing scheme, and it can be applied to many other populations.

One of the main limitations of our study is its restriction to populations of fixed size and structure. Changes to the population size and/or structure as the population evolves could lead to additional interesting behavior with respect to social goods (e.g. population growth rates or the possibility of extinction). Our analysis also focuses on relatively simple social goods (ff, pp, and pf), which represent idealized versions of what might arise in a real population. Congestible goods, for example, are non-rival when there are few consumers but become rival when many stand to benefit from them. A good could also be anti-rival, with the per-capita benefit being an increasing function of the total number of recipients. We see all of these possibilities as promising topics for future studies on the evolutionary dynamics of social goods.

Reciprocity is another mechanism that is known to promote the evolution of cooperation \citep{trivers:TQRB:1971,nowak:Science:2006,press:PNAS:2012,stewart:PNAS:2013,stewart:SR:2016,hilbe:NHB:2018}, which holds in our model as well. In the context of repeated interactions, reciprocal prosocial behaviors can evolve in under a broader range of conditions than their unconditional counterparts. However, this also means that repeated interactions can enable the evolution of harmful and/or wealth-reducing prosociality in conditions where it would not evolve for one-shot interactions. While our model touches upon one aspect of reciprocity in the context of social goods, we also view this topic as an important area for future research.

Our work raises thorny questions about the nature and consequences of cooperative or prosocial behaviors. Under what circumstances should we consider such traits desirable? If they increase a population's total wealth? If they make everyone better off? If they distribute wealth evenly? In general, these conditions are distinct, and each must be considered in the context of the underlying population structure.

In summary, we have shown that heterogeneous population structures act as strong promoters of the evolution of prosocial behaviors. However, the resulting prosocial behaviors can lead to payoff distributions in which a few highly-connected nodes accumulate much of the total wealth \citep{santos:Nature:2008}, while poorly-connected nodes end up being harmed. When the population structure is interpreted as describing informal social ties within a group, these examples may be seen as instances of the ``tyranny of structurelessness'' \citep{freeman:WSQ:2013}. In particular, the absence of a formal system of governance can lead to situations in which some (or many) in a group are worse off. While the impact of institutions \citep{zhang:EE:2013} on evolutionary dynamics is a deep topic, our results provide insight into when an intervention might be necessary. These outcomes call for the design of mechanisms to redistribute wealth in order to maintain a stable society, which engages in and benefits from prosocial behaviors.

\section{Methods}

\subsection{Modeling evolutionary dynamics}
We model a general evolutionary process in a population of finite size, $N$, using the notion of a replacement rule \citep{allen:JMB:2014,allen:JMB:2019}. If the process is in state $\vx\in\left\{0,1\right\}^{N}$ at a given time step, where $x_{i}=1$ indicates that $i$ has type $A$ and $x_{i}=0$ indicates that $i$ has type $B$, then a replacement event $\left(R,\alpha\right)$ is chosen with probability $p_{\left(R,\alpha\right)}\left(\mathbf{x}\right)$. A replacement event consists of a set of individuals to be replaced, $R\subseteq\left\{1,\dots ,N\right\}$, and a parentage map, $\alpha :R\rightarrow\left\{1,\dots ,N\right\}$, where, for $i\in R$, $\alpha\left(i\right) =j$ indicates that $i$ is replaced by the offspring of $j$. The distribution $\left\{p_{\left(R,\alpha\right)}\left(\mathbf{x}\right)\right\}_{\left(R,\alpha\right)}$ defines the replacement rule and specifies the process driving evolution in the population.

For a fixed payoff scheme, the total payoff to $i$ in state $\vx\in\left\{0,1\right\}^{N}$ is
\begin{align}
u_{i}\left(\mathbf{x}\right) &= \sum_{j=1}^{N} \left( - x_{i} C_{ij} + x_{j} B_{ji} \right) ,
\end{align}
where $C_{ij}$ is the cost $i$ pays to donate to $j$ when $i$ is a producer and $B_{ji}$ is the benefit $j$ provides to $i$ when $j$ is a producer. This payoff is converted to relative fecundity via the formula $F_{i}\left(\mathbf{x}\right) =\exp\left\{\delta u_{i}\left(\mathbf{x}\right)\right\}$, which is then used to determine the probability of choosing replacement event $\left(R,\alpha\right)$ in state $\mathbf{x}$. For example, suppose that $\left(w_{ij}\right)_{i,j=1}^{N}$ is the adjacency matrix for an undirected, unweighted graph. Let $p_{ij}\coloneqq w_{ij}/w_{i}$ be the one-step probability of moving from $i$ to $j$ in a random walk on this graph, where $w_{i}\coloneqq\sum_{j=1}^{N}w_{ij}$ is the number of individuals neighboring $i$. Under pairwise-comparison (PC) updating, the probability of choosing $\left(R,\alpha\right)$ is
\begin{align}
p_{\left(R,\alpha\right)}\left(\mathbf{x}\right) &=
\begin{cases}
\frac{1}{N} p_{i\alpha\left(i\right)} \frac{F_{\alpha\left(i\right)}\left(\mathbf{x}\right)}{F_{i}\left(\mathbf{x}\right) +F_{\alpha\left(i\right)}\left(\mathbf{x}\right)} & R=\left\{i\right\}\textrm{ for some }i\in\left\{1,\dots ,N\right\} ,\ \alpha\left(i\right)\neq i , \\
& \\
\frac{1}{N} \sum_{j=1}^{N} p_{ij} \frac{F_{i}\left(\mathbf{x}\right)}{F_{i}\left(\mathbf{x}\right) +F_{j}\left(\mathbf{x}\right)} & R=\left\{i\right\}\textrm{ for some }i\in\left\{1,\dots ,N\right\} ,\ \alpha\left(i\right) =i , \\
& \\
0 & \textrm{otherwise} .
\end{cases}
\end{align}
Similarly, the probability of replacement event $\left(R,\alpha\right)$ is
\begin{align}
p_{\left(R,\alpha\right)}\left(\vx\right) &=
\begin{cases}
\frac{1}{N} \frac{w_{i\alpha\left(i\right)}F_{\alpha\left(i\right)}\left(\vx\right)}{\sum_{k=1}^{N}w_{ik}F_{k}\left(\vx\right)} & R=\left\{i\right\}\textrm{ for some }i\in\left\{1,\dots ,N\right\} , \\
& \\
0 & \textrm{otherwise}
\end{cases}
\end{align}
under DB updating and
\begin{align}
p_{\left(R,\alpha\right)}\left(\vx\right) &=
\begin{cases}
\frac{1}{N} \frac{w_{i\alpha\left(i\right)}F_{\alpha\left(i\right)}\left(\vx\right)}{F_{i}\left(\vx\right) +\sum_{k=1}^{N}w_{ik}F_{k}\left(\vx\right)} & R=\left\{i\right\}\textrm{ for some }i\in\left\{1,\dots ,N\right\} ,\ \alpha\left(i\right)\neq i , \\
& \\
\frac{1}{N}\frac{F_{i}\left(\vx\right)}{F_{i}\left(\vx\right) +\sum_{k=1}^{N}w_{ik}F_{k}\left(\vx\right)} & R=\left\{i\right\}\textrm{ for some }i\in\left\{1,\dots ,N\right\} ,\ \alpha\left(i\right) = i , \\
& \\
0 & \textrm{otherwise}
\end{cases}
\end{align}
under IM updating.

\subsection{Fixation probabilities, transient dynamics, and the selection condition}
We assume that there is at least one individual who can generate a lineage that takes over the population. As a result, the population must eventually end up in one of the two monomorphic states, all-$A$ or all-$B$. Let $\rho_{A}^{i}$ (resp. $\rho_{B}^{i}$) be the probability that a single $A$ (resp. $B$), placed initially at location $i$, takes over a population of type $B$ (resp. type $A$). The mean fixation probabilities for the two types are then $\rho_{A}=\frac{1}{N}\sum_{i=1}^{N}\rho_{A}^{i}$ and $\rho_{B}=\frac{1}{N}\sum_{i=1}^{N}\rho_{B}^{i}$, respectively.

By analyzing the demographic variables (e.g. birth rates and death probabilities) resulting from this process, together with the transient dynamics (i.e. prior to absorption in all-$A$ or all-$B$), we derive a condition for $\rho_{A}>\rho_{B}$ under weak selection ($\delta\ll 1$). Specifically, if \emph{(i)} $\pi_{i}$ is the reproductive value of vertex $i$; \emph{(ii)} $x_{ij}$ is the probability that $i$ and $j$ have the same type in the neutral process (with respect to a distribution described in the SI); and \textit{(iii)} $m_{k}^{ij}$ is the marginal effect of $k$'s fecundity on $i$ replacing $j$, then $\rho_{A}>\rho_{B}$ under weak selection if and only if \eq{generalConditionSG} holds. Moreover, this condition can be evaluated by solving a linear system of $O\left(N^{2}\right)$ equations. The details of this derivation, including a more formal description of $x_{ij}$, may be found in the SI.

\subsubsection{PC updating}
For pairwise-comparison (PC) updating, we have $\rho_{A}>\rho_{B}$ under weak selection if and only if
\begin{align}
\sum_{i=1}^{N} \pi_{i} \sum_{\ell =1}^{N}\left( -x_{ii}C_{i\ell} + x_{i\ell}B_{\ell i} \right) &> \sum_{i,j=1}^{N} \pi_{i} p_{ij}\sum_{\ell =1}^{N}\left( -x_{ij}C_{j\ell} + x_{i\ell}B_{\ell j} \right) . \label{eq:mainConditionPC_mainText}
\end{align}
Informally, this condition says that the expected payoff to a producer, $i$, chosen with probability $\pi_{i}$ (which in this case is $w_{i}/\sum_{j=1}^{N}w_{j}$), exceeds that of a random neighbor. In the SI, we show that we can evaluate \eq{mainConditionPC_mainText} by replacing $1-x_{ij}$ by $\tau_{ij}$, where $\tau_{ii}=0$ for every $i$ and
\begin{align}
\tau_{ij} &= 1+\frac{1}{2}\sum_{k=1}^{N}p_{ik}\tau_{kj}+\frac{1}{2}\sum_{k=1}^{N}p_{jk}\tau_{ik}
\end{align}
whenever $i\neq j$. Finally, to see what \eq{mainConditionPC_mainText} looks like for a given kind of donation, we just need to give formulas for $B_{ij}$ and $C_{ij}$. For pp-goods, we have $B_{ij}=w_{ij}b$ and $C_{ij}=w_{ij}c$. For ff-goods, $B_{ij}=w_{ij}b/w_{i}$ and $C_{ij}=w_{ij}c/w_{i}$. Finally, for pf-goods, $B_{ij}=w_{ij}b$ and $C_{i}=w_{ij}c/w_{i}$. In each case, \eq{mainConditionPC_mainText} can be expressed as $\gamma b>\beta c$, where both $\beta$ and $\gamma$ are independent of $b$ and $c$, and $\beta >0$. With $b,c>0$, a necessary condition for producers to be favored is thus $\gamma >0$. Writing $B_{ij}=b\beta_{ij}$ and $C_{ij}=c\gamma_{ij}$, we see that when $\gamma >0$, the selection condition is
\begin{align}
\rho_{A}>\rho_{B} &\iff \frac{b}{c}>\left(\frac{b}{c}\right)^{\ast} \coloneqq \frac{\beta}{\gamma} =  \frac{\sum_{i,j=1}^{N} \pi_{i} p_{ij} \tau_{ij} \sum_{\ell =1}^{N} \gamma_{j\ell}}{\sum_{i,j=1}^{N} \pi_{i} p_{ij} \sum_{\ell =1}^{N}\left( \tau_{i\ell} - \tau_{j\ell} \right) \beta_{\ell j}} . \label{eq:tauConditionPC_mainText}
\end{align}

On a homogeneous graph, one can show that $\gamma <0$ (see SI), which means that such a structure cannot support the evolution of producers under PC updating. In a prior study, it was shown that on large homogeneous graphs, the evolution of cooperation was possible under PC updating only in the presence of game transitions \citep{su:PNAS:2019} (in which case the interaction was essentially a coordination game). In contrast, on any homogeneous graph without game transitions, producers (including cooperators) cannot be favored under PC updating. However, as shown in the text, there are many examples of heterogeneous population structures with $\gamma >0$.

\subsubsection{DB updating}
For death-birth (DB) updating, we have $\rho_{A}>\rho_{B}$ under weak selection if and only if
\begin{align}
\sum_{i=1}^{N} \pi_{i} \sum_{\ell =1}^{N}\left( -x_{ii}C_{i\ell} + x_{i\ell}B_{\ell i} \right) &> \sum_{i,j=1}^{N} \pi_{i} p_{ij}^{\left(2\right)} \sum_{\ell =1}^{N}\left( -x_{ij}C_{j\ell} + x_{i\ell}B_{\ell j} \right) . \label{eq:mainConditionDB_mainText}
\end{align}
The interpretation of this condition is analogous to that of the condition for PC updating, except here the comparison is between a producer and a two-step (rather than one-step) neighbor.

Using the same values of $\pi_{i}$ and $\tau_{ij}$ described above for PC updating, this condition becomes
\begin{align}
\rho_{A}>\rho_{B} &\iff \frac{b}{c}>\left(\frac{b}{c}\right)^{\ast} \coloneqq \frac{\beta}{\gamma} = \frac{\sum_{i,j=1}^{N} \pi_{i} p_{ij}^{\left(2\right)} \tau_{ij} \sum_{\ell =1}^{N} \gamma_{j\ell}}{\sum_{i,j=1}^{N} \pi_{i} p_{ij}^{\left(2\right)} \sum_{\ell =1}^{N}\left( \tau_{i\ell} - \tau_{j\ell} \right) \beta_{\ell j}} . \label{eq:tauConditionDB_mainText}
\end{align}

\subsubsection{IM updating}
For both PC and DB updating, the reproductive value, $\pi_{i}$, random walk step-probability, $p_{ij}$, and recurrence relation for $\tau_{ij}$ (which is a proxy for $x_{ij}$) were the same in both processes. For imitation (IM) updating, these quantities need to be modified slightly. Let $\widetilde{w}$ be matrix obtained from $w$ by adding a self-loop to each vertex, i.e. $\widetilde{w}_{ij}=w_{ij}$ for $i\neq j$ and $\widetilde{w}_{ii}=1$ for all $i=1,\dots ,N$. Let $\widetilde{p}_{ij}\coloneqq\widetilde{w}_{ij}/\widetilde{w}_{i}$ be the probability of transitioning from $i$ to $j$ in one step of a random walk on this modified graph. The reproductive value of $i$ under IM updating is $\pi_{i}=\widetilde{w}_{i}/\sum_{j=1}^{N}\widetilde{w}_{j}$, where $\widetilde{w}_{i}=1+w_{i}$. Under weak selection, we have $\rho_{A}>\rho_{B}$ if and only if
\begin{align}
\sum_{i=1}^{N}\pi_{i}\sum_{\ell =1}^{N}\left( -x_{ii}C_{i\ell} + x_{i\ell}B_{\ell i} \right) &> \sum_{i,j=1}^{N}\pi_{i} \widetilde{p}_{ij}^{\left(2\right)}\sum_{\ell =1}^{N}\left( -x_{ij}C_{j\ell} + x_{i\ell}B_{\ell j} \right) . \label{eq:mainConditionIM_mainText}
\end{align}

The terms $x_{ij}$ are again different here than they were for PC and DB updating. However, interpreting $x_{ij}$ in the same way as before (as the probability that $i$ and $j$ have the same type with respect to a certain distribution in the neutral process), \eq{mainConditionIM_mainText} has a natural interpretation as a comparison between the expected payoff to a producer (left-hand side) and the expected payoff to a two-step neighbor of a producer (right-hand side). To evaluate this condition, we can replace $1-x_{ij}$ by $\tau_{ij}$, where $\tau_{ii}=0$ for every $i$ and, whenever $i\neq j$, $\tau_{ij}$ satisfies the recurrence
\begin{align}
\tau_{ij} &= 1+\frac{1}{2}\sum_{k=1}^{N}\widetilde{p}_{ik}\tau_{kj}+\frac{1}{2}\sum_{k=1}^{N}\widetilde{p}_{jk}\tau_{ik} .
\end{align}
For $B_{ij}=b\beta_{ij}$ and $C_{ij}=c\gamma_{ij}$, the selection condition when $\gamma >0$ is
\begin{align}
\rho_{A}>\rho_{B} &\iff \frac{b}{c}>\left(\frac{b}{c}\right)^{\ast} \coloneqq \frac{\beta}{\gamma} = \frac{\sum_{i,j=1}^{N} \pi_{i} \widetilde{p}_{ij}^{\left(2\right)} \tau_{ij} \sum_{\ell =1}^{N} \gamma_{j\ell}}{\sum_{i,j=1}^{N} \pi_{i} \widetilde{p}_{ij}^{\left(2\right)} \sum_{\ell =1}^{N}\left( \tau_{i\ell} - \tau_{j\ell} \right) \beta_{\ell j}} . \label{eq:tauConditionIM_mainText}
\end{align}

\subsection{Asymmetric games}
Here, we consider the effects of asymmetric payoffs on rich club graphs with $m$ central nodes and $n$ peripheral nodes. We are especially interested in the case in which $m$ is fixed and $n$ grows.

Suppose that a producer on the periphery produces a good of cost $c_{\textrm{p}}$ and divides the benefit, $b_{\textrm{p}}$, among all neighbors. A producer in the center group pays $c_{\textrm{r}}$ and distributes the benefit, $b_{\textrm{r}}$, evenly. We consider the case in which $b_{\textrm{p}}=b$ and $c_{\textrm{p}}=c$, and that $b_{\textrm{r}}=s\left(m,n\right) b_{\textrm{p}}$ and $c_{\textrm{r}}=s\left(m,n\right) c_{\textrm{p}}$ for some function $s\left(m,n\right)$. If everyone on the rich club is a producer (whose state in $\left\{0,1\right\}^{N}$ is denoted $\vA$), then the payoff to individual $i$ is
\begin{align}
u_{i}\left(\vA\right) &= 
\begin{cases}
-s\left(m,n\right) c+\frac{m-1}{m+n-1} s\left(m,n\right) b+\frac{n}{m} b & i\textrm{ is in the central rich club} , \\
& \\
-c+\frac{m}{m+n-1} s\left(m,n\right) b & i\textrm{ is on the periphery} .
\end{cases}\label{eq:asymmetricPayoffs}
\end{align}
For a peripheral individual's payoff to remain non-negative as $n\rightarrow\infty$, $s\left(m,n\right)$ must grow at least linearly in $n$. At the same time, if $\lim_{n\rightarrow\infty}\frac{1}{n}s\left(m,n\right) =\infty$, then the payoff to an individual in the central rich club will eventually become negative as $n\rightarrow\infty$. Therefore, we consider $s\left(m,n\right) =k_{1}\left(m\right) n+k_{0}\left(m\right)$ for some functions $k_{0}$ and $k_{1}$ of $m$. Letting $n\rightarrow\infty$ gives
\begin{align}
\lim_{n\rightarrow\infty} u_{i}\left(\vA\right) &= 
\begin{cases}
-\infty & i\textrm{ is in the central rich club},\ \frac{b}{c}<mk_{1}\left(m\right) , \\
& \\
0 & i\textrm{ is in the central rich club},\ \frac{b}{c}=mk_{1}\left(m\right) , \\
& \\
+\infty & i\textrm{ is in the central rich club},\ \frac{b}{c}>mk_{1}\left(m\right) , \\
& \\
-c+mk_{1}\left(m\right) b & i\textrm{ is on the periphery} .
\end{cases}
\end{align}

For the first payoff in \eq{asymmetricPayoffs} to stay non-negative as $n\rightarrow\infty$, we require $b/c\geqslant mk_{1}\left(m\right)$. But we also want the second payoff in \eq{asymmetricPayoffs} to be non-negative, i.e. $b/c\geqslant 1/\left(mk_{1}\left(m\right)\right)$. Thus, we seek $k_{1}$ with
\begin{align}
\frac{c}{b} \leqslant mk_{1}\left(m\right) \leqslant \frac{b}{c} .
\end{align}
Such a $k_{1}$ exists if and only if $b\geqslant c$, in which case we can set $k_{0}\left(m\right) =0$ and $k_{1}\left(m\right) =1/m$. For large $n$, it follows that each producer in the central clique gives an average of $b/m$ to each neighbor at a cost of $c/m$. The per-neighbor benefit and cost values are the same ($b/m$ and $c/m$, respectively) for each peripheral individual as well, which effectively transforms the payoff structure into that of a pp-good with benefit-to-cost ratio $\left(b/m\right) /\left(c/m\right) =b/c$, and we already know that it is much more difficult for producers of such a good to evolve (if they can at all).

\subsection{Reciprocity}
For the model of reciprocity defined in the text, we let $A$ and $B$ denote the strategies TFT and ALLD, respectively. Let $u_{i}^{t}$ be the payoff to player $i$ in the $t$th round of the game, i.e.
\begin{align}
u_{i}^{t}\left(\mathbf{x}\right) &= 
\begin{cases}
\sum_{j=1}^{N} \left( - x_{i} C_{ij} + x_{j} B_{ji} \right) & t = 1 , \\
& \\
\sum_{j=1}^{N} x_{i}x_{j} \left( - C_{ij} + B_{ji} \right) & t > 1 .
\end{cases}
\end{align}
If the discounting factor (continuation probability) is $\lambda\in\left[0,1\right)$, then the overall payoff to $i$ is
\begin{align}
u_{i}\left(\mathbf{x}\right) &= \left(1-\lambda\right) \left( \sum_{j=1}^{N} \left( - x_{i} C_{ij} + x_{j} B_{ji} \right) + \sum_{t=1}^{\infty}\lambda^{t} \sum_{j=1}^{N} x_{i}x_{j} \left( - C_{ij} + B_{ji} \right) \right) \nonumber \\
&= \left(1-\lambda\right) \sum_{j=1}^{N} \left( - x_{i} C_{ij} + x_{j} B_{ji} \right) + \lambda \sum_{j=1}^{N} x_{i}x_{j} \left( - C_{ij} + B_{ji} \right) \nonumber \\
&= \sum_{j=1}^{N} \left( -\left( 1-\lambda + \lambda x_{j}\right) x_{i} C_{ij} + \left(1-\lambda + \lambda x_{i}\right) x_{j} B_{ji} \right) .
\end{align}

Writing down the payoffs when $i$ uses TFT ($x_{i}=1$) and ALLD ($x_{i}=0$) separately gives
\begin{align}
u_{i}\left(\mathbf{x}\right) &= 
\begin{cases}
\sum_{j=1}^{N} 	\left( -C_{ij} + x_{j}B_{ji} \right) + \lambda \sum_{j=1}^{N} \left(1-x_{j}\right) C_{ij} & x_{i} = 1 , \\
& \\
\left(1-\lambda \right) \sum_{j=1}^{N} x_{j} B_{ji} & x_{i} = 0 .
\end{cases}
\end{align}
It follows that, in each state $\vx\in\left\{0,1\right\}^{N}$, increasing $\lambda$ does not decrease the payoff to an $A$ and does not increase the payoff to a $B$. For any reasonable process favoring individuals with higher payoffs (including PC, DB, and IM updating), it follows that $\rho_{A}$ is monotonically increasing and $\rho_{B}$ is monotonically decreasing in $\lambda$. Furthermore, we note that when $\lambda =1$, $i$ gets $0$ when using ALLD and $\sum_{j=1}^{N}x_{j}\left( -C_{ij} + B_{ji} \right)$ when using TFT. If the underlying behavior is pairwise mutually beneficial (PMB), then in every state each ALLD has a payoff of zero and each TFT has a payoff of at least zero. Therefore, no reasonable process favoring traits with higher payoffs should disfavor TFT relative to ALLD when the interactions have an infinite time horizon.

The selection condition for ALLD versus TFT, \eq{generalConditionReciprocity}, is derived in the SI. Here, we explore how this condition can be evaluated for the update rules considered in the text.

\subsubsection{PC updating}
For PC updating, the condition for selection to favor TFT relative to ALLD is
\begin{align}
\rho_{A}>\rho_{B} \iff \sum_{i=1}^{N} &\pi_{i} \sum_{\ell =1}^{N} \left( - \left(x_{ii}+\lambda x_{i\ell}\right) C_{i\ell} + \left(x_{i\ell}+\lambda x_{ii}\right) B_{\ell i} \right) \nonumber \\
&> \sum_{i,j=1}^{N} \pi_{i} p_{ij} \sum_{\ell =1}^{N} \left( - \left(x_{ij}+\lambda x_{i\ell}\right) C_{j\ell} + \left(x_{i\ell}+\lambda x_{ij}\right) B_{\ell j} \right) .
\end{align}
For social goods satisfying $B_{ij}=b\beta_{ij}$ and $C_{ij}=c\gamma_{ij}$, if $\gamma >0$ then this condition is
\begin{align}
\frac{b}{c}>\left(\frac{b}{c}\right)_{\lambda}^{\ast} = \frac{\sum_{i,j=1}^{N} \pi_{i} p_{ij} \sum_{\ell =1}^{N} \left(\tau_{ij}+\lambda \tau_{i\ell} - \lambda \tau_{j\ell}\right) \gamma_{j\ell}}{\sum_{i,j=1}^{N} \pi_{i} p_{ij} \sum_{\ell =1}^{N} \left(\lambda \tau_{ij}+\tau_{i\ell} -\tau_{j\ell}\right) \beta_{\ell j}} ,
\end{align}
where $\pi_{i}$ and $\tau_{ij}$ are the same as they were previously for PC updating with $\lambda =0$.

\subsubsection{DB updating}
For DB updating, the condition for selection to favor TFT relative to ALLD is
\begin{align}
\rho_{A}>\rho_{B} \iff \sum_{i=1}^{N} &\pi_{i} \sum_{\ell =1}^{N} \left( - \left(x_{ii}+\lambda x_{i\ell}\right) C_{i\ell} + \left(x_{i\ell}+\lambda x_{ii}\right) B_{\ell i} \right) \nonumber \\
&> \sum_{i,j=1}^{N} \pi_{i} p_{ij}^{\left(2\right)} \sum_{\ell =1}^{N} \left( - \left(x_{ij}+\lambda x_{i\ell}\right) C_{j\ell} + \left(x_{i\ell}+\lambda x_{ij}\right) B_{\ell j} \right) .
\end{align}
For social goods satisfying $B_{ij}=b\beta_{ij}$ and $C_{ij}=c\gamma_{ij}$, if $\gamma >0$ then this condition is
\begin{align}
\frac{b}{c}>\left(\frac{b}{c}\right)_{\lambda}^{\ast} = \frac{\sum_{i,j,\ell =1}^{N} \pi_{i} p_{ij}^{\left(2\right)} \left(\tau_{ij}+\lambda \tau_{i\ell}-\lambda\tau_{j\ell}\right) \gamma_{j\ell}}{\sum_{i,j,\ell =1}^{N} \pi_{i} p_{ij}^{\left(2\right)} \left(\lambda\tau_{ij}+\tau_{i\ell}-\tau_{j\ell}\right) \beta_{\ell j}} ,
\end{align}
where $\pi_{i}$ and $\tau_{ij}$ are the same as they were previously for DB updating with $\lambda =0$.

\subsubsection{IM updating}
For IM updating, the condition for selection to favor TFT relative to ALLD is
\begin{align}
\sum_{i,\ell =1}^{N} &\pi_{i} \left( - \left(x_{ii}+\lambda x_{i\ell}\right) C_{i\ell} + \left(x_{i\ell}+\lambda x_{ii}\right) B_{\ell i} \right) \nonumber \\
&> \sum_{i,j,\ell =1}^{N} \pi_{i} \widetilde{p}_{ij}^{\left(2\right)} \left( - \left(x_{ij}+\lambda x_{i\ell}\right) C_{j\ell} + \left(x_{i\ell}+\lambda x_{ij}\right) B_{\ell j} \right) .
\end{align}
For social goods satisfying $B_{ij}=b\beta_{ij}$ and $C_{ij}=c\gamma_{ij}$, if $\gamma >0$ then this condition is
\begin{align}
\frac{b}{c}>\left(\frac{b}{c}\right)_{\lambda}^{\ast} = \frac{\sum_{i,j,\ell =1}^{N} \pi_{i} \widetilde{p}_{ij}^{\left(2\right)} \left(\tau_{ij}+\lambda \tau_{i\ell}-\lambda\tau_{j\ell}\right) \gamma_{j\ell}}{\sum_{i,j,\ell =1}^{N} \pi_{i} \widetilde{p}_{ij}^{\left(2\right)} \left(\lambda\tau_{ij}+\tau_{i\ell}-\tau_{j\ell}\right) \beta_{\ell j}} ,
\end{align}
where $\pi_{i}$ and $\tau_{ij}$ are the same as they were previously for IM updating with $\lambda =0$.

\subsection{Numerical examples}
Extended~Data~Figures~\ref{fig:preferential_attachment}~and~\ref{fig:ER_SW} depict critical ratios on random graphs. \sfig{preferential_attachment} was generated using a Barab\'{a}si-Albert preferential attachment model \citep{barabasi:Science:1999}; the population begins with $m_{0}=2$ individuals, and each new individual is connected to $m=1$ existing members of the population according to the standard degree-weighted distribution. \sfig{ER_SW}\textbf{a} was generated using the $G\left(N,p\right)$ Erd\"{o}s-R\'{e}nyi model \citep{bollobas:CUP:2001}. \sfig{ER_SW}\textbf{b} was created using the Watts-Strogatz model \citep{watts:Nature:1998} starting from a cycle with edge-rewiring probability $p$.

For the empirical networks considered in the main text and Extended Data Figures (which are available on public databases \citep{snapnets,nr}), we have used the largest connected component to ensure that fixation is possible in the case of societies with more than one connected component. The sources of all empirical networks used here are provided in the corresponding figure captions.

\section*{Acknowledgments}
We are grateful to Joshua Plotkin for constructive feedback. We would also like to thank Babak Fotouhi and Christian Hilbe for helpful conversations. This work was supported by the Army Research Laboratory (grant W911NF-18-2-0265), the Bill \& Melinda Gates Foundation (grant OPP1148627), the John Templeton Foundation (grant 61443), the National Science Foundation (grant DMS-1715315), and the Office of Naval Research (grant N00014-16-1-2914). The funders had no role in study design, data collection and analysis, decision to publish, or preparation of the manuscript.

\newpage

\begin{center}
\Large\textbf{Supporting Information}
\end{center}

\setcounter{equation}{0}
\setcounter{figure}{0}
\setcounter{section}{0}
\setcounter{table}{0}
\renewcommand{\thesection}{SI.\arabic{section}}
\renewcommand{\thesubsection}{SI.\arabic{section}.\arabic{subsection}}
\renewcommand{\theequation}{SI.\arabic{equation}}
\renewcommand{\thetable}{SI.\arabic{table}}
\renewcommand{\figurename}{\footnotesize Supplementary~Figure}
\setcounter{figure}{0}

\section{Formally modeling evolutionary dynamics}\label{sec:appendixA}
In this section, we formally establish the conditions for $A$ to be favored over $B$ stated in the main text. The notation and modeling techniques follow the conventions of \citet{allen:JMB:2019}.

Since we model two types, $A$ and $B$, in a population of finite size, $N$, the state of the population is specified by the configuration of $A$ and $B$. For simplicity, we denote the state of the population by a binary vector, $\vx\in\left\{0,1\right\}^{N}$, where $x_{i}=1$ means that the individual at location $i$ is has type $A$, and $x_{i}=0$ means that this individual has type $B$. We denote by $\vA\coloneqq\left(1,1,\dots ,1\right)$ and $\vB\coloneqq\left(0,0,\dots ,0\right)$ the two monomorphic (or monoallelic) states, respectively.

We describe evolutionary dynamics using the notion of a replacement rule \citep{allen:JMB:2014,allen:JMB:2019}. A replacement rule is a distribution $\left\{p_{\left(R,\alpha\right)}\left(\mathbf{x}\right)\right\}_{\left(R,\alpha\right)}$ over pairs $\left(R,\alpha\right)$, where $R$ is a subset of $\left\{1,\dots ,N\right\}$ and $\alpha$ is a map $R\rightarrow\left\{1,\dots ,N\right\}$. $R$ is the set of individuals replaced in a given time step, and $\alpha$ is the parentage map, which means that $\alpha\left(i\right) =j$ if and only if the offspring of $j$ replaces $i$. The probability of choosing replacement event $\left(R,\alpha\right)$, which is denoted by $p_{\left(R,\alpha\right)}\left(\mathbf{x}\right)$, usually depends on the current population state, $\mathbf{x}\in\left\{0,1\right\}^{N}$, because an individual's type can influence the probability with which they reproduce and/or die.

The notion of a replacement rule is flexible enough to accommodate a wide variety of models of spatially structured populations\cite{allen:JMB:2014,allen:JMB:2019}. Although the update rules considered in the main text have at most one individual replaced per time step (in our notation, $p_{\left(R,\alpha\right)}\left(\mathbf{x}\right)=0$ whenever $\left| R\right| >1$), the results of this section allow for any fixed or variable number of individuals to be replaced.  In Section \ref{sec:examples_updateRules} we give specific formulas for $p_{\left(R,\alpha\right)}\left(\mathbf{x}\right)$ for various update rules.

We assume that $\left\{p_{\left(R,\alpha\right)}\left(\mathbf{x}\right)\right\}_{\left(R,\alpha\right)}$ satisfies the following property \citep{allen:JMB:2019}, which formally states that there is at least one individual that can eventually be an ancestor of the entire population:
\begin{fixation}
	There exists $i\in\left\{1,\dots ,N\right\}$, an integer $m\geqslant 1$, and a sequence of replacement events $\left\{\left(R_{k},\alpha_{k}\right)\right\}_{k=1}^{m}$ such that
	\begin{itemize}
		
		\item $p_{\left(R_{k},\alpha_{k}\right)}\left(\vx\right) >0$ for every $k\in\left\{1,\dots ,m\right\}$ and $\vx\in\left\{0,1\right\}^{N}$;
		
		\item $i\in R_{k}$ for some $k\in\left\{1,\dots ,m\right\}$;
		
		\item $\widetilde{\alpha}_{1}\circ\widetilde{\alpha}_{2}\circ\cdots\circ\widetilde{\alpha}_{m}\left(j\right) =i$ for every $j\in\left\{1,\dots ,N\right\}$, where $\widetilde{\alpha}:\left\{1,\dots ,N\right\}\rightarrow\left\{1,\dots ,N\right\}$ denotes the extension of $\alpha :R\rightarrow\left\{1,\dots ,N\right\}$ to all of $\left\{1,\dots ,N\right\}$, i.e.
		\begin{align}
		\widetilde{\alpha}\left(j\right) &\coloneqq 
		\begin{cases}
		\alpha\left(j\right) & j\in R , \\
		& \\
		j & j\not\in R .
		\end{cases}
		\end{align}
		
	\end{itemize}
\end{fixation}

We now describe how a replacement event affects the state of the population. Suppose that the current state is $\vx\in\left\{0,1\right\}^{N}$. Following the replacement event $\left(R,\alpha\right)$, individual $i$ inherits his or her type from $\widetilde{\alpha}\left(i\right)$; thus, the type of $i$ is updated from $x_{i}$ to $x_{\widetilde{\alpha}\left(i\right)}$. We denote by $\vx_{\widetilde{\alpha}}$ this resulting state. Given a replacement rule $\left\{p_{\left(R,\alpha\right)}\left(\mathbf{x}\right)\right\}_{\left(R,\alpha\right)}$ , we can then define a Markov chain on $\left\{0,1\right\}^{N}$, with the transition probability between $\vx$ and $\vy$ in $\left\{0,1\right\}^{N}$ given by
\begin{align}
P_{\vx\rightarrow\vy} &\coloneqq \sum_{\substack{\left(R,\alpha\right) \\ \vx_{\widetilde{\alpha}}=\vy}} p_{\left(R,\alpha\right)}\left(\vx\right) .
\end{align}

Since $\vA_{\widetilde{\alpha}}=\vA$ and $\vB_{\widetilde{\alpha}}=\vB$ for every replacement event, $\left(R,\alpha\right)$, $\vA$ and $\vB$ are absorbing states for this chain. Moreover, the Fixation Axiom implies that every non-monomorphic state is transient, so the chain will eventually end up in either $\vA$ or $\vB$ given any initial state. We are particularly interested in rare-mutant states, which for us means that there are $N-1$ of type $A$ and one of type $B$ or $N-1$ of type $B$ and one of type $A$. Denote by $\rho_{A}^{i}$ (resp. $\rho_{B}^{i}$) the probability that $\vA$ (resp. $\vB$) is reached from the state with just a single $A$ (resp. $B$) at location $i$. The mean fixation probabilities of $A$ and $B$ are then $\rho_{A}\coloneqq\left(1/N\right)\sum_{i=1}^{N}\rho_{A}^{i}$ and $\rho_{B}\coloneqq\left(1/N\right)\sum_{i=1}^{N}\rho_{B}^{i}$, respectively. This section is dedicated to using a replacement rule to derive a condition for selection to favor type $A$ relative to type $B$, i.e. $\rho_{A}>\rho_{B}$.

\subsection{Demographic variables}\label{subsec:demographicVariables}
The marginal probability of transmission from $i$ to $j$ in state $\mathbf{x}\in\left\{0,1\right\}^{N}$ is
\begin{align}
e_{ij}\left(\mathbf{x}\right) &\coloneqq \sum_{\substack{\left(R,\alpha\right) \\ \alpha\left(j\right) =i}} p_{\left(R,\alpha\right)}\left(\mathbf{x}\right) .
\end{align}
From these marginal probabilities come birth rates and death probabilities,
\begin{subequations}
	\begin{align}
	b_{i}\left(\mathbf{x}\right) &\coloneqq \sum_{j=1}^{N} e_{ij}\left(\mathbf{x}\right) ; \\
	d_{i}\left(\mathbf{x}\right) &\coloneqq \sum_{j=1}^{N} e_{ji}\left(\mathbf{x}\right) .
	\end{align}
\end{subequations}
Since the population size is fixed, the average birth rate, $\overline{b}\left(\vx\right) =\frac{1}{N}\sum_{i=1}^{N}b_{i}\left(\vx\right)$, coincides with the average death probability, $\overline{d}\left(\vx\right) =\frac{1}{N}\sum_{i=1}^{N}d_{i}\left(\vx\right)$, in every state, $\vx\in\left\{0,1\right\}^{N}$.

We let $\left\{p_{\left(R,\alpha\right)}\left(\vx\right)\right\}_{\left(R,\alpha\right)}$ depend on a non-negative parameter, $\delta\geqslant 0$, which quantifies the intensity of selection within the population. We  assume that $p_{\left(R,\alpha\right)}\left(\vx\right)$ is continuously-differentiable in a neighborhood of $\delta =0$, and we use the ``prime'' notation (e.g. $e_{ij}'$) to denote the $\delta$-derivative evaluated at $\delta =0$.   We further assume that, for $\delta=0$, $p_{\left(R,\alpha\right)}\left(\vx\right)$ is independent of $\vx$ (as are the demographic variables derived from the replacement rule). We refer to the case $\delta=0$ as neutral drift, and we use the superscript $\circ$ to denote this special case (e.g. $e_{ij}^{\circ}$).

\subsection{Reproductive value}\label{subsec:reproductiveValue}
Suppose that $\left\{\pi_{i}\right\}_{i=1}^{N}$ consists of scalars such that, for every $i=1,\dots ,N$,
\begin{align}\label{eq:RVbalance}
\sum_{j=1}^{N} e_{ij}^{\circ} \pi_{j} &= \sum_{j=1}^{N} e_{ji}^{\circ} \pi_{i} .
\end{align}
The (unique) solution to \eq{RVbalance} satisfying $\sum_{i=1}^{N}\pi_{i}=1$ is the so-called ``reproductive value'' (RV) of $i$ \citep{taylor:AN:1990,maciejewski:JTB:2014a}. $\pi_{i}$ can be interpreted as the probability that, under neutral drift, $i$ generates a lineage that takes over the population \citep{allen:JMB:2019}. Let $\widehat{b}_{i}\left(\mathbf{x}\right)\coloneqq\sum_{j=1}^{N} e_{ij}\left(\mathbf{x}\right) \pi_{j}$ and $\widehat{d}_{i}\left(\mathbf{x}\right)\coloneqq\sum_{j=1}^{N} e_{ji}\left(\mathbf{x}\right) \pi_{i}$ be the RV-weighted birth and death rates, respectively. The mean change in $\sum_{i=1}^{N}\pi_{i}x_{i}$, the RV-weighted abundance of $A$ in state $\mathbf{x}$, is
\begin{align}
\delhatsel\left(\mathbf{x}\right) &\coloneqq \sum_{i=1}^{N} x_{i}\left( \widehat{b}_{i}\left(\mathbf{x}\right) - \widehat{d}_{i}\left(\mathbf{x}\right)\right) ,
\end{align}
which enjoys the convenient property that $\delhatsel^{\circ}\left(\mathbf{x}\right) =0$ for every $\mathbf{x}$.

\subsection{Mutation and the RMC distribution}\label{subsec:mutationRMC}
For $u\in\left[0,1\right]$, consider the Markov chain on $\left\{0,1\right\}^{N}$ with transition probabilities
\begin{align}
\widetilde{P}_{\vx\rightarrow\vy} &\coloneqq 
\begin{cases}
u\frac{1}{N} & \vx =\vA ,\ \overline{y}=\frac{1}{N} , \\
& \\
\left(1-u\right) P_{\vA\rightarrow\vy} & \vx =\vA ,\ \overline{y}\neq\frac{1}{N} , \\
& \\
u\frac{1}{N} & \vx =\vB ,\ \overline{y}=1-\frac{1}{N} , \\
& \\
\left(1-u\right) P_{\vB\rightarrow\vy} & \vx =\vB ,\ \overline{y}\neq 1-\frac{1}{N} , \\
& \\
P_{\vx\rightarrow\vy} & \vx\not\in\left\{\vA ,\vB\right\} .
\end{cases} \label{eq:MSS_chain}
\end{align}
In other words, transitions from a non-monomorphic state are the same as in the mutation-free process; in a monomorphic state, a single mutant type ($B$ in the all-$A$ state and $A$ in the all-$B$ state) will arise with probability $u$, and the initial location of this mutant is chosen uniformly-at-random from the vertices (see \sifig{mutationFixation}). When $u=0$, we have $\widetilde{P}_{\vx\rightarrow\vy}=P_{\vx\rightarrow\vy}$ for every $\vx ,\vy\in\left\{0,1\right\}^{N}$. But when $u>0$, this Markov chain has a single closed communicating class by the Fixation Axiom, and consequently it has a unique stationary distribution, $\pi_{\MSS}$.

\begin{figure}
	\centering
	\includegraphics[width=1.0\textwidth]{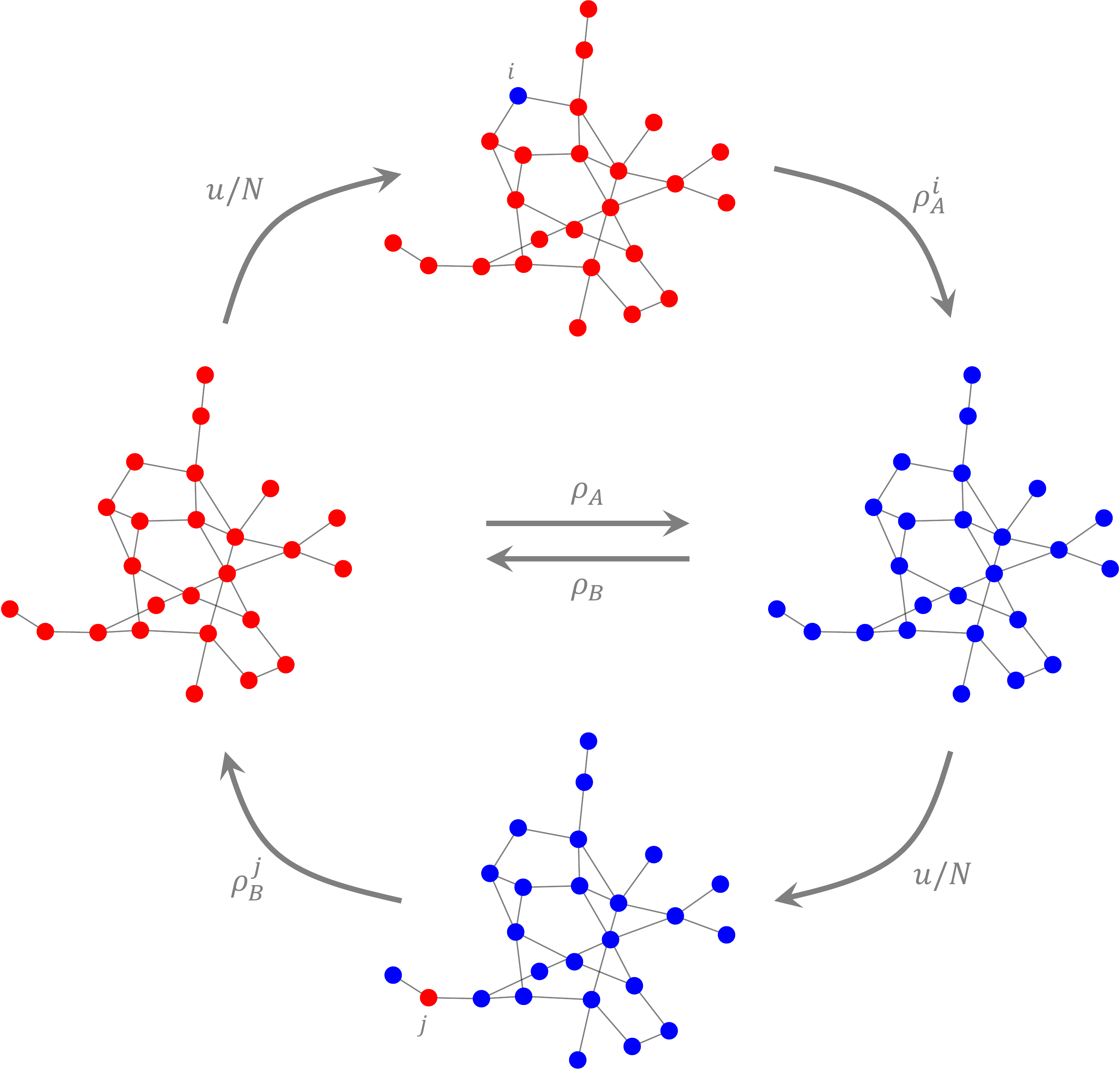}
	\caption{Mutation and absorption into one of the two monomorphic states. A single mutant arises with probability $u$ and is placed at location $i$ with probability $1/N$. Mutants do not arise in the non-monomorphic states, and eventually the process returns to a monomorphic state (all-$A$ with probability $\rho_{A}^{i}$ and all-$B$ with probability $1-\rho_{A}^{i}$) by the Fixation Axiom. The mean fixation probabilities of $A$ and be are then $\rho_{A}=\left(1/N\right)\sum_{i=1}^{N}\rho_{A}^{i}$ and $\rho_{B}=\left(1/N\right)\sum_{i=1}^{N}\rho_{B}^{i}$, respectively. When $u>0$, the resulting Markov chain has a unique stationary distribution, $\pi_{\MSS}$, which can be used to understand when $\rho_{A}>\rho_{B}$.\label{fig:mutationFixation}}
\end{figure}

In order to get a better understanding of the transient-state dynamics, we define
\begin{align}
\pi_{\RMC}\left(\vx\right) &\coloneqq \lim_{u\rightarrow 0} \frac{\pi_{\MSS}\left(\vx\right)}{1-\pi_{\MSS}\left(\vA\right) -\pi_{\MSS}\left(\vB\right)} ,
\end{align}
which we call the rare-mutation conditional (RMC) distribution. One may also characterize the RMC distribution as the stationary distribution for the chain on $\left\{0,1\right\}^{N}-\left\{\vA ,\vB\right\}$ defined by
\begin{align}
P_{\mathbf{x}\rightarrow\mathbf{y}}^{\RMC} &\coloneqq
\begin{cases}
P_{\mathbf{x}\rightarrow\mathbf{y}} + \frac{1}{N}P_{\mathbf{x}\rightarrow\mathbf{A}} & \overline{y}=\frac{1}{N} , \\
& \\
P_{\mathbf{x}\rightarrow\mathbf{y}} + \frac{1}{N}P_{\mathbf{x}\rightarrow\mathbf{B}} & \overline{y}=1-\frac{1}{N} , \\
& \\
P_{\mathbf{x}\rightarrow\mathbf{y}} & \frac{1}{N}<\overline{y}<1-\frac{1}{N} .
\end{cases} \label{eq:RMC_chain}
\end{align}
The definitions of both of these chains (Equations.~\ref{eq:MSS_chain}~and~\ref{eq:RMC_chain}) are slightly different from their definitions given in \citet{allen:JMB:2019}, but nonetheless they have the following properties (whose proofs follow from the arguments of \citet{allen:JMB:2019}):
\begin{enumerate}
	
	\item[Fact 1.] $\lim_{u\rightarrow 0}\pi_{\MSS}\left(\vA\right) =\rho_{A}/\left(\rho_{A}+\rho_{B}\right)$ and $\lim_{u\rightarrow 0}\pi_{\MSS}\left(\vB\right) =\rho_{B}/\left(\rho_{A}+\rho_{B}\right)$;
	
	\item[Fact 2.] The quantity
	\begin{align}
	K &\coloneqq \lim_{u\rightarrow 0} \frac{u}{\left(1-\pi_{\MSS}\left(\vA\right) -\pi_{\MSS}\left(\vB\right)\right)}
	\end{align}
	is strictly positive and differentiable in $\delta$ in a small neighborhood of $\delta =0$;
	
	\item[Fact 3.] $\mathbb{E}_{\RMC}\left[\varphi\right] =K\frac{d}{du}\Big\vert_{u=0}\mathbb{E}_{\MSS}\left[\varphi\right]$ for any $\varphi :\left\{0,1\right\}^{N}\rightarrow\mathbb{R}$ with $\varphi\left(\vA\right) =\varphi\left(\vB\right) =0$.
	
	\item[Fact 4.] $\mathbb{E}_{\RMC}^{\circ}\left[ x_{i}\right] =1/2$ for $i=1,\dots ,N$.
	
\end{enumerate}

\subsection{Type distributions under neutral drift}
The probability that $i$ and $j$ have the same type in the neutral $\RMC$ distribution is
\begin{align}
x_{ij} &\coloneqq \mathbb{E}_{\RMC}^{\circ}\left[ x_{i}x_{j}+\left(1-x_{i}\right)\left(1-x_{j}\right)\right] =2\mathbb{E}_{\RMC}^{\circ}\left[ x_{i}x_{j}\right] .
\end{align}
Let $K^{\circ}\coloneqq\lim_{u\rightarrow 0}u/\left(1-\pi_{\MSS}^{\circ}\left(\vA\right) -\pi_{\MSS}^{\circ}\left(\vB\right)\right)$ be the value of $K$ at $\delta =0$. If $i\neq j$, then
\begin{align}
\mathbb{E}_{\RMC}^{\circ}\left[ x_{i}x_{j}\right] &= \sum_{\vx\neq\vA ,\vB} \pi_{\RMC}^{\circ}\left(\vx\right) x_{i}x_{j} \nonumber \\
&= \sum_{\vx ,\vy\neq\vA ,\vB} \pi_{\RMC}^{\circ}\left(\vy\right) P_{\vy\rightarrow\vx}^{\RMC} x_{i}x_{j} \nonumber \\
&= \sum_{\vx ,\vy\neq\vA ,\vB} \pi_{\RMC}^{\circ}\left(\vy\right) P_{\vy\rightarrow\vx} x_{i}x_{j} + \sum_{\vy\neq\vA ,\vB} \pi_{\RMC}^{\circ}\left(\vy\right) \sum_{\ell\neq i,j} \frac{1}{N} P_{\vy\rightarrow\vB} \nonumber \\
&= \sum_{\vy\neq\vA ,\vB} \pi_{\RMC}^{\circ}\left(\vy\right) \sum_{\left(R,\alpha\right)} p_{\left(R,\alpha\right)}^{\circ} y_{\widetilde{\alpha}\left(i\right)}y_{\widetilde{\alpha}\left(j\right)} - \sum_{\vy\neq\vA ,\vB} \pi_{\RMC}^{\circ}\left(\vy\right) P_{\vy\rightarrow\vA} \nonumber \\
&\quad + \sum_{\vy\neq\vA ,\vB} \pi_{\RMC}^{\circ}\left(\vy\right) \sum_{\ell\neq i,j} \frac{1}{N} P_{\vy\rightarrow\vB} \nonumber \\
&= \sum_{\left(R,\alpha\right)} p_{\left(R,\alpha\right)}^{\circ} \mathbb{E}_{\RMC}^{\circ}\left[ x_{\widetilde{\alpha}\left(i\right)}x_{\widetilde{\alpha}\left(j\right)} \right] - \sum_{\vy\neq\vA ,\vB} \pi_{\RMC}^{\circ}\left(\vy\right) P_{\vy\rightarrow\vA} \nonumber \\
&\quad + \frac{N-2}{N}\sum_{\vy\neq\vA ,\vB} \pi_{\RMC}^{\circ}\left(\vy\right) P_{\vy\rightarrow\vB} \nonumber \\
&= \sum_{\left(R,\alpha\right)} p_{\left(R,\alpha\right)}^{\circ} \mathbb{E}_{\RMC}^{\circ}\left[ x_{\widetilde{\alpha}\left(i\right)}x_{\widetilde{\alpha}\left(j\right)} \right] - \frac{K^{\circ}}{2} + \frac{N-2}{N}\frac{K^{\circ}}{2} \nonumber \\
&= \sum_{\left(R,\alpha\right)} p_{\left(R,\alpha\right)}^{\circ} \mathbb{E}_{\RMC}^{\circ}\left[ x_{\widetilde{\alpha}\left(i\right)}x_{\widetilde{\alpha}\left(j\right)} \right] - \frac{K^{\circ}}{N} ,
\end{align}
Letting $\tau_{ij}\coloneqq\left(1-x_{ij}\right) /2K^{\circ}=\mathbb{E}_{\RMC}^{\circ}\left[ x-x_{i}x_{j}\right] /K^{\circ}$ gives, for $i\neq j$, the recurrence relation
\begin{align}
\tau_{ij} &= \frac{1}{N} + \sum_{\left(R,\alpha\right)} p_{\left(R,\alpha\right)}^{\circ} \tau_{\widetilde{\alpha}\left(i\right)\widetilde{\alpha}\left(j\right)} . \label{eq:tauRecurrence}
\end{align}
For $i=j$, we have $\tau_{ii}=0$ for every $i=1,\dots ,N$. This system of equations uniquely determines $\left\{\tau_{ij}\right\}_{i,j=1}^{N}$.

\subsection{Payoffs, fecundities, and social goods}\label{subsec:stochDet}
In the class of evolutionary processes considered here, interactions result in payoffs, which are then converted into fecundities. These fecundities are subsequently used to update the state of the population. More abstractly, every state $\vx\in\left\{0,1\right\}^{N}$ results in a vector of fecundities, $\mathbf{F}\in\left[0,\infty\right)^{N}$, with one entry for every individual. However, this state-to-fecundity mapping need not be deterministic; it could also be stochastic (\fig{concentrated}). In this section, we show that stochastic mappings can be reduced to deterministic mappings under weak selection.

In practice, state-based replacement rules (introduced in \S\ref{subsec:demographicVariables}) often depend on the state, $\vx$, only through its effect on fecundity. In other words, if there is a deterministic state-to-fecundity map, $\vx\mapsto\mathbf{F}\left(\vx\right)$, where $\mathbf{F}\in\left[0,\infty\right)^{N}$, then $p_{\left(R,\alpha\right)}\left(\vx\right) =p_{\left(R,\alpha\right)}\left(\mathbf{F}\left(\vx\right)\right)$. To avoid confusion when we consider stochastic state-to-fecundity mappings below, we denote by $p_{\left(R,\alpha\right)}$ a replacement rule that is a function of state, $\vx$, and by $q_{\left(R,\alpha\right)}$ a replacement rule that is a function of fecundity, $\mathbf{F}$. For example, DB updating in a graph-structured population, $\left(w_{ij}\right)_{i,j=1}^{N}$, is defined by the rule
\begin{align}
q_{\left(R,\alpha\right)}\left(\mathbf{F}\right) &= 
\begin{cases}
\frac{1}{N}\frac{F_{\alpha\left(i\right)}w_{\alpha\left(i\right) i}}{\sum_{j=1}^{N}F_{j}w_{ji}} & R=\left\{i\right\}\textrm{ for some }i, \\
& \\
0 & \textrm{otherwise.}
\end{cases} \label{eq:replacementDB}
\end{align}
In most traditional formulations of evolutionary games with DB updating, there exists a payoff function $u:\left\{0,1\right\}^{N}\rightarrow\mathbb{R}^{N}$ that gives a payoff vector for the population as a function of the state, $\vx$. The payoff for player $i$, $u_{i}\left(\vx\right)$, is then converted to fecundity, $F_{i}\left(\vx\right)$, by letting $F_{i}\left(\vx\right)\coloneqq\exp\left\{\delta u_{i}\left(\vx\right)\right\}$ for some selection intensity parameter, $\delta\geqslant 0$. Thus, we can write
\begin{align}
p_{\left(R,\alpha\right)}\left(\vx\right) &= 
\begin{cases}
\frac{1}{N}\frac{F_{\alpha\left(i\right)}\left(\vx\right) w_{\alpha\left(i\right) i}}{\sum_{j=1}^{N}F_{j}\left(\vx\right) w_{ji}} & R=\left\{i\right\}\textrm{ for some }i, \\
& \\
0 & \textrm{otherwise.}
\end{cases} \label{eq:replacementDB_state}
\end{align}

Suppose now that we have a fixed fecundity-based replacement rule, $q$, together with a \textit{stochastic} state-to-fecundity mapping. Thus, for every state $\mathbf{x}\in\left\{0,1\right\}^{N}$, there is a distribution over fecundity vectors $\mathbf{F}\in\left[0,\infty\right)^{N}$. If $\mathbb{E}_{\mathbf{x}}$ denotes expectation with respect to this distribution, then one obtains a state-based rule,
\begin{align}
p_{\left(R,\alpha\right)}\left(\mathbf{x}\right) &\coloneqq \mathbb{E}_{\mathbf{x}}\left[ q_{\left(R,\alpha\right)} \right] .
\end{align}
We consider fecundity-based replacement rules for which $q_{\left(R,\alpha\right)}\left(\mathbf{F}\right) =q_{\left(R,\alpha\right)}^{\circ}+\delta q_{\left(R,\alpha\right)}'\left(\mathbf{F}\right) +O\left(\delta^{2}\right)$ for some function $q_{\left(R,\alpha\right)}'\left(\mathbf{F}\right)$ whenever $\delta\ll 1$. Moreover, we assume that the distribution on fecundity is determined by randomness in the payoffs. Specifically, every individual, $i$, receives a payoff, $u_{i}$, based on some probability distribution. This payoff is then converted to fecundity via the formula $F_{i}\left(\delta\right) =\exp\left\{\delta u_{i}\right\}$, where $\delta\geqslant 0$ is the intensity of selection. It follows that
\begin{align}
\frac{d}{d\delta}\mathbb{E}_{\mathbf{x}}\left[ F_{i}\left(\delta\right) \right] &= \mathbb{E}_{\mathbf{x}}\left[ F_{i}'\left(\delta\right) \right]
\end{align}
for every $i=1,\dots ,N$.

Consider the ``averaged'' replacement rule, $\overline{p}_{\left(R,\alpha\right)}\left(\vx\right)\coloneqq q_{\left(R,\alpha\right)}\left(\mathbb{E}_{\vx}\left[\mathbf{F}\right]\right)$. By the chain rule,
\begin{align}
p_{\left(R,\alpha\right)}\left(\mathbf{x}\right) &= \mathbb{E}_{\mathbf{x}}\left[ q_{\left(R,\alpha\right)} \right] \nonumber \\
&= q_{\left(R,\alpha\right)}^{\circ}+\delta\mathbb{E}_{\mathbf{x}}\left[ q_{\left(R,\alpha\right)}' \right] +O\left(\delta^{2}\right) \nonumber \\
&= q_{\left(R,\alpha\right)}^{\circ}+\delta\sum_{i=1}^{N}\frac{\partial q_{\left(R,\alpha\right)}}{\partial F_{i}}\Bigg\vert_{\mathbf{F}=\mathbf{F}\left(0\right)}\mathbb{E}_{\mathbf{x}}\left[ F_{i}'\left(0\right) \right] +O\left(\delta^{2}\right) \nonumber \\
&= q_{\left(R,\alpha\right)}^{\circ}+\delta\sum_{i=1}^{N}\frac{\partial q_{\left(R,\alpha\right)}}{\partial F_{i}}\Bigg\vert_{\mathbf{F}=\mathbf{F}\left(0\right)}\frac{d\mathbb{E}_{\mathbf{x}}\left[ F_{i}\left(\delta\right) \right]}{d\delta}\Bigg\vert_{\delta =0} +O\left(\delta^{2}\right) .
\end{align}
It follows that
\begin{align}
\frac{d}{d\delta}\Bigg\vert_{\delta =0} p_{\left(R,\alpha\right)}\left(\mathbf{x}\right) &= \frac{d}{d\delta}\Bigg\vert_{\delta =0} \overline{p}_{\left(R,\alpha\right)}\left(\mathbf{x}\right) .
\end{align}
From this equation, together with the fact that fixation probabilities are defined by linear systems whose coefficients are based on the replacement rule, we conclude that the first-order behavior of the fixation probabilities must coincide under realized and expected payoffs. Thus, without a loss of generality, we may assume that the state-to-fecundity mapping is deterministic.

Suppose that type $A$ at location $i$ pays $C_{ij}$ to donate $B_{ij}$ to $j$. Type $B$ at this same location pays $c_{ij}$ to donate $b_{ij}$ to $j$. Although this formulation accounts for general additive interactions, we are mostly interested in the case in which $A$ is a producer and $B$ is a non-producer, meaning $b_{ij}=c_{ij}=0$ for every $i$ and $j$. In the most general case, however, the cumulative payoff to $i$ is
\begin{align}
\label{eq:udef}
u_{i}\left(\mathbf{x}\right) &= \sum_{j=1}^{N}\left( -x_{i}C_{ij} -\left(1-x_{i}\right) c_{ij} +x_{j}B_{ji} +\left(1-x_{j}\right) b_{ji} \right) .
\end{align}
This formulation can account for both accumulated and averaged payoffs due to the dependence of $B_{ij}$, $C_{ij}$, $b_{ij}$, and $c_{ij}$ on both $i$ and $j$; we give explicit examples along these lines in \S\ref{sec:examples_socialGoods}.

\subsection{Marginal replacement effects}
By the results of \S\ref{subsec:stochDet}, we may assume (without a loss of generality) that the payoff-to-fecundity map is deterministic. Let $u:\left\{0,1\right\}^{N}\rightarrow\mathbb{R}^{N}$ be a payoff function that assigns a real number, $u_{i}\left(\vx\right)$, to each individual, $i$, and state, $\vx$. Individual $i$'s payoff is then converted to fecundity by $F_{i}\left(\vx\right)\coloneqq\exp\left\{\delta u_{i}\left(\vx\right)\right\}$. Moreover, we assume that for every $\mathbf{x}\in\left\{0,1\right\}^{N}$, we have $e_{ij}\left(\vx\right) =e_{ij}\left(\mathbf{F}\right)$ (meaning $e_{ij}$ depends on $\vx$ only through the effects of $\vx$ on fecundity). Letting $m_{k}^{ij}\coloneqq\frac{\partial}{\partial F_{k}}\Big\vert_{\mathbf{F}=\mathbf{F}\left(0\right)}e_{ij}\left(\mathbf{F}\right)$ be the \textit{marginal effect of $k$ on $i$ replacing $j$,} we see that
\begin{align}
\label{eq:em}
e_{ij}'\left(\vx\right) &= \sum_{k=1}^{N} \left(\frac{\partial}{\partial F_{k}}\Bigg\vert_{\mathbf{F}=\mathbf{F}\left(0\right)}e_{ij}\left(\mathbf{F}\right)\right) \left(\frac{d}{d\delta}\Bigg\vert_{\delta =0}F_{k}\left(\delta\right)\right) = \sum_{k=1}^{N} m_{k}^{ij} u_{k}\left(\vx\right) .
\end{align}
Note, in particular, that $m_{k}^{ij}$ is independent of the payoffs and structure of the game. It can also be easily calculated for any process from the details of the update rule; we give examples below.

\subsection{Condition for evolutionary success of $A$ relative to $B$}
Our condition for the success of $A$ relative to $B$ is based on the following result, which is a modification of Theorem 8 of \citet{allen:JMB:2019}:
\begin{lemma}\label{lem:delhatsel}
	For the class of processes described herein,
	\begin{align}
	\rho_{A}' > \rho_{B}' &\iff \mathbb{E}_{\RMC}^{\circ}\left[\delhatsel '\right] > 0 . \label{eq:theorem8mod}
	\end{align}
\end{lemma}
\begin{proof}[Sketch of proof]
	The expected change in RV-weighted abundance of $A$ in state $\vx$ is
	\begin{align}
	\delhat\left(\vx\right) &= 
	\begin{cases}
	-u\frac{1}{N} & \vx =\vA , \\
	& \\
	u\frac{1}{N} & \vx =\vB , \\
	& \\
	\delhatsel\left(\vx\right) & \vx\not\in\left\{\vA ,\vB\right\} .
	\end{cases}
	\end{align}
	Since this expected change must average out to $0$ over the stationary distribution for the chain,
	\begin{align}
	0 &= \mathbb{E}_{\MSS}\left[ \widehat{\Delta} \right] = \mathbb{E}_{\MSS}\left[ \delhatsel \right] - u\frac{1}{N}\pi_{\MSS}\left(\vA\right) + u\frac{1}{N}\pi_{\MSS}\left(\vB\right) .
	\end{align}
	Differentiating this equation with respect to $u$ at $u=0$ gives
	\begin{align}
	\frac{N}{K}\mathbb{E}_{\RMC}\left[ \delhatsel \right] &= \lim_{u\rightarrow 0}\left(\pi_{\MSS}\left(\vA\right) - \pi_{\MSS}\left(\vB\right)\right) \nonumber \\
	&= \frac{\rho_{A}-\rho_{B}}{\rho_{A}+\rho_{B}} , \label{eq:uDerivative}
	\end{align}
	where, again, $K\coloneqq\lim_{u\rightarrow 0} u/\left(1-\pi_{\MSS}\left(\vA\right) -\pi_{\MSS}\left(\vB\right)\right)$. Since $\delhatsel^{\circ}\left(\vx\right) =0$ for every $\vx\in\left\{0,1\right\}^{N}$, differentiating \eq{uDerivative} with respect to $\delta$ at $\delta =0$ gives
	\begin{align}
	\frac{N}{K^{\circ}} \mathbb{E}_{\RMC}^{\circ}\left[\delhatsel '\right] & =
	\frac{\rho_{A}'-\rho_{B}'}{2\rho_{A}^{\circ}}.
	\end{align}
	by Facts~1--3 in \S\ref{subsec:mutationRMC}. Since $K^{\circ}>0$ and $\rho_{A}^{\circ}=\left(1/N\right)\sum_{i=1}^{N}\pi_{i}=1/N$ (see \S\ref{subsec:reproductiveValue}), it follows that $\rho_{A}'>\rho_{B}'$ if and only if $\mathbb{E}_{\RMC}^{\circ}\left[\delhatsel '\right] >0$, as desired.
\end{proof}

\begin{theorem}\label{thm:generalAdditive}
	\begin{align}
	\rho_{A}' > \rho_{B}' &\iff \sum_{i,j,k,\ell=1}^{N} \pi_{i} m_{k}^{ji} \left( -\left(x_{jk}-x_{ik}\right)\left(C_{k\ell}-c_{k\ell}\right) +\left(x_{j\ell}-x_{i\ell}\right)\left(B_{\ell k}-b_{\ell k}\right) \right) > 0 .
	\end{align}
\end{theorem}
\begin{proof}
	A straightforward calculation using the definition $m_{k}^{ij}\coloneqq\frac{\partial}{\partial F_{k}}\Big\vert_{\mathbf{F}=\mathbf{F}\left(0\right)}e_{ij}\left(\mathbf{F}\right)$ gives
	\begin{align}
	\delhatsel '\left(\mathbf{x}\right) &= \sum_{i=1}^{N}x_{i}\left(\widehat{b}_{i}'\left(\vx\right) -\widehat{d}_{i}'\left(\vx\right)\right) \nonumber \\
	&= \sum_{i,j=1}^{N}x_{i}\left(e_{ij}'\left(\vx\right) \pi_{j}-e_{ji}'\left(\vx\right) \pi_{i}\right) \nonumber \\
	&= \sum_{i,j=1}^{N}\pi_{i}\left( x_{j}-x_{i}\right) e_{ji}'\left(\vx\right) \nonumber \\
	&= \sum_{i,j,k=1}^{N} \pi_{i} m_{k}^{ji} \left( x_{j}-x_{i}\right) u_{k}\left(\vx\right) . \label{eq:delhatselev}
	\end{align}
	Since $x_{ij}\coloneqq 2\mathbb{E}_{\RMC}^{\circ}\left[ x_{i}x_{j}\right]$, we see that
	\begin{align}
	\mathbb{E}_{\RMC}^{\circ}\left[ \left( x_{j}-x_{i}\right) u_{k}\left(\vx\right) \right] &= \mathbb{E}_{\RMC}^{\circ}\left[ \left(x_{j}-x_{i}\right)\sum_{\ell =1}^{N}\left( \substack{-x_{k}C_{k\ell} -\left(1-x_{k}\right) c_{k\ell} \\ + x_{\ell}B_{\ell k} +\left(1-x_{\ell}\right) b_{\ell k}} \right) \right] \nonumber \\
	&= \frac{1}{2}\sum_{\ell =1}^{N} \left( -x_{jk}C_{k\ell} -\left(1-x_{jk}\right) c_{k\ell} +x_{j\ell}B_{\ell k} +\left(1-x_{j\ell}\right) b_{\ell k} \right) \nonumber \\
	&\quad -\frac{1}{2}\sum_{\ell =1}^{N} \left( -x_{ik}C_{k\ell} -\left(1-x_{ik}\right) c_{k\ell} +x_{i\ell}B_{\ell k} +\left(1-x_{i\ell}\right) b_{\ell k} \right) \nonumber \\
	&= \frac{1}{2}\sum_{\ell =1}^{N} \left( -x_{jk}\left(C_{k\ell}-c_{k\ell}\right) +x_{j\ell}\left(B_{\ell k}-b_{\ell k}\right) \right) \nonumber \\
	&\quad -\frac{1}{2}\sum_{\ell =1}^{N} \left( -x_{ik}\left(C_{k\ell}-c_{k\ell}\right) +x_{i\ell}\left(B_{\ell k}-b_{\ell k}\right) \right) .
	\end{align}
	The theorem then follows at once from \lem{delhatsel}.
\end{proof}

\begin{corollary}\label{cor:producerNonProducer}
	When $A$ is a producer and $B$ is a non-producer ($b_{ij}=c_{ij}=0$ for every $i$ and $j$),
	\begin{align}
	\rho_{A}'>\rho_{B}' &\iff \sum_{i,j,k,\ell =1}^{N} \pi_{i} m_{k}^{ji} \left( -x_{jk}C_{k\ell} + x_{j\ell}B_{\ell k} \right) > \sum_{i,j,k,\ell =1}^{N} \pi_{i} m_{k}^{ji} \left( -x_{ik}C_{k\ell} + x_{i\ell}B_{\ell k} \right) .
	\end{align}
\end{corollary}

\begin{remark}
	Although Theorem~\ref{thm:generalAdditive} and Corollary~\ref{cor:producerNonProducer} are stated in terms of $x_{ij}$, they can be evaluated by replacing $1-x_{ij}$ by $\tau_{ij}$ since $\tau_{ij}\coloneqq\left(1-x_{ij}\right) /2K^{\circ}$ for some $K^{\circ}>0$. $\tau$ can be easily calculated using \eq{tauRecurrence} and the fact that $\tau_{ii}=0$ for $i=1,\dots ,N$. However, the statements of Theorem~\ref{thm:generalAdditive} and Corollary~\ref{cor:producerNonProducer} are somewhat more intuitive using these probabilities, $x_{ij}$, directly, since then the constituent terms can be interpreted as expected payoffs.
\end{remark}

Returning to \eq{generalConditionSG} in the main text, let \emph{(i)} $P_{ji}$ be the probability that a producer in location $j$ replaces a random individual in location $i$ and \emph{(ii)} $Q_{ji}$ be the probability that a random individual in location $j$ replaces a producer in location $i$. Since $P_{ji}^{\circ}=Q_{ji}^{\circ}$ (i.e. when $\delta =0$) and
\begin{subequations}
	\begin{align}
	P_{ji}' &= \sum_{k,\ell =1}^{N} m_{k}^{ji} \left( -x_{jk}C_{k\ell} + x_{j\ell}B_{\ell k} \right) ; \\
	Q_{ji}' &= \sum_{k,\ell =1}^{N} m_{k}^{ji} \left( -x_{ik}C_{k\ell} + x_{i\ell}B_{\ell k} \right) ,
	\end{align}
\end{subequations}
we see that the signs of the first-order terms of $\rho_{A}-\rho_{B}$ and $\sum_{i,j=1}^{N}\pi_{i}\left(P_{ji}-Q_{ji}\right)$ agree.

\subsection{Relationship to inclusive fitness theory}
Inclusive fitness theory \citep{hamilton:JTB:1964a,lehmann:PTRSB:2014,birch:OUP:2017} is often used to model the evolution of social behavior. According to this theory, individuals evolve to act as if maximizing a quantity called inclusive fitness, which is a sum of fitness effects caused by an actor, each weighted by relatedness to the recipient. Careful analysis has revealed that to define the inclusive fitness of an individual requires weak selection, additivity of fitness effects, and other assumptions\cite{nowak:Nature:2010,allen:COBS:2016,birch:OUP:2017}, which hold in the model considered here.

Before identifying inclusive fitness effects, we must define individual fitness. The fitness of a vertex $i$ in a given state $\vx$ can be defined as\cite{allen:JMB:2019}
\begin{align}
\omega_{i}\left(\vx\right) &= \pi_{i} + \widehat{b}_{i}\left(\vx\right) - \widehat{d}_{i}\left(\vx\right) = \pi_{i} + \sum_{j=1}^{N} \left( e_{ij}\left(\vx\right) \pi_{j} - e_{ji}\left(\vx\right) \pi_{i} \right) .
\end{align}
We observe that for neutral drift, fitness is simply equal to reproductive value, $\omega_{i}^{\circ}=\pi_{i}$. For weak selection, we have
\begin{align}
\label{eq:fitweak}
\omega_{i}'\left(\vx\right) &= \sum_{j=1}^{N} \left( e_{ij}'\left(\vx\right) \pi_{j} - e_{ji}'\left(\vx\right) \pi_{i} \right) .
\end{align}
From \eq{em}, this can be written as
\begin{align}
\omega_{i}'\left(\vx\right) &= \sum_{k=1}^{N} M_{k}^{i} u_{k}\left(\vx\right) ,
\end{align}
where the quantity
\begin{align}
M_{k}^{i} &\coloneqq \frac{\partial}{\partial F_{k}}\Bigg\vert_{\mathbf{F}=\mathbf{F}\left(0\right)} \left( e_{ij} \left(\mathbf{F}\right) \pi_{j} - e_{ji} \left(\mathbf{F}\right) \pi_{i} \right) = \sum_{j=1}^{N} \left( m_{k}^{ij} \pi_{j}  - m_{k}^{ji} \pi_{i} \right)
\end{align} 
describes how the payoff to vertex $k$ affects the fitness of vertex $i$. We observe from Equations~\eqref{eq:delhatselev}~and~\eqref{eq:fitweak} that 
\begin{align}
\delhatsel '\left(\vx\right) &= \sum_{i=1}^{N} x_{i} \omega_{i}'\left(\vx\right) .
\end{align}

Substituting from \eq{udef}, we obtain
\begin{align}
\omega_{i}'\left(\vx\right) & = \sum_{k=1}^{N} \left[ M_{k}^{i} \sum_{\ell =1}^{N}\left( -x_{k}C_{k\ell} -\left(1-x_{k}\right) c_{k\ell} +x_{\ell}B_{\ell k} +\left(1-x_{\ell}\right) b_{\ell k} \right) \right] \nonumber \\
& = \sum_{k=1}^{N}  \left[M_{k}^{i} \sum_{\ell =1}^{N}\left( -x_{k}C_{k\ell} -\left(1-x_{k}\right) c_{k\ell} \right) + \sum_{\ell =1}^N M^i_\ell \left( x_{k}B_{k \ell} +\left(1-x_{k}\right) b_{k \ell} \right) \right] \nonumber \\
\label{eq:NMfit}
& = \sum_{k=1}^{N} \left[ x_{k} \sum_{\ell =1}^{N} \left( - C_{k\ell} M_{k}^{i} + B_{k\ell} M_{\ell}^{i} \right) + \left(1-x_{k}\right) \sum_{\ell =1}^{N}  \left( - c_{k\ell} M_{k}^{i} + b_{k\ell} M_{\ell}^{i} \right) \right] .
\end{align}
The final line of \eq{NMfit} expresses the neighbor-modulated fitness of vertex $i$, in that it identifies the  contribution (fitness effect) of each vertex $k$ to the fitness of $i$. Specifically, this fitness effect is $\sum_{\ell =1}^{N} \left( - C_{k\ell} M_{k}^{i} + B_{k \ell} M_{\ell}^{i} \right)$ for vertices $k$ of type $A$, and $\sum_{\ell =1}^{N} \left( - c_{k\ell} M_{k}^{i} + b_{k \ell} M_{\ell}^{i} \right)$ for vertices $k$ of type $B$. 

To move from neighbor-modulated to inclusive fitness, one must causally attribute every fitness effect on every individual to a particular actor in the population \citep{allen:COBS:2016,birch:OUP:2017}. Even in our simple model, this causal attribution is rather arbitrary and artificial. For example, if individual $i$ gives benefit $B_{ij}$ to individual $j$, which in turn alters individual $k$'s fitness by an amount $B_{ij} M_{j}^{k}$, to whom should this effect on $k$'s fitness be causally attributed? To $i$, to $j$, or both? To make progress, we adopt the convention that all fitness effects are attributed to the originator of the social good (e.g.,~terms of the form $B_{ij} M_{j}^{k}$ are attributed to individual $i$).

To formulate the inclusive fitness of individual $i$, we now multiply each effect attributable to individual $i$ by the relatedness of $i$ to the recipient. For a given state $\vx$, we use a notion of ``relatedness in state,'' such that the relatedness of two individuals is one if they have the same type and zero otherwise; as a formula, the relatedness of $i$ and $j$ is $x_{i}x_{j}+\left(1-x_{i}\right)\left(1-x_{j}\right)$. (Later, we will average over states to obtain relatedness coefficients between zero and one.) Multiplying the fitness effects from Eq.~\eqref{eq:NMfit} by the corresponding relatedness-in-state coefficients, we obtain the inclusive fitness effect of vertex $i$ in state $\vx$ as
\begin{align}
\label{eq:IFcases}
\omega_{i}^{\mathrm{IF}}\left(\vx\right) &= 
\begin{cases}
\sum_{k,\ell =1}^{N} \left(-C_{i\ell} M_{i}^{k} + B_{i\ell} M_{\ell}^{k} \right) x_{k} & x_{i}=1 , \\
& \\
\sum_{k,\ell =1}^{N} \left(-c_{i\ell} M_{i}^{k} + b_{i\ell} M_{\ell}^{k} \right) \left(1-x_{k}\right) & x_{i}=0 .
\end{cases}
\end{align}
We immediately see that
\begin{align}
\sum_{i=1}^{N} x_{i} \omega_{i}^{\mathrm{IF}}\left(\vx\right) &= \sum_{i,k=1}^{N} \left(-C_{i\ell} M_{i}^{k} + B_{i\ell} M_{\ell}^{k} \right) x_{i} x_{k} \nonumber \\
& = \sum_{i,k=1}^{N} \left(-C_{k\ell} M_{k}^{i} + B_{k\ell} M_{\ell}^{i} \right) x_{i} x_{k} \nonumber \\
& = \sum_{i=1}^{N} x_{i} \omega_{i}'\left(\vx\right) \nonumber \\
\label{eq:NMvsIF}
& = \delhatsel '\left(\vx\right) .
\end{align}
Inclusive fitness, when it exists, is therefore an alternative accounting method that leads to the same result for the gradient of selection, $\delhatsel '\left(\vx\right)$.

To obtain fitness quantities that apply to the overall evolutionary process, we average over the neutral RMC distribution, conditioned on the type of vertex $i$. First, we consider neighbor-modulated fitness. If vertex $i$ has type $A$ then, noting that $\mathbb{E}_{\RMC}^{\circ}\left[ x_{j}\ \mid\ x_{i}=1\right] = x_{ij}$, we obtain
\begin{align}
\mathbb{E}_{\RMC}^{\circ}\left[ \omega_{i}'\left(\vx\right)\ \mid\ x_{i}=1\right] = \sum_{k=1}^N \Bigg[ x_{ik} &\sum_{\ell =1}^{N} \left( - C_{k\ell} M_{k}^{i} + B_{k\ell} M_{\ell}^{i} \right)
\nonumber \\
&+ \left(1-x_{ik}\right) \sum_{\ell =1}^{N}  \left( -c_{k\ell} M_{k}^{i} + b_{k\ell} M_{\ell}^{i} \right) \Bigg].
\end{align}
If $i$ has type $B$, then since $\mathbb{E}_{\RMC}^{\circ}\left[ x_{j}\ \mid\ x_{i}=0\right] = 1-x_{ij}$, we have
\begin{align}
\mathbb{E}_{\RMC}^{\circ}\left[ \omega_{i}'\left(\vx\right)\ \mid\ x_{i}=0\right] = \sum_{k=1}^{N} \Bigg[ \left(1- x_{ik}\right) &\sum_{\ell =1}^{N} \left( - C_{k\ell} M_{k}^{i} + B_{k\ell} M_{\ell}^{i} \right)
\nonumber \\
&+ x_{ik} \sum_{\ell =1}^{N}  \left( - c_{k\ell} M_{k}^{i} + b_{k\ell} M_{\ell}^{i} \right) \Bigg].
\end{align}
Now, turning to inclusive fitness, we have
\begin{align}
\mathbb{E}_{\RMC}^{\circ}\left[ \omega_{i}^{\mathrm{IF}}\left(\vx\right)\ \mid\ x_{i}=1\right] &= \sum_{k,\ell =1}^{N} \left(-C_{i\ell} M_{i}^{k} + B_{i\ell} M_{\ell}^{k} \right) x_{ik} ,
\end{align}
for type $A$, and
\begin{align}
\mathbb{E}_{\RMC}^{\circ}\left[ \omega_{i}^{\mathrm{IF}}\left(\vx\right)\ \mid\ x_{i}=0\right] &= \sum_{k,\ell =1}^{N} \left(-c_{i\ell} M_{i}^{k} + b_{i\ell} M_{\ell}^{k} \right) x_{ik} ,
\end{align}
for type $B$.

From \eq{NMvsIF} we have $\mathbb{E}_{\RMC}^{\circ}\left[\delhatsel ' \right] = \mathbb{E}_{\RMC}^{\circ}\left[\sum_{i=1}^N x_{i} \omega_{i}^{\mathrm{IF}} \right] = \mathbb{E}_{\RMC}^{\circ}\left[\sum_{i=1}^N x_{i} \omega_{i}'  \right]$, again demonstrating the equivalence of the neighbor-modulated and inclusive fitness accounting methods, in the case of this model.

\section{Specific update rules}\label{sec:examples_updateRules}
We now turn to specific examples of update rules (\fig{update_rules}) in graph-structured populations. Let $\left( w_{ij}\right)_{i,j=1}^{N}$ be an undirected, unweighted, connected graph on $N$ vertices. The degree of vertex $i$ is simply the number of links connected to that vertex, i.e. $w_{i}\coloneqq\sum_{j=1}^{N}w_{ij}$. This graph defines the structure of the population, with links indicating neighbor relationships.

\subsection{Pairwise-comparison (PC) updating}\label{subsec:PCrule}
Under PC updating (see \fig{update_rules}), the probability of replacement event $\left(R,\alpha\right)$ is
\begin{align}
p_{\left(R,\alpha\right)}\left(\mathbf{x}\right) &=
\begin{cases}
\frac{1}{N} p_{i\alpha\left(i\right)} \frac{F_{\alpha\left(i\right)}\left(\mathbf{x}\right)}{F_{i}\left(\mathbf{x}\right) +F_{\alpha\left(i\right)}\left(\mathbf{x}\right)} & R=\left\{i\right\}\textrm{ for some }i\in\left\{1,\dots ,N\right\} ,\ \alpha\left(i\right)\neq i , \\
& \\
\frac{1}{N} \sum_{j=1}^{N} p_{ij} \frac{F_{i}\left(\mathbf{x}\right)}{F_{i}\left(\mathbf{x}\right) +F_{j}\left(\mathbf{x}\right)} & R=\left\{i\right\}\textrm{ for some }i\in\left\{1,\dots ,N\right\} ,\ \alpha\left(i\right) =i , \\
& \\
0 & \textrm{otherwise} .
\end{cases}
\end{align}
For $i,j=1,\dots ,N$, the marginal probability that $i$ transmits its offspring to $j\neq i$ is
\begin{align}\label{eq:marginalNotSelfPC}
e_{ij}\left(\mathbf{x}\right) &= \frac{1}{N} p_{ji} \frac{F_{i}\left(\mathbf{x}\right)}{F_{i}\left(\mathbf{x}\right) +F_{j}\left(\mathbf{x}\right)} ,
\end{align}
which gives a marginal effect of $k$ on $i$ replacing $j$ of
\begin{align}
m_{k}^{ij} &= 
\begin{cases}
\frac{1}{4N}p_{ji} & k = i , \\
& \\
-\frac{1}{4N}p_{ji} & k = j , \\
& \\
0 & k\neq i,j .
\end{cases}
\end{align}
It follows that
\begin{subequations}
	\begin{align}
	\sum_{i,j,k,\ell =1}^{N} \pi_{i} m_{k}^{ji} \left( -x_{jk}C_{k\ell} + x_{j\ell}B_{\ell k} \right) = \sum_{i,j=1}^{N} \pi_{i} \Bigg[ &\frac{1}{4N}p_{ij}\sum_{\ell =1}^{N}\left( -x_{jj}C_{j\ell} + x_{j\ell}B_{\ell j} \right) \nonumber \\
	&\quad -\frac{1}{4N}p_{ij}\sum_{\ell =1}^{N}\left( -x_{ji}C_{i\ell} + x_{j\ell}B_{\ell i} \right) \Bigg] ; \\
	\sum_{i,j,k,\ell =1}^{N} \pi_{i} m_{k}^{ji} \left( -x_{ik}C_{k\ell} + x_{i\ell}B_{\ell k} \right) = \sum_{i,j=1}^{N} \pi_{i} \Bigg[ &\frac{1}{4N}p_{ij}\sum_{\ell =1}^{N}\left( -x_{ij}C_{j\ell} + x_{i\ell}B_{\ell j} \right) \nonumber \\
	&\quad -\frac{1}{4N}p_{ij}\sum_{\ell =1}^{N}\left( -x_{ii}C_{i\ell} + x_{i\ell}B_{\ell i} \right) \Bigg] .
	\end{align}
\end{subequations}
The reproductive value of $i$ under PC updating is $\pi_{i}=w_{i}/\sum_{k=1}^{N}w_{k}$. Since $\pi_{i}p_{ij}=\pi_{j}p_{ji}$ for every $i$ and $j$, it follows that $\rho_{A}>\rho_{B}$ for small $\delta >0$ if and only if 
\begin{align}
\sum_{i=1}^{N} \pi_{i} \sum_{\ell =1}^{N}\left( -x_{ii}C_{i\ell} + x_{i\ell}B_{\ell i} \right) &> \sum_{i,j=1}^{N} \pi_{i} p_{ij}\sum_{\ell =1}^{N}\left( -x_{ij}C_{j\ell} + x_{i\ell}B_{\ell j} \right) ,
\end{align}
which gives \eq{mainConditionPC_mainText} in the main text. To evaluate this condition, recall that we can replace $1-x_{ij}$ by $\tau_{ij}$, where, by \eq{tauRecurrence},
\begin{align}
\tau_{ij} &= \frac{1}{N} + \frac{1}{2N} \sum_{k=1}^{N} p_{ik} \tau_{kj} + \frac{1}{2N} \sum_{k=1}^{N} p_{jk} \tau_{ik} + \left(1-\frac{1}{N}\right) \tau_{ij} ,
\end{align}
which implies that $\tau_{ij}=1+\left(1/2\right)\sum_{k=1}^{N}p_{ik}\tau_{kj}+\left(1/2\right)\sum_{k=1}^{N}p_{jk}\tau_{ik}$ whenever $i\neq j$. (For $i=j$, we have $\tau_{ii}=0$ for $i=1,\dots ,N$). Using this substitution, we obtain \eq{tauConditionPC_mainText} in the main text.

\subsection{Death-birth (DB) updating}\label{subsec:DBrule}
Under DB updating (see \fig{update_rules}), the probability of replacement event $\left(R,\alpha\right)$ is
\begin{align}
p_{\left(R,\alpha\right)}\left(\vx\right) &=
\begin{cases}
\frac{1}{N} \frac{w_{i\alpha\left(i\right)}F_{\alpha\left(i\right)}\left(\vx\right)}{\sum_{k=1}^{N}w_{ik}F_{k}\left(\vx\right)} & R=\left\{i\right\}\textrm{ for some }i\in\left\{1,\dots ,N\right\} , \\
& \\
0 & \textrm{otherwise} .
\end{cases}
\end{align}
For $i,j=1,\dots ,N$, the marginal probability that $i$ transmits its offspring to $j$ is
\begin{align}\label{eq:marginalDB}
e_{ij}\left(\mathbf{x}\right) &= \frac{1}{N} \frac{w_{ji}F_{i}\left(\vx\right)}{\sum_{k=1}^{N}w_{jk}F_{k}\left(\vx\right)} .
\end{align}
Therefore, the marginal effect of $k$ on $i$ replacing $j$ is
\begin{align}
m_{k}^{ij} &= 
\begin{cases}
\frac{1}{N}p_{ji}\left(1-p_{ji}\right) & k = i , \\
& \\
-\frac{1}{N}p_{ji}p_{jk} & k \neq i .
\end{cases}
\end{align}
Again, the reproductive value of $i$ is $\pi_{i}=w_{i}/\sum_{k=1}^{N}w_{k}$. Since $\pi_{i}p_{ij}=\pi_{j}p_{ji}$ for every $i$ and $j$,
\begin{subequations}
	\begin{align}
	\sum_{i,j,k,\ell =1}^{N} \pi_{i} m_{k}^{ji} \left( -x_{jk}C_{k\ell} + x_{j\ell}B_{\ell k} \right) &= \sum_{i=1}^{N}\pi_{i}\frac{1}{N}\sum_{\ell =1}^{N}\left( -x_{ii}C_{i\ell} + x_{i\ell}B_{\ell i} \right) \nonumber \\
	&\quad - \sum_{i=1}^{N}\pi_{i} \sum_{k=1}^{N}\frac{1}{N}p_{ik}^{\left(2\right)}\sum_{\ell =1}^{N}\left( -x_{ik}C_{k\ell} + x_{i\ell}B_{\ell k} \right) ; \\
	\sum_{i,j,k,\ell =1}^{N} \pi_{i} m_{k}^{ji} \left( -x_{ik}C_{k\ell} + x_{i\ell}B_{\ell k} \right) &= 0 .
	\end{align}
\end{subequations}
Consequently, we see that $\rho_{A}>\rho_{B}$ for small $\delta >0$ if and only if
\begin{align}
\sum_{i=1}^{N} \pi_{i} \sum_{\ell =1}^{N}\left( -x_{ii}C_{i\ell} + x_{i\ell}B_{\ell i} \right) &> \sum_{i,j=1}^{N} \pi_{i} p_{ij}^{\left(2\right)} \sum_{\ell =1}^{N}\left( -x_{ij}C_{j\ell} + x_{i\ell}B_{\ell j} \right) ,
\end{align}
which gives \eq{mainConditionDB_mainText} in the main text. To evaluate this condition, for $i\neq j$ we can replace $1-x_{ij}$ by $\tau_{ij}$, where, by \eq{tauRecurrence},
\begin{align}
\tau_{ij} &= \frac{1}{N} + \frac{1}{N} \sum_{k=1}^{N} p_{ik} \tau_{kj} + \frac{1}{N} \sum_{k=1}^{N} p_{jk} \tau_{ik} + \left(1-\frac{2}{N}\right) \tau_{ij} ,
\end{align}
which implies that $\tau_{ij}=1/2+\left(1/2\right)\sum_{k=1}^{N}p_{ik}\tau_{kj}+\left(1/2\right)\sum_{k=1}^{N}p_{jk}\tau_{ik}$ whenever $i\neq j$. (Again, we have $\tau_{ii}=0$ for $i=1,\dots ,N$.) Since we can evaluate \eq{tauConditionDB_mainText} using any non-zero multiple of $\tau$, it suffices to replace $\tau$ by $2\tau$ and use the recurrence $\tau_{ij}=1+\left(1/2\right)\sum_{k=1}^{N}p_{ik}\tau_{kj}+\left(1/2\right)\sum_{k=1}^{N}p_{jk}\tau_{ik}$ when $i\neq j$, the same as that of PC updating \citep[and also of][]{allen:Nature:2017,fotouhi:NHB:2018,allen:JMB:2019}. This substitution leads to \eq{tauConditionDB_mainText} in the main text.

\subsection{Imitation (IM) updating}\label{subsec:IMrule}
Under IM updating (see \fig{update_rules}), the probability of replacement event $\left(R,\alpha\right)$ is
\begin{align}
p_{\left(R,\alpha\right)}\left(\vx\right) &=
\begin{cases}
\frac{1}{N} \frac{w_{i\alpha\left(i\right)}F_{\alpha\left(i\right)}\left(\vx\right)}{F_{i}\left(\vx\right) +\sum_{k=1}^{N}w_{ik}F_{k}\left(\vx\right)} & R=\left\{i\right\}\textrm{ for some }i\in\left\{1,\dots ,N\right\} ,\ \alpha\left(i\right)\neq i , \\
& \\
\frac{1}{N}\frac{F_{i}\left(\vx\right)}{F_{i}\left(\vx\right) +\sum_{k=1}^{N}w_{ik}F_{k}\left(\vx\right)} & R=\left\{i\right\}\textrm{ for some }i\in\left\{1,\dots ,N\right\} ,\ \alpha\left(i\right) = i , \\
& \\
0 & \textrm{otherwise} .
\end{cases}
\end{align}
For $i,j=1,\dots ,N$, the probability that $i$ transmits its offspring to $j\neq i$ is
\begin{align}\label{eq:marginalIM}
e_{ij}\left(\mathbf{x}\right) &= \frac{1}{N} \frac{w_{ji}F_{i}\left(\vx\right)}{F_{j}\left(\vx\right) +\sum_{k=1}^{N}w_{jk}F_{k}\left(\vx\right)} ,
\end{align}
and the marginal effect of $k$ on $i$ replacing $j\neq i$ is
\begin{align}
m_{k}^{ij} &= 
\begin{cases}
\frac{1}{N}\frac{w_{ji}}{w_{j}+1}\left(1-\frac{w_{ji}}{w_{j}+1}\right) & k = i , \\
& \\
-\frac{1}{N}\left(\frac{w_{ji}}{w_{j}+1}\right)^{2} & k = j , \\
& \\
-\frac{1}{N}\frac{w_{ji}}{w_{j}+1}\frac{w_{jk}}{w_{j}+1} & k \neq i,j .
\end{cases}
\end{align}
Let $\left(\widetilde{w}_{ij}\right)_{i,j=1}^{N}$ be the matrix defined by $\widetilde{w}_{ii}=1$ and $\widetilde{w}_{ij}=w_{ij}$ for $i\neq j$. We also define an analogue of $p_{ij}=w_{ij}/w_{i}$, namely $\widetilde{p}_{ij}\coloneqq\widetilde{w}_{ij}/\widetilde{w}_{i}$. Under IM updating, reproductive value is now $\pi_{i}=\widetilde{w}_{i}/\sum_{k=1}^{N}\widetilde{w}_{k}$. Since $\pi_{i}\widetilde{p}_{ij}=\pi_{j}\widetilde{p}_{ji}$ for every $i$ and $j$, a straightforward calculation gives
\begin{align}
\sum_{i,j,k,\ell =1}^{N} &\pi_{i} m_{k}^{ji} \left( -x_{jk}C_{k\ell} + x_{j\ell}B_{\ell k} \right) - \sum_{i,j,k,\ell =1}^{N} \pi_{i} m_{k}^{ji} \left( -x_{ik}C_{k\ell} + x_{i\ell}B_{\ell k} \right) \nonumber \\
&= \frac{1}{N}\sum_{i,\ell =1}^{N} \pi_{i} \left( -x_{ii}C_{i\ell} + x_{i\ell}B_{\ell i} \right) - \frac{1}{N}\sum_{i,j,\ell =1}^{N} \pi_{i} \widetilde{p}_{ij}^{\left(2\right)} \left( -x_{ij}C_{j\ell} + x_{i\ell}B_{\ell j} \right) 
\end{align}
Therefore, $\rho_{A}>\rho_{B}$ for small $\delta >0$ if and only if 
\begin{align}
\sum_{i=1}^{N}\pi_{i}\sum_{\ell =1}^{N}\left( -x_{ii}C_{i\ell} + x_{i\ell}B_{\ell i} \right) &> \sum_{i,j=1}^{N}\pi_{i} \widetilde{p}_{ij}^{\left(2\right)}\sum_{\ell =1}^{N}\left( -x_{ij}C_{j\ell} + x_{i\ell}B_{\ell j} \right) ,
\end{align}
which gives \eq{mainConditionIM_mainText} in the main text. To evaluate this condition, we replace $1-x_{ij}$ by $\tau_{ij}$, where $\tau_{ii}=0$ for $i=1,\dots ,N$ and, for $i\neq j$,
\begin{align}
\tau_{ij} &= \frac{1}{N} + \frac{1}{N}\sum_{k=1}^{N} \widetilde{p}_{ik} \tau_{kj} + \frac{1}{N} \sum_{k=1}^{N} \widetilde{p}_{jk} \tau_{ik} + \left(1-\frac{2}{N}\right) \tau_{ij} ,
\end{align}
which gives the recurrence $\tau_{ij}=1/2+\left(1/2\right)\sum_{k=1}^{N}\widetilde{p}_{ik}\tau_{kj}+\left(1/2\right)\sum_{k=1}^{N}\widetilde{p}_{jk}\tau_{ik}$. Again, we can replace $\tau$ by $2\tau$ (since \eq{tauConditionIM_mainText} can be evaluated using any non-zero multiple of $\tau$) and use $\tau_{ij}=1+\left(1/2\right)\sum_{k=1}^{N}\widetilde{p}_{ik}\tau_{kj}+\left(1/2\right)\sum_{k=1}^{N}\widetilde{p}_{jk}\tau_{ik}$ when $i\neq j$. We therefore obtain \eq{tauConditionIM_mainText} in the main text. We note that \eq{tauConditionIM_mainText}, as well as $\tau$, can be obtained from the corresponding results for DB updating by replacing $w$ with $\widetilde{w}$ and $p$ with $\widetilde{p}$.

\section{Specific social goods}\label{sec:examples_socialGoods}
For the three main kinds of social goods we consider, we have $B_{ij}=b\beta_{ij}$ and $C_{ij}=c\gamma_{ij}$, where
\begin{align}
\left( \beta_{ij},\gamma_{ij} \right) = 
\begin{cases}
\left( w_{ij} , w_{ij} \right) & \textrm{pp}, \\
& \\
\left( w_{ij}/w_{i} , w_{ij}/w_{i} \right) & \textrm{ff}, \\
& \\
\left( w_{ij} , w_{ij}/w_{i} \right) & \textrm{pf}.
\end{cases} \label{eq:threeDonations}
\end{align}
When considering the differences between fixed and proportional benefits and costs, there is one more case: fp (fixed benefits and proportional costs), i.e. $\left(\beta_{ij},\gamma_{ij}\right) =\left( w_{ij}/w_{i},w_{ij}\right)$. However, this kind of good is somewhat less natural than the others because it involves the production of a good with fixed total benefit whose cost rises with the number of recipients. A good of this form could model a situation in which the cost is tied to the transmission of the good rather than its production. For example, consider a rival, divisible good that is either costless to produce (or else is readily available to a donor, who does not need to ``produce'' it). Then, if there are $k$ neighbors, each one gets $b/k$, where $b$ is the benefit of the good. But if the process of transmitting or delivering the good costs $c$ per recipient, then the total benefit is $b$ and the total cost is $kc$. While our general theory can account for this kind of good, we focus primarily on the other three.

\subsection{PC updating}
On regular graphs of degree $d$, all individuals have exactly $d$ neighbors, and there is a simple correspondence between different kinds of donations. For pp-goods, a producer pays $dc$ in total in order to donate $b$ to each and every neighbor (for a total of $db$). Thus, ff-goods can be obtained from pp-goods by scaling both $b$ and $c$ by $1/d$. pf-goods can be obtained from pp-goods by scaling $c$ (but not $b$) by $1/d$. But we can say more: from the recurrence $\tau_{ij}=1+\left(1/2\right)\sum_{k=1}^{N}p_{ik}\tau_{kj}+\left(1/2\right)\sum_{k=1}^{N}p_{jk}\tau_{ik}$, we see that
\begin{align}
\sum_{i,j,\ell =1}^{N} \pi_{i} p_{ij} p_{\ell j} \tau_{i\ell} &= \sum_{i,j=1}^{N} \pi_{i} p_{ij} \tau_{ij} - 1 .
\end{align}
Therefore,
\begin{align}
\sum_{i,j=1}^{N} \pi_{i} p_{ij} \sum_{\ell =1}^{N}\left( \tau_{i\ell} - \tau_{j\ell} \right) p_{\ell j} &= -1 .
\end{align}
It follows that the denominator of $\left(b/c\right)^{\ast}$ in \eq{tauConditionPC_mainText} is negative for pp-, ff-, pf-, and fp-goods on regular graphs. Thus, regular graphs can never support producers over non-producers under PC updating.

We now turn to a couple of examples of PC updating on heterogeneous graphs:

\begin{example}[Cluster of stars]
	On heterogeneous graphs, the results are much more interesting. Consider, for example, a cluster of stars conjoined by a complete graph at their hubs (\sfig{clusterOfStars}). Let $i\sim j$ indicate that $i$ and $j$ reside on the same star, and let $H$ and $L$ be the set of hubs and leaves of the structure, respectively. If there are $m$ stars in total, each of size $n$, then
	\begin{align}
	\tau_{ij} &= 
	\begin{cases}
	\frac{m^{3} + 4m^{2}n - 3m^{2} + 4mn^{2} - 6mn - 2m + 2n^{2} - 9n + 8}{2m^{2} + 4mn - 6m + 2n^{2} - 5n + 4} & i \in H,\ j \in L ,\ i\sim j , \\
	& \\
	\frac{2m^{3} + 9m^{2}n - 10m^{2} + 12mn^{2} - 28mn + 16m + 4n^{3} - 14n^{2} + 16n - 8}{2m^{2} + 4mn - 6m + 2n^{2} - 5n + 4} & i \in H,\ j \in L ,\ i\not\sim j , \\
	& \\
	\frac{2m^{3} + 9m^{2}n - 13m^{2} + 12mn^{2} - 34mn + 25m + 4n^{3} - 16n^{2} + 21n - 10}{2m^{2} + 4mn - 6m + 2n^{2} - 5n + 4} & i \in H,\ j \in H ,\ i\neq j , \\
	& \\
	\frac{m^{3} + 4m^{2}n - m^{2} + 4mn^{2} - 2mn - 8m + 4n^{2} - 14n + 12}{2m^{2} + 4mn - 6m + 2n^{2} - 5n + 4} & i \in L,\ j \in L ,\ i\sim j ,\ i\neq j , \\
	& \\
	\frac{2m^{3} + 9m^{2}n - 8m^{2} + 12mn^{2} - 24mn + 10m + 4n^{3} - 12n^{2} + 11n - 4}{2m^{2} + 4mn - 6m + 2n^{2} - 5n + 4} & i \in L,\ j \in L ,\ i\not\sim j , \\
	& \\
	0 & i = j
	\end{cases} \label{eq:tauStarCluster}
	\end{align}
	by \eq{tauRecurrence}. These quantities give all of the times $\tau$ by symmetry.
	
	Using \eq{tauConditionPC_mainText}, we find that
	\begin{align}\label{eq:criticalRatiosStarClusterPC}
	\left(\frac{b}{c}\right)^{\ast} &= \begin{cases}
	\frac{\left(\substack{2m^{2} + 4mn - 6m \\ + 2n^{2} - 5n + 4}\right)\left(\substack{m^{3} + 4m^{2}n - 7m^{2} + 4mn^{2} \\ - 14mn + 13m - n^{2} + 3n - 3}\right)}{\left(\substack{m^{4}n - 3m^{4} + 5m^{3}n^{2} - 21m^{3}n + 22m^{3} + 8m^{2}n^{3} - 47m^{2}n^{2} + 85m^{2}n \\ - 53m^{2} + 4mn^{4} - 34mn^{3} + 88mn^{2} - 96mn + 42m - 2n^{4} + 7n^{3} - 5n^{2} - n}\right)} & \textrm{pp}, \\
	& \\
	\frac{\left(\substack{2m^{5} + 13m^{4}n - 21m^{4} + 31m^{3}n^{2} - 96m^{3}n + 78m^{3} \\ + 32m^{2}n^{3} - 145m^{2}n^{2} + 221m^{2}n - 119m^{2} + 12mn^{4} \\ - 72mn^{3} + 153mn^{2} - 144mn + 56m - 6n^{3} + 25n^{2} - 32n + 12}\right)}{\left(\substack{m^{5}n - m^{5} + 6m^{4}n^{2} - 17m^{4}n + 9m^{4} + 13m^{3}n^{3} - 61m^{3}n^{2} + 83m^{3}n - 29m^{3} \\ + 12m^{2}n^{4} - 81m^{2}n^{3} + 180m^{2}n^{2} - 157m^{2}n + 39m^{2} + 4mn^{5} - 38mn^{4} \\ + 120mn^{3} - 167mn^{2} + 103mn - 18m - 2n^{5} + 11n^{4} - 21n^{3} + 18n^{2} - 7n}\right)} & \textrm{ff}, \\
	& \\
	\frac{\left(\substack{2m^{4} + 11m^{3}n - 17m^{3} + 20m^{2}n^{2} - 57m^{2}n + 44m^{2} \\ + 12mn^{3} - 48mn^{2} + 63mn - 31m - 6n^{2} + 13n - 6}\right)}{\left(\substack{m^{4}n - 3m^{4} + 5m^{3}n^{2} - 21m^{3}n + 22m^{3} + 8m^{2}n^{3} - 47m^{2}n^{2} + 85m^{2}n \\ - 53m^{2} + 4mn^{4} - 34mn^{3} + 88mn^{2} - 96mn + 42m - 2n^{4} + 7n^{3} - 5n^{2} - n}\right)} & \textrm{pf}.
	\end{cases}
	\end{align}
	As $n$ grows, this critical ratio approaches $\left(4m-1\right) /\left(2m-1\right)$ for pp-goods. For pf-goods and ff-goods, this ratio is asymptotic to $\left(6m/\left(2m-1\right)\right) /n$ as $n\rightarrow\infty$; in particular, it approaches $0$. Therefore, for any $b,c>0$, there exists $n$ such that producers are favored over non-producers.
\end{example}

\begin{example}[Rich club]
	The ``rich club'' \citep{fotouhi:NHB:2018} is a structure consisting of a well-connected group of $m$ individuals, together with $n$ individuals at the periphery (see \fig{richClub}). Each peripheral individual is connected to all $m$ members of the central rich club and nobody else. Every member of the rich club is connected to all other individuals in the population. An example of such a structure is shown in \fig{richClub}. A straightforward calculation shows that the unique solution to \eq{tauRecurrence} with $\tau_{ii}=0$ for $i=1,\dots ,N$ is
	\begin{align}
	\tau_{ij} &= 
	\begin{cases}
	\frac{m^{3} + 4m^{2}n - m^{2} + 4mn^{2} - 4mn - n^{2} + 2n - 1}{m^{2} + 3mn + n^{2} - n} & i \in R,\ j \in P , \\
	& \\
	\frac{m^{3} + 4m^{2}n + 4mn^{2} - mn + n - 1}{m^{2} + 3mn + n^{2} - n} & i \in P,\ j \in P ,\ i\neq j , \\
	& \\
	\frac{m\left(m^{2} + 4mn - m + 4n^{2} - 4n\right)}{m^{2} + 3mn + n^{2} - n} & i \in R,\ j \in R ,\ i\neq j , \\
	& \\
	0 & i = j .
	\end{cases} \label{eq:tauRichClub}
	\end{align}
	Again, these quantities give all of the times $\tau$ by symmetry.
	
	By \eq{tauConditionPC_mainText}, it follows that
	\begin{align}\label{eq:criticalRatiosRichClubPC}
	\left(\frac{b}{c}\right)^{\ast} &= \begin{cases}
	\frac{\left(m^{2} + 3mn + n^{2} - n\right)\left(m^{3} + 4m^{2}n - 3m^{2} + 4mn^{2} - 7mn + 3m - n^{2} + 2n - 1\right)}{- m^{4} - 6m^{3}n + 2m^{3} - 11m^{2}n^{2} + 9m^{2}n - m^{2} - 5mn^{3} + 9mn^{2} - 4mn + n^{4} - 3n^{2} + 2n} & \textrm{pp}, \\
	& \\
	\frac{m\left(m + n - 1\right)\left(\substack{m^{4} + 6m^{3}n - 2m^{3} + 12m^{2}n^{2} - 10m^{2}n + m^{2} \\ + 8mn^{3} - 12mn^{2} + 4mn - 2n^{3} + 4n^{2} - 2n}\right)}{\left(\substack{- m^{5} - 6m^{4}n + 2m^{4} - 12m^{3}n^{2} + 10m^{3}n - m^{3} - 8m^{2}n^{3} \\ + 12m^{2}n^{2} - 4m^{2}n + 2mn^{3} - 4mn^{2} + 2mn + n^{5} - 2n^{4} + 2n^{2} - n}\right)} & \textrm{ff}, \\
	& \\
	\frac{m^{4} + 6m^{3}n - 2m^{3} + 12m^{2}n^{2} - 10m^{2}n + m^{2} + 8mn^{3} - 12mn^{2} + 4mn - 2n^{3} + 4n^{2} - 2n}{- m^{4} - 6m^{3}n + 2m^{3} - 11m^{2}n^{2} + 9m^{2}n - m^{2} - 5mn^{3} + 9mn^{2} - 4mn + n^{4} - 3n^{2} + 2n} & \textrm{pf}.
	\end{cases}
	\end{align}
	As $n$ grows, this ratio approaches $4m-1$ for pp-goods. This ratio is asymptotic to $\left(2m\left(4m-1\right)\right) /n$ for ff-goods and $\left(2\left(4m-1\right)\right) /n$ for pf-goods.
\end{example}

\subsection{DB updating}
For pp-goods on a $d$-regular graph, $\left(b/c\right)^{\ast}=d\left(N-2\right) /\left(N-2d\right)$ \citep{chen:AAP:2013,chen:SR:2016,mcavoy:preprint:2019a}, which therefore gives $\left(b/c\right)^{\ast}$ for ff-goods as well. For pf-goods, $\left(b/c\right)^{\ast}=\left(N-2\right) /\left(N-2d\right)$. For all three kinds of donation, it is clear that selection can favor producers over non-producers on a $d$-regular graph only if $b>c$. However, we can find heterogeneous graphs on which producers are favored over non-producers even when $b<c$.
\begin{example}[Cluster of stars]
	By Equations~\ref{eq:tauStarCluster}~and~\ref{eq:tauConditionDB_mainText}, under DB updating we have
	\begin{align}\label{eq:criticalRatiosStarClusterDB}
	\left(\frac{b}{c}\right)^{\ast} &= \begin{cases}
	\frac{\left(\substack{2m^{5} + 11m^{4}n - 21m^{4} + 21m^{3}n^{2} - 80m^{3}n + 78m^{3} \\ + 16m^{2}n^{3} - 93m^{2}n^{2} + 186m^{2}n - 127m^{2} + 4mn^{4} - 34mn^{3} \\ + 114mn^{2} - 171mn + 96m - 4n^{4} + 26n^{3} - 70n^{2} + 82n - 36}\right)}{\left(m - 1\right)\left(\substack{m^{3}n - 3m^{3} + 5m^{2}n^{2} - 17m^{2}n + 18m^{2} + 8mn^{3} \\ - 34mn^{2} + 58mn - 37m + 4n^{4} - 22n^{3} + 52n^{2} - 62n + 30}\right)} & \textrm{pp}, \\
	& \\
	\frac{\left(\substack{2m^{3} + 6m^{2}n - 10m^{2} + 6mn^{2} \\ - 19mn + 16m + 2n^{3} - 9n^{2} + 14n - 8}\right)\left(\substack{2m^{4} + 11m^{3}n - 19m^{3} + 20m^{2}n^{2} \\ - 65m^{2}n + 56m^{2} + 12mn^{3} - 58mn^{2} \\ + 92mn - 53m - 4n^{3} + 10n^{2} - 10n + 6}\right)}{\left(m-1\right)\left(\substack{2m^{2} + 4mn - 6m \\ + 2n^{2} - 5n + 4}\right)\left(\substack{m^{4}n - m^{4} + 6m^{3}n^{2} - 16m^{3}n + 8m^{3} + 13m^{2}n^{3} \\ - 55m^{2}n^{2} + 75m^{2}n - 27m^{2} + 12mn^{4} - 70mn^{3} + 155mn^{2} \\ - 150mn + 48m + 4n^{5} - 30n^{4} + 92n^{3} - 144n^{2} + 116n - 36}\right)} & \textrm{ff}, \\
	& \\
	\frac{\left(\substack{2m^{4} + 11m^{3}n - 19m^{3} + 20m^{2}n^{2} - 65m^{2}n + 56m^{2} \\ + 12mn^{3} - 58mn^{2} + 92mn - 53m - 4n^{3} + 10n^{2} - 10n + 6}\right)}{\left(m-1\right)\left(\substack{m^{3}n - 3m^{3} + 5m^{2}n^{2} - 17m^{2}n + 18m^{2} + 8mn^{3} \\ - 34mn^{2} + 58mn - 37m + 4n^{4} - 22n^{3} + 52n^{2} - 62n + 30}\right)} & \textrm{pf}.
	\end{cases}
	\end{align}
	As $n\rightarrow\infty$, this ratio approaches $1$ under pp. For both pf-goods and ff-goods, this ratio is asymptotic to $\left(\left(3m-1\right) /\left(m-1\right)\right) /n$ as $n\rightarrow\infty$. Thus, for any $m>1$, there exists $n$ for which this population structure can support producers over non-producers even when $0<b<c$.
\end{example}

\begin{example}[Rich club]
	For the rich club evolving according to DB updating,
	\begin{align}\label{eq:criticalRatiosRichClubDB}
	\left(\frac{b}{c}\right)^{\ast} &= \begin{cases}
	\frac{\left(\substack{- m^{5} - 7m^{4}n + 4m^{4} - 17m^{3}n^{2} + 23m^{3}n - 5m^{3} - 16m^{2}n^{3} + 40m^{2}n^{2} \\ - 23m^{2}n + 2m^{2} - 4mn^{4} + 21mn^{3} - 29mn^{2} + 12mn + 4n^{4} - 9n^{3} + 6n^{2} - n}\right)}{\left(m - 1\right)\left(m + 3n - 2\right)\left(m^{2} + 3mn + n^{2} - n\right)} & \textrm{pp}, \\
	& \\
	\frac{m\left(m + n - 1\right)\left(\substack{- m^{4} - 6m^{3}n + 3m^{3} - 12m^{2}n^{2} + 15m^{2}n - 2m^{2} \\ - 8mn^{3} + 19mn^{2} - 8mn + 4n^{3} - 7n^{2} + 3n}\right)}{\left(m-1\right)\left(\substack{m^{4} + 6m^{3}n - 2m^{3} + 11m^{2}n^{2} - 8m^{2}n \\ + 5mn^{3} - 6mn^{2} + mn - n^{4} + n^{3} + n^{2} - n}\right)} & \textrm{ff}, \\
	& \\
	\frac{\left(\substack{- m^{4} - 6m^{3}n + 3m^{3} - 12m^{2}n^{2} + 15m^{2}n - 2m^{2} \\ - 8mn^{3} + 19mn^{2} - 8mn + 4n^{3} - 7n^{2} + 3n}\right)}{\left(m - 1\right)\left(m + 3n - 2\right)\left(m^{2} + 3mn + n^{2} - n\right)} & \textrm{pf}.
	\end{cases}
	\end{align}
	For fixed $m>1$, we see that
	\begin{align}
	\lim_{n\rightarrow\infty}\left(\frac{b}{c}\right)^{\ast} &= \begin{cases}
	-\infty & \textrm{pp}, \\
	& \\
	\frac{4m\left(2m - 1\right)}{m - 1} & \textrm{ff}, \\
	& \\
	-\frac{4\left(2m - 1\right)}{3\left(m - 1\right)} & \textrm{pf}.
	\end{cases}
	\end{align}
\end{example}

\subsection{On the necessity of $b>c$ for pp-goods}
When $0<b\leqslant c$, pp-goods have the property that no producer can have a payoff of more than $0$. Since every non-producer has a payoff of at least $0$, it follows that every non-producer has at least the payoff of the best-performing producer in the population. Informally speaking, when higher payoffs result in more reproductive success, it should follow that producers of pp-goods cannot be favored over non-producers when $0<b\leqslant c$. Here, we make this claim more formally.

Recall that the expected change RV-weighted abundance of $A$ due to selection is
\begin{align}
\delhatsel\left(\mathbf{x}\right) &= \sum_{i,j=1}^{N}\pi_{i}\left( x_{j}-x_{i}\right) e_{ji}\left(\vx\right) \nonumber \\
&= \sum_{i,j=1}^{N}\pi_{i}\left( \left(1-x_{i}\right) x_{j} - x_{i}\left(1-x_{j}\right) \right) e_{ji}\left(\vx\right) \nonumber \\
&= \sum_{i=1}^{N}\pi_{i} \left[ \left(1-x_{i}\right) \sum_{j=1}^{N} x_{j} e_{ji}\left(\vx\right) - x_{i} \sum_{j=1}^{N} \left(1-x_{j}\right) e_{ji}\left(\vx\right) \right] .
\end{align}
The term $\sum_{j=1}^{N} x_{j} e_{ji}\left(\vx\right)$ is the probability that a producer replaces $i$, while $\sum_{j=1}^{N} \left(1-x_{j}\right) e_{ji}\left(\vx\right)$ is the probability that a non-producer replaces $i$. When $i$ is a non-producer the former should be reduced by selection; when $i$ is a producer, the latter should be increased by selection. This property, of course, is not guaranteed to hold for \textit{any} replacement rule, but it is reasonable to expect it to hold in models for which higher payoffs result in greater competitive abilities.

For example, this property is easily seen to hold for PC, DB, and IM updating for pp-goods when $b\leqslant c$. We do not include all of the details here, but it is straightforward to verify this claim using Equations~\ref{eq:marginalNotSelfPC},~\ref{eq:marginalDB},~and~\ref{eq:marginalIM}. Since then $\delhatsel '\left(\mathbf{x}\right)\leqslant 0$ for every $\vx\in\left\{0,1\right\}^{N}$, producers cannot be favored over non-producers; such behavior would require $b>c$.

The same argument works for fp-goods as well. However, we have already seen numerous examples of when selection favors producers of ff- and pf-goods even when $0<b\leqslant c$. The reason the argument presented here does not apply to ff- and pf-goods is that there can still be states in which some producers have higher payoffs than non-producers when $0<b\leqslant c$.

\subsection{Monomorphic states and prosocial inequality}
Let $\mathbf{A}$ and $\mathbf{B}$ denote the all-$A$ and all-$B$ states, respectively. As a simple example of what the monomorphic states look like in a heterogeneous population, consider the star graph (\fig{richClub} or \sfig{clusterOfStars} with $m=1$). For all three kinds of donation, $u_{\textrm{hub}}\left(\mathbf{B}\right) =u_{\textrm{leaf}}\left(\mathbf{B}\right)=0$. We look at the other monomorphic state, $\mathbf{A}$, separately:

For pp-goods, we have $u_{\textrm{hub}}\left(\mathbf{A}\right) =\left(N-1\right)\left(b-c\right)$ and $u_{\textrm{leaf}}\left(\mathbf{A}\right) =b-c$, which gives
\begin{subequations}
	\begin{align}
	\overline{u}\left(\mathbf{A}\right) &= 2\left(1-\frac{1}{N}\right)\left(b-c\right) ; \\
	u_{\textrm{hub}}\left(\mathbf{A}\right) -u_{\textrm{leaf}}\left(\mathbf{A}\right) &= \left(N-2\right)\left(b-c\right) ; \\
	\max\left\{ u_{\textrm{hub}}\left(\mathbf{A}\right) , u_{\textrm{leaf}}\left(\mathbf{A}\right)\right\} &= \begin{cases}\left(N-1\right)\left(b-c\right) & b\geqslant c , \\ & \\ b-c & b<c ; \end{cases} \\
	\min\left\{ u_{\textrm{hub}}\left(\mathbf{A}\right) , u_{\textrm{leaf}}\left(\mathbf{A}\right)\right\} &= \begin{cases}b-c & b\geqslant c , \\ & \\ \left(N-1\right)\left(b-c\right) & b<c . \end{cases}
	\end{align}
\end{subequations}

For ff-goods, $u_{\textrm{hub}}\left(\mathbf{A}\right) =\left(N-1\right) b-c$ and $u_{\textrm{leaf}}\left(\mathbf{A}\right) =b/\left(N-1\right) -c$, giving
\begin{subequations}
	\begin{align}
	\overline{u}\left(\mathbf{A}\right) &= b-c ; \\
	u_{\textrm{hub}}\left(\mathbf{A}\right) -u_{\textrm{leaf}}\left(\mathbf{A}\right) &= \frac{N\left(N-2\right)}{N-1} b ; \\
	\max\left\{ u_{\textrm{hub}}\left(\mathbf{A}\right) , u_{\textrm{leaf}}\left(\mathbf{A}\right)\right\} &= \left(N-1\right) b-c ; \\
	\min\left\{ u_{\textrm{hub}}\left(\mathbf{A}\right) , u_{\textrm{leaf}}\left(\mathbf{A}\right)\right\} &= \frac{1}{N-1}b-c .
	\end{align}
\end{subequations}

Finally, for pf-goods, $u_{\textrm{hub}}\left(\mathbf{A}\right) =\left(N-1\right) b-c$ and $u_{\textrm{leaf}}\left(\mathbf{A}\right) =b-c$, so we have
\begin{subequations}
	\begin{align}
	\overline{u}\left(\mathbf{A}\right) &= 2\left(1-\frac{1}{N}\right) b-c ; \\
	u_{\textrm{hub}}\left(\mathbf{A}\right) -u_{\textrm{leaf}}\left(\mathbf{A}\right) &= \left(N-2\right) b ; \\
	\max\left\{ u_{\textrm{hub}}\left(\mathbf{A}\right) , u_{\textrm{leaf}}\left(\mathbf{A}\right)\right\} &= \left(N-1\right) b-c ; \\
	\min\left\{ u_{\textrm{hub}}\left(\mathbf{A}\right) , u_{\textrm{leaf}}\left(\mathbf{A}\right)\right\} &= b-c .
	\end{align}
\end{subequations}

\subsection{Accumulated versus averaged payoffs}\label{sec:accumulatedAveraged}
For each kind of social good considered thus far, we have used accumulation to determine overall payoff in a population. An alternative method, which is popular for pp-goods \citep{szabo:PR:2007,tomassini:IJMPC:2007,antonioni:ACS:2012,allen:AN:2013,wu:Games:2013,maciejewski:PLoSCB:2014,allen:Nature:2017}, involves averaging these payoffs instead of adding them. In other words, if, under accumulated payoffs, $B_{ij}$ is the benefit to $j$ due to the donation of $i$ and $C_{ij}$ is the corresponding cost (to $i$), then, under averaged payoffs, these terms undergo the transformation
\begin{subequations}
	\begin{align}
	\overline{B}_{ij} &= \frac{B_{ij}}{w_{j}} ; \\
	\overline{C}_{ij} &= \frac{C_{ij}}{w_{i}} .
	\end{align}
\end{subequations}
Using averaged payoffs, we obtain $\rho_{A}'>\rho_{B}'$ under PC updating if and only if $\gamma b>\beta c$, where
\begin{align}\label{eq:criticalRatios}
\beta\,\, ;\,\, \gamma &= 
\begin{cases}
\sum_{i,j=1}^{N}\pi_{i}p_{ij}\tau_{ij}\,\, ; \,\, -1 & \textrm{pp}, \\
& \\
\sum_{i,j=1}^{N}\pi_{i}p_{ij}\tau_{ij}\frac{1}{w_{j}}\,\, ;\,\, \sum_{i,j,\ell =1}^{N} \pi_{i}p_{ij}\left(\tau_{i\ell}-\tau_{j\ell}\right)\frac{p_{j\ell}}{w_{\ell}} & \textrm{ff}, \\
& \\
\sum_{i,j=1}^{N}\pi_{i}p_{ij}\tau_{ij}\frac{1}{w_{j}}\,\, ;\,\, -1 & \textrm{pf}.
\end{cases}
\end{align}

For both pp-goods and pf-goods with payoff averaging, we have $\gamma <0$, which means that producers can never be favored over non-producers in the limit of weak selection. In contrast, for ff-goods with averaging, there do exist structures on which $\gamma >0$. For example, on the star of size $N$, we have $\left(b/c\right)^{\ast}=\left(3N^{2} - 4N\right) /\left(N^{2} - 8N + 8\right)$, which converges to $3$ as $N\rightarrow\infty$. Coincidentally, this ratio is the same as that of pp-goods with accumulated payoffs on the star (c.f. \eq{criticalRatiosStarClusterPC}). In general, however, the conditions for $\rho_{A}>\rho_{B}$ are distinct in these two cases (pp-goods with accumulation and ff-goods with averaging).

\subsection{Additional contribution to a public pool}
Consider ff-goods on a star. Suppose that for some $\theta\in\left[0,1\right]$, each producer contributes $\left(1-\theta\right) b$ to their neighborhood and $\theta b$ to a common pool. The overall cost is $c$, and the common pool gets divided among all members of the population. This scenario is depicted in \sfig{publicPool}. If $\theta$ is the contribution rate to the common pool, then, under PC updating, selection favors producers over non-producers on a star of size $N$ when
\begin{align}
\frac{b}{c} > \left(\frac{b}{c}\right)_{\theta}^{\ast} \coloneqq \frac{1}{1-\theta}\frac{6N^{2}-14N+8}{N^{3}-6N^{2}+6N} .
\end{align}
In the all-non-producer state, everyone has a payoff of $0$. In the all-producer state, the payoff to the hub is $\left(N-1\right)\left(1-\theta\right) b-c+\theta b$ and the payoff each leaf is $\left(1-\theta\right) b/\left(N-1\right) -c+\theta b$. The payoff to each leaf individual is at least zero if and only if $b>c$ and
\begin{align}
\theta &> \theta^{\ast} = \frac{\left(N-1\right) c-b}{b\left(N-2\right)} .
\end{align}
Note that $0<\theta^{\ast}<1$ if and only if $c<b<\left(N-1\right) c$. At the contribution level $\theta =\theta^{\ast}$, the hub gets $N\left(b-c\right)$ and the leaves all get $0$ in the all-producer state.

For example, suppose that $b=2$ and $c=1$. Then,
\begin{align}
\theta^{\ast} &= \frac{N-3}{2\left(N-2\right)} ,
\end{align}
which approaches $c/b=1/2$ as $N\rightarrow\infty$. Therefore, everyone is at least as well off as in the all-non-producer state, and moreover producers can evolve whenever $N$ is not too small, since
\begin{align}
\frac{b}{c} = 2 > \left(\frac{b}{c}\right)^{\ast} = \frac{2\left(N-2\right)}{N-1} \frac{6N^{2}-14N+8}{N^{3}-6N^{2}+6N} \asymp \frac{12}{N} .
\end{align}
Note that for pp-goods without contribution to a common pool, we require $b/c>3$ on the star.

\section{Reciprocity}
We now consider more complex behavioral types as outlined in the model of reciprocity described in the main text. Let $A$ denote the strategy ``tit-for-tat'' (TFT) and let $B$ denote the strategy ``always defect'' (ALLD). The payoff to player $i$ in the $t$th round of the game is
\begin{align}
u_{i}^{t}\left(\mathbf{x}\right) &= 
\begin{cases}
\sum_{j=1}^{N} \left( - x_{i} C_{ij} + x_{j} B_{ji} \right) & t = 1 , \\
& \\
\sum_{j=1}^{N} x_{i}x_{j} \left( - C_{ij} + B_{ji} \right) & t > 1 .
\end{cases}
\end{align}
If the discounting factor (continuation probability) is $\lambda\in\left[0,1\right)$, then the overall payoff is
\begin{align}
u_{i}\left(\mathbf{x}\right) &= \left(1-\lambda\right) \left( \sum_{j=1}^{N} \left( - x_{i} C_{ij} + x_{j} B_{ji} \right) + \sum_{t=1}^{\infty}\lambda^{t} \sum_{j=1}^{N} x_{i}x_{j} \left( - C_{ij} + B_{ji} \right) \right) \nonumber \\
&= \left(1-\lambda\right) \sum_{j=1}^{N} \left( - x_{i} C_{ij} + x_{j} B_{ji} \right) + \lambda \sum_{j=1}^{N} x_{i}x_{j} \left( - C_{ij} + B_{ji} \right) .
\end{align}

The case $\lambda =0$ recovers the standard setup of social goods, where $A$ is a producer and $B$ is a non-producer. Since intermediate $\lambda$ can be captured by a linear combination of $\lambda =0$ and $\lambda =1$, we now focus on the condition for $\lambda =1$. When $\lambda =1$, the $\delta$-derivative of the change in the RV-weighted abundance of $A$ due to selection at $\delta =0$ is
\begin{align}
\delhatsel '\left(\mathbf{x}\right) &= \sum_{i,j,k=1}^{N} \pi_{i} m_{k}^{ji} \left( x_{j}-x_{i}\right) u_{k}\left(\vx\right) \nonumber \\
&= \sum_{i,j,k=1}^{N} \pi_{i} m_{k}^{ji} \left( x_{j}-x_{i}\right) \sum_{\ell =1}^{N} x_{k}x_{\ell} \left( - C_{k\ell} + B_{\ell k} \right) \nonumber \\
&= \sum_{i,j,k,\ell =1}^{N} \pi_{i} m_{k}^{ji} \left( - C_{k\ell} + B_{\ell k} \right) \left( x_{j}x_{k}x_{\ell}-x_{i}x_{k}x_{\ell}\right) .
\end{align}
Using the identity \citep{allen:JMB:2019}
\begin{align}
\mathbb{E}_{\RMC}^{\circ}\left[ x_{i}x_{j}x_{k} \right] &= \frac{1}{2}\left(\mathbb{E}_{\RMC}^{\circ}\left[ x_{i}x_{j} \right] + \mathbb{E}_{\RMC}^{\circ}\left[ x_{i}x_{k} \right] + \mathbb{E}_{\RMC}^{\circ}\left[ x_{j}x_{k} \right]\right) - \frac{1}{4} ,
\end{align}
we see that
\begin{align}
\mathbb{E}_{\RMC}^{\circ}\left[ x_{j}x_{k}x_{\ell} - x_{i}x_{k}x_{\ell} \right] &= \frac{1}{2}\Big( \mathbb{E}_{\RMC}^{\circ}\left[ x_{j}x_{k} \right] + \mathbb{E}_{\RMC}^{\circ}\left[ x_{j}x_{\ell} \right] \nonumber \\
&\quad\quad - \mathbb{E}_{\RMC}^{\circ}\left[ x_{i}x_{k} \right] - \mathbb{E}_{\RMC}^{\circ}\left[ x_{i}x_{\ell} \right] \Big) .
\end{align}
With $x_{ij}=2\mathbb{E}_{\RMC}^{\circ}\left[ x_{i}x_{j}\right]$, it follows that
\begin{align}
\mathbb{E}_{\RMC}^{\circ}\left[\delhatsel '\right] &= \frac{1}{4}\sum_{i,j,k,\ell =1}^{N} \pi_{i} m_{k}^{ji} \left( - C_{k\ell} + B_{\ell k} \right) \left( x_{jk} + x_{j\ell} - x_{ik} - x_{i\ell} \right) .
\end{align}
Therefore, for $\lambda =1$, we see that $\mathbb{E}_{\RMC}^{\circ}\left[\delhatsel '\right] >0$ if and only if
\begin{align}
\sum_{i,j,k,\ell =1}^{N} \pi_{i} m_{k}^{ji} \left( - C_{k\ell} + B_{\ell k} \right) \left( x_{jk} + x_{j\ell} \right) &> \sum_{i,j,k,\ell =1}^{N} \pi_{i} m_{k}^{ji} \left( - C_{k\ell} + B_{\ell k} \right) \left( x_{ik} + x_{i\ell} \right) .
\end{align}
Consequently, for any $\lambda\in\left[0,1\right]$, we have $\mathbb{E}_{\RMC}^{\circ}\left[\delhatsel '\right] >0$ (and thus $\rho_{A}>\rho_{B}$) if and only if
\begin{align}
\sum_{i,j,k,\ell =1}^{N} &\pi_{i} m_{k}^{ji} \left( - \left(x_{jk}+\lambda x_{j\ell}\right) C_{k\ell} + \left(x_{j\ell}+\lambda x_{jk}\right) B_{\ell k} \right) \nonumber \\
&> \sum_{i,j,k,\ell =1}^{N} \pi_{i} m_{k}^{ji} \left( - \left(x_{ik}+\lambda x_{i\ell}\right) C_{k\ell} + \left(x_{i\ell}+\lambda x_{ik}\right) B_{\ell k} \right) ,
\end{align}
which is \eq{generalConditionReciprocity} in the main text.

We now turn to specific update rules. In each case, the terms $\pi_{i}$, $m_{k}^{ji}$, $x_{ij}$, and $\tau_{ij}$ are exactly the same as they were when there was no reciprocity involved (i.e. $\lambda =0$). Therefore, in each case we take these quantities to be the same as they were previously and instead focus on the selection condition.

\subsection{PC updating}
Using the approach in \S\ref{subsec:PCrule}, for PC updating we find that $\rho_{A}'>\rho_{B}'$ if and only if
\begin{align}
\sum_{i=1}^{N} &\pi_{i} \sum_{\ell =1}^{N} \left( - \left(x_{ii}+\lambda x_{i\ell}\right) C_{i\ell} + \left(x_{i\ell}+\lambda x_{ii}\right) B_{\ell i} \right) \nonumber \\
&> \sum_{i,j=1}^{N} \pi_{i} p_{ij} \sum_{\ell =1}^{N} \left( - \left(x_{ij}+\lambda x_{i\ell}\right) C_{j\ell} + \left(x_{i\ell}+\lambda x_{ij}\right) B_{\ell j} \right) .
\end{align}

\subsection{DB updating}
Using the approach in \S\ref{subsec:DBrule}, for DB updating we find that $\rho_{A}'>\rho_{B}'$ if and only if
\begin{align}
\sum_{i,\ell =1}^{N} &\pi_{i} \left( - \left(x_{ii}+\lambda x_{i\ell}\right) C_{i\ell} + \left(x_{i\ell}+\lambda x_{ii}\right) B_{\ell i} \right) \nonumber \\
&> \sum_{i,j,\ell =1}^{N} \pi_{i} p_{ij}^{\left(2\right)} \left( - \left(x_{ij}+\lambda x_{i\ell}\right) C_{j\ell} + \left(x_{i\ell}+\lambda x_{ij}\right) B_{\ell j} \right) .
\end{align}

\subsection{IM updating}
Using the approach in \S\ref{subsec:IMrule}, for IM updating we find that $\rho_{A}'>\rho_{B}'$ if and only if
\begin{align}
\sum_{i,\ell =1}^{N} &\pi_{i} \left( - \left(x_{ii}+\lambda x_{i\ell}\right) C_{i\ell} + \left(x_{i\ell}+\lambda x_{ii}\right) B_{\ell i} \right) \nonumber \\
&> \sum_{i,j,\ell =1}^{N} \pi_{i} \widetilde{p}_{ij}^{\left(2\right)} \left( - \left(x_{ij}+\lambda x_{i\ell}\right) C_{j\ell} + \left(x_{i\ell}+\lambda x_{ij}\right) B_{\ell j} \right) .
\end{align}

\subsection{Undiscounted games ($\lambda =1$)}
When $\lambda =1$, the conditions stated above for PC, DB, and IM updating each have the form
\begin{align}
\sum_{i,\ell =1}^{N} &\pi_{i} \left( - \left(x_{ii}+x_{i\ell}\right) C_{i\ell} + \left(x_{i\ell}+x_{ii}\right) B_{\ell i} \right) \nonumber \\
&> \sum_{i,j,\ell =1}^{N} \pi_{i} P_{ij} \left( - \left(x_{ij}+x_{i\ell}\right) C_{j\ell} + \left(x_{i\ell}+x_{ij}\right) B_{\ell j} \right)
\end{align}
for some $N\times N$ stochastic matrix $\left(P_{ij}\right)_{i,j=1}^{N}$ with $\pi_{i}P_{ij}=\pi_{j}P_{ji}$ for every $i$ and $j$. Replacing $1-x_{ij}$ by $\tau_{ij}$ and rearranging gives an equivalent condition for $\rho_{A}>\rho_{B}$, namely
\begin{align}
\sum_{i,j,\ell =1}^{N} \pi_{i} P_{ij} \left( \tau_{ij}+\tau_{j\ell}-\tau_{i\ell} \right) \left( B_{\ell i}-C_{i\ell} \right) &> 0 .
\end{align}

\begin{lemma}\label{lem:triangle}
	If $N\geqslant 3$, then $\sum_{j=1}^{N}\pi_{i}P_{ij}\left( \tau_{ij}+\tau_{j\ell}-\tau_{i\ell} \right) >0$ for every $i$ and $\ell$.
\end{lemma}
\begin{proof}
	We first show that $1-x_{ij}-x_{j\ell}+x_{i\ell}\geqslant 0$. Note that
	\begin{align}
	1-x_{ij}-x_{j\ell}+x_{i\ell} &= \mathbb{E}_{\RMC}^{\circ}\left[ \substack{x_{i}\left(1-x_{j}\right) + \left(1-x_{i}\right) x_{j} \\ - x_{j}x_{\ell} - \left(1-x_{j}\right)\left(1-x_{\ell}\right) \\ + x_{i}x_{\ell} + \left(1-x_{i}\right)\left(1-x_{\ell}\right) } \right] \nonumber \\
	&= 2\mathbb{E}_{\RMC}^{\circ}\left[ x_{j} - x_{i}x_{j} + x_{i}x_{\ell} - x_{j}x_{\ell} \right] .
	\end{align}
	The term $x_{j} - x_{i}x_{j} + x_{i}x_{\ell} - x_{j}x_{\ell}$ is $1$ if and only if $x_{i}=x_{\ell}$ and $x_{j}\neq x_{i},x_{\ell}$; otherwise, it is $0$. Therefore, its expectation (with respect to any distribution) is non-negative. Furthermore, for any $j$, the state $\mathbf{e}_{j}$ with $\left(\mathbf{e}_{j}\right)_{j}=1$ and $\left(\mathbf{e}_{j}\right)_{k}=0$ for $k\neq j$ satisfies $\pi_{\RMC}^{\circ}\left(\mathbf{e}_{j}\right) >0$ since it arises from state $\vB$ after a mutation to $A$ in location $j$. Thus, $\tau_{ij}+\tau_{j\ell}>\tau_{i\ell}$ whenever $j\neq i,\ell$.
	
	Now, fix $i$ and $\ell$, and let $j\neq i,\ell$ be such that $P_{ij}>0$. Since $P_{ij}$ is the probability of transitioning from $i$ to $j$ in either a one- or two-step random walk, such a choice of $j$ is possible as long as $N\geqslant 3$. We also have $\tau_{ij}+\tau_{j\ell}>\tau_{i\ell}$, hence $\pi_{i}P_{ij}\left( \tau_{ij}+\tau_{j\ell}-\tau_{i\ell} \right) >0$, as desired.
\end{proof}

As a consequence of \lem{triangle}, we see that PMB goods (those with $B_{ji}\geqslant C_{ij}$ for every $i$ and $j$, with strict inequality for at least one pair $\left(i,j\right)$) are always favored by selection when $\lambda =1$. Furthermore, for $B_{ij}=b\beta_{ij}$ and $C_{ij}=c\gamma_{ij}$ (not necessarily PMB), the denominator of
\begin{align}
\left(\frac{b}{c}\right)_{\lambda =1}^{\ast} &= \frac{\sum_{i,j,\ell =1}^{N} \pi_{i} P_{ij} \left( \tau_{ij}+\tau_{j\ell}-\tau_{i\ell} \right)\gamma_{i\ell}}{\sum_{i,j,\ell =1}^{N} \pi_{i} P_{ij} \left( \tau_{ij}+\tau_{j\ell}-\tau_{i\ell} \right)\beta_{\ell i}}
\end{align}
is positive (provided $\beta$ is not identically zero). Thus, the critical ratio gives a lower bound on $b/c$ for the evolution of TFT relative to ALLD (as opposed to an upper bound, which occurs when the denominator of the critical ratio is negative). In particular, for any social good with $B_{ij}=b\beta_{ij}$ and $C_{ij}=c\gamma_{ij}$ for every $i$ and $j$, where $\beta$ and $\gamma$ are independent of $b$ and $c$, producers of this social good can evolve in the undiscounted game whenever $b/c$ is sufficiently large.

\section{Relationship to prior literature}
Broadly speaking, our contribution is twofold: \emph{(i)} a general theory of the evolutionary success of producers and \emph{(ii)} applications of this theory that uncover surprising results about producers of certain social goods. Our framework and results are general in two distinct ways. First, each type ($A$ and $B$) can be producers of arbitrary social goods. If $i$ has type $A$, then he or she provides $B_{ij}$ to $j$ at a cost of $C_{ij}$. If $i$ has type $B$, then he or she gives $b_{ij}$ to $j$ at a cost of $c_{ij}$. Second, we do not assume that evolution proceeds according to a specific update rule (e.g. PC, DB, or IM). We treat specific update rules in the examples, but our conditions for evolutionary success hold in a much more general context, namely for distributions over replacement events, $\left\{p_{\left(R,\alpha\right)}\left(\vx\right)\right\}_{\left(R,\alpha\right)}$.

Here, we review prior related literature for evolutionary games in structured populations and highlight how our work relates to these previous results.

\subsection{Pairwise games in structured populations}
There is an extensive literature on $2 \times 2$ matrix games in graph-structured populations \citep{nowak:Nature:1992,blume:GEB:1993,hauert:Nature:2004,santos:PRL:2005,ohtsuki:Nature:2006,gomez:PRL:2007,szabo:PR:2007,tarnita:JTB:2009,nowak:PTRSB:2009,debarre:NC:2014,allen:Nature:2017}. Such a matrix game can be written in general form as
\begin{align}\label{eq:prisonersDilemma}
\bordermatrix{%
	& A & B \cr
	A &\ R & \ S \cr
	B &\ T & \ P \cr
}\ ,
\end{align}
For $T>R>P>S$, this game is a prisoner's dilemma, with strategy $A$ representing cooperation. When a player has many neighbors, as is the case in graph-structured populations, a total payoff is obtained by either accumulating or averaging the payoffs from many pairwise interactions. This kind of model corresponds to our notion of pp-goods (see also \S\ref{sec:accumulatedAveraged}).

A specific instance of the prisoner's dilemma known as the \textit{donation game} is often used since it provides a simple description of an altruistic act: a cooperator ($A$) pays a cost of $c$ to donate $b$ to the opponent; defectors ($B$) pay no costs and provide no donations. In terms of \eq{prisonersDilemma}, the donation game satisfies $R=b-c$, $S=-c$, $T=b$, and $P=0$. This donation game is typically easier to analyze than the general matrix game \eq{prisonersDilemma}. However, once conditions for $\rho_A > \rho_B$ have been obtained for the donation game, these conditions can be generalized to an arbitrary game of the form \eq{prisonersDilemma} by means of the Structure Coefficient Theorem \citep{tarnita:JTB:2009,nowak:PTRSB:2009}.

Whereas these studies focus on pairwise interactions described by a fixed matrix game, we consider dilemmas of social goods, in which the goods produced by an individual are distributed to his or her neighbors in a manner that depends on the population structure. If $i$ and $j$ are neighbors, then, instead of their interaction being described by a symmetric payoff matrix
\begin{align}\label{eq:donationGame}
\bordermatrix{%
	& A & B \cr
	A &\ b-c ,\ b-c & \ -c ,\ b \cr
	B &\ b ,\ -c & \ 0 ,\ 0 \cr
}\ ,
\end{align}
in a social goods dilemma it is described by an asymmetric payoff matrix,
\begin{align}
M^{ij} &= 
\bordermatrix{%
	& A & B \cr
	A &\ B_{ji}-C_{ij} ,\ B_{ij}-C_{ji} & \ -C_{ij} ,\ B_{ij} \cr
	B &\ B_{ji} ,\ -C_{ji} & \ 0 ,\ 0 \cr
}\ . \label{eq:socialGoodsMatrix}
\end{align}
In other words, type $A$ (producer) at location $i$ pays $C_{ij}$ to donate $B_{ij}$ to $j$. Type $B$ (non-producer) at this same location pays nothing and gives nothing. Our theory applies to more general bimatrix games \citep{ohtsuki:JTB:2010,sekiguchi:DGA:2015,mcavoy:PLoSCB:2015,mcavoy:JRSI:2015}, although we are particularly focused on asymmetric games arising from the production of social goods.

Below, we give a brief summary of work on symmetric games in structured populations.

\subsubsection{Regular and homogeneous graphs}
A regular graph is one in which all individuals share the same number of neighbors, i.e. $w_{i}=d$ for every $i=1,\dots ,N$ and some fixed $d$. Using the pair-approximation method on regular graphs, \citet{ohtsuki:Nature:2006} showed that weak selection favors the evolution of cooperation if $b/c>d$. Moreover, the difference between accumulated and averaged payoffs on a regular graph amounts to a simple rescaling of both $b$ and $c$, so this critical threshold, namely $\left(b/c\right)^{\ast}=d$, is the same for both methods of obtaining a net payoff from many interactions. This result has also been refined to capture finite-population effects on weighted vertex-transitive graphs \citep{taylor:Nature:2007,allen:EMS:2014,debarre:NC:2014} and unweighted regular graphs \citep{chen:AAP:2013}.

\subsubsection{Heterogeneous graphs}
Although there are many simulation-based investigations of pairwise games on heterogeneous graphs \citealp{abramson:PRE:2001,santos:PRL:2005,gomez:PRL:2007,szabo:PR:2007,maciejewski:PLoSCB:2014}, an analytic solution for weak selection was not found until the work of \citet{allen:Nature:2017}. Conditions for a strategy to be favored under weak selection were found for DB and BD updating, using either averaged or accumulated payoffs, in terms of coalescence times. The properties of this solution were further explored in follow-up works by \citet{fotouhi:NHB:2018,fotouhi:JRSI:2019}.

The results of \citet{allen:Nature:2017}, in the case of the donation game with averaged payoffs, can be recovered from our results by setting $B_{ij}=bw_{ij}/w_{j}$ and $C_{ij}=cw_{ij}/w_{i}$ for each $i$ and $j$. \eq{tauRecurrence} is equivalent, under DB or BD updating, to the systems of equations for coalescence times used by \citet{allen:Nature:2017}.

\subsection{Public goods games in structured populations}
Numerous works \citep{santos:Nature:2008,rong:EPL:2009,szolnoki:PRE:2011,perc:JRSI:2013,pena:PLOSCB:2016,pena:JRSI:2016} have considered public goods games on graphs. In additional to classical public goods, one may also consider public goods that travel through the network by diffusion. This framework can be used to model the production of diffusible chemical goods in microbes \citep{legac:E:2010,julou:PNAS:2013,kummerli:EL:2014}. Conditions for the production of such goods on weighted vertex-transitive graphs were obtained by \citet{allen:EL:2013}.

Since our focus is on social goods donated to immediate neighbors with no self-donation, we do not give a complete analysis of public goods games. However, there is one study of public goods games, carried out by \citet{santos:Nature:2008}, that is particularly relevant to our model since it involves a numerical comparison of fixed- and proportional-cost goods in multiplayer interactions.

\citet{santos:Nature:2008} considered a model of public goods games in heterogeneous populations in which every individual is involved in multiple games. Each individual initiates an interaction with neighbors, and everyone involved benefits from the public good. An individual with $k$ neighbors therefore participates in $k+1$ interactions, one initiated by themselves and $k$ initiated by neighbors. In their model, a cooperator contributes either $c$ per game or $c$ in total (i.e. $c/\left(k+1\right)$ per game); non-producers do nothing. The total contribution is then multiplied by an enhancement factor, $r$, and divided among all individuals involved in the game.

If $\gamma_{i}$ denotes the per-game contribution when $i$ cooperates, this model satisfies
\begin{subequations}
	\begin{align}
	B_{ij} &= \sum_{k=1}^{N} \widetilde{w}_{ik} \widetilde{w}_{kj} \left( r\frac{\gamma_{i}}{\widetilde{w}_{k}} \right) ; \\
	C_{ij} &= \sum_{k=1}^{N} \widetilde{w}_{ik} \widetilde{w}_{kj} \left( \frac{\gamma_{i}}{\widetilde{w}_{k}} \right) ,
	\end{align}
\end{subequations}
where, again, $\widetilde{w}_{ii}=1$ for $i=1,\dots ,N$ and $\widetilde{w}_{ij}=w_{ij}$ if $i\neq j$. Therefore, this model can be seen as a special case of ours.

For both PC and DB updating, let $\tau_{ii}=0$ for $i=1,\dots ,N$ and $\tau_{ij}=1+\left(1/2\right)\sum_{k=1}^{N}p_{ik}\tau_{kj}+\left(1/2\right)\sum_{k=1}^{N}p_{jk}\tau_{ik}$ for $i\neq j$. Under PC updating, for example, we have
\begin{align}
\rho_{A}' > \rho_{B}' &\iff r\sum_{i,j=1}^{N} \pi_{i} p_{ij} \sum_{\ell =1}^{N} x_{i\ell} \left( C_{\ell i} - C_{\ell j} \right) > \sum_{i,j=1}^{N} \pi_{i} p_{ij} \sum_{\ell =1}^{N}\left( x_{ii}C_{i\ell} - x_{ij}C_{j\ell} \right) \nonumber \\
&\iff r\sum_{i,j=1}^{N} \pi_{i} p_{ij} \sum_{\ell =1}^{N} \tau_{i\ell} \left( C_{\ell j} - C_{\ell i} \right) > \sum_{i,j=1}^{N} \pi_{i} p_{ij} \sum_{\ell =1}^{N}\tau_{ij}C_{j\ell} . \label{eq:santosPC}
\end{align}
Under DB updating,
\begin{align}
\rho_{A}' > \rho_{B}' &\iff r\sum_{i,j=1}^{N} \pi_{i} p_{ij}^{\left(2\right)} \sum_{\ell =1}^{N} x_{i\ell} \left( C_{\ell i} - C_{\ell j} \right) > \sum_{i,j=1}^{N} \pi_{i} p_{ij}^{\left(2\right)} \sum_{\ell =1}^{N}\left( x_{ii}C_{i\ell} - x_{ij}C_{j\ell} \right) \nonumber \\
&\iff r\sum_{i,j=1}^{N} \pi_{i} p_{ij}^{\left(2\right)} \sum_{\ell =1}^{N} \tau_{i\ell} \left( C_{\ell j} - C_{\ell i} \right) > \sum_{i,j=1}^{N} \pi_{i} p_{ij}^{\left(2\right)} \sum_{\ell =1}^{N} \tau_{ij}C_{j\ell} . \label{eq:santosDB}
\end{align}

Finally, we note that there is nearly always wealth inequality in heterogeneous populations, even for pp-goods with accumulated payoffs. \citet{santos:Nature:2008} observed that the nature of how players participate in public goods games in heterogeneous populations can amplify inequality. Our notion of ``prosocial inequality'' is somewhat different, however. Whereas inequality can be present even when selection improves the payoffs of all players in the population, we show that selection can actually decrease the payoff of some while improving the payoff of others. Thus, prosocial inequality is not a measure of relative inequality within the all-$A$ state; it is a measure of inequality in the all-$A$ state as compared to the all-$B$ state.

\subsection{Deterministic versus stochastic payoffs}
One final way in which our model and results differ from those of prior studies is that we allow for stochastic payoffs in addition to deterministic payoffs. In fact, we show in \S\ref{subsec:stochDet} that a model with stochastic payoffs can be replaced by one with deterministic payoffs under weak selection. This result is particularly relevant for goods with concentrated, stochastic benefits (\fig{concentrated}), in which case the critical threshold for producers to evolve is identical to that of ff-goods.

\end{document}